\documentclass[twoside,leqno]{article}

\usepackage[letterpaper]{geometry}

\usepackage{ltexpprt}
\usepackage{nicefrac,amstext}

\usepackage[utf8]{inputenc}
\usepackage{enumitem}
\usepackage[algo2e, ruled,vlined,linesnumbered]{algorithm2e}
\usepackage{algpseudocode}
\usepackage{amsmath,amsfonts,amssymb,dsfont}
\usepackage{color}
\usepackage{graphicx}
\usepackage{tikz}
\usepackage{caption}
\usepackage{subcaption}

\usepackage{xcolor}
\usepackage{nameref}
\definecolor{ForestGreen}{rgb}{0.1333,0.5451,0.1333}
\definecolor{DarkRed}{rgb}{0.65,0,0}
\definecolor{Red}{rgb}{1,0,0}
\usepackage{hyperref}
\usepackage{cleveref}

\newcommand{\tO}{\tilde{O}}

\newcommand{\Sbar}{\overline{S}}
\newcommand{\dem}{\mathrm{dem}}
\newcommand{\rdem}{\mathrm{rdem}}

\newcommand{\super}{\mu}

\newcommand{\alert}[1]{{\textcolor{red}{#1}}}

\newcommand{\eat}[1]{}

\newcommand{\polylog}{{\rm polylog}}

\Crefname{algorithm}{Algorithm}{Algorithms}
\Crefname{Definition}{Definition}{Definitions}
\Crefname{fact}{Fact}{Facts}
\Crefname{@theorem}{Theorem}{Theorems}

\begin{document}

\newcommand\relatedversion{}
\renewcommand\relatedversion{\thanks{The full version of the paper can be accessed at \protect\url{https://arxiv.org/abs/xxx}}} 

\title{Steiner Connectivity Augmentation and Splitting-off in Poly-logarithmic Maximum Flows\thanks{Ruoxu Cen and Debmalya Panigrahi were supported in part by NSF grants CCF-1750140 (CAREER Award) and CCF-1955703.}}
\author{Ruoxu Cen\thanks{Department of Computer Science, Duke University. Email: {\tt ruoxu.cen@duke.edu}} \and William He\thanks{Duke University. Email: {\tt william.he@duke.edu}} \and Jason Li\thanks{Simons Institute for the Theory of Computing, UC Berkeley. Email: {\tt jmli@alumni.cmu.edu}} \and Debmalya Panigrahi\thanks{Department of Computer Science, Duke University. Email: {\tt debmalya@cs.duke.edu}}}
\date{}

\maketitle


\begin{abstract} \small\baselineskip=9pt
    We give an almost-linear time algorithm for the Steiner connectivity augmentation problem: given an undirected graph, find a smallest (or minimum weight) set of edges whose addition makes a given set of terminals $\tau$-connected (for any given $\tau > 0$). The running time of our algorithm is dominated by polylogarithmic calls to {\em any} maximum flow subroutine; using the recent almost-linear time maximum flow algorithm (Chen et al., FOCS 2022), we get an almost-linear running time for our algorithm as well. This is tight up to the polylogarithmic factor even for just two terminals. Prior to our work, an almost-linear (in fact, near-linear) running time was known only for the special case of global connectivity augmentation, i.e., when all vertices are terminals (Cen et al., STOC 2022). 
    
    We also extend our algorithm to the closely related Steiner splitting-off problem, where the edges incident on a vertex have to be {\em split-off} while maintaining the (Steiner) connectivity of a given set of terminals. Prior to our work, a nearly-linear time algorithm was known only for the special case of global connectivity (Cen et al., STOC 2022). The only known generalization beyond global connectivity was to preserve all pairwise connectivities using a much slower algorithm that makes $n$ calls to an all-pairs maximum flow (or Gomory-Hu tree) subroutine (Lau and Yung, SICOMP 2013), as against $\polylog(n)$ calls to a (single-pair) maximum flow subroutine in this work.
\end{abstract}


\section{Introduction}
\label{sec:introduction}
Given an undirected graph $G = (V, E)$ with a set of terminals $T\subseteq V$ and a target connectivity $\tau > 0$, the {\em Steiner connectivity augmentation} problem asks for an edge set of minimum cardinality whose addition to $G$ makes $T$ $\tau$-connected, where a set of vertices $T$ is $\tau$-connected if for every pair of vertices $s, t\in T$, the $(s,t)$-connectivity is $\ge\tau$. In the weighted version of the problem, edges in $G$ can have nonnegative integer weights and the goal is to minimize the total weight of edges added to $G$ to make $T$ $\tau$-connected. All edge weights and the target connectivity are assumed to be polynomially bounded in the size of the graph. We distinguish between adding edges of nonnegative {\em weight} and nonnegative {\em cost}. An edge of {\em weight} $w$ is equivalent to $w$ parallel edges; adding it increases the value of every cut that it belongs to by $w$ and adds $w$ to the objective that we seek to minimize. In contrast, if edges have nonuniform {\em costs} and the goal is to minimize the total cost of the added edges, then an edge only increases the value of a cut it belongs to by 1, but incurs a given nonnegative cost in the objective. This version of the problem with nonuniform costs is known to be NP-hard. We do not consider this version of the problem in this paper. 

In this paper, we give an algorithm for Steiner connectivity augmentation in weighted graphs that runs in $\polylog(n)$ times $F(m, n)$ time, where $F(m, n)$ is the running time of any maximum flow algorithm on a graph with $m$ edges and $n$ vertices. This is optimal up to the $\polylog(n)$ term since a Steiner connectivity augmentation algorithm can be used to determine $(s,t)$ connectivity by setting $T = \{s, t\}$ and (binary) searching for the maximum $\tau$ that admits an empty solution. If we use the recent breakthrough $m^{1+o(1)}$-time maximum flow algorithm of Chen  et al.~\cite{ChenKLPPS22}, then the running time of our algorithm becomes $m^{1+o(1)}$.

The special case of this problem where $T = V$, i.e., the goal is to increase the {\em global} connectivity of the graph to a target $\tau$, has been extensively studied in the literature. Watanabe and Nakamura~\cite{WatanabeN87} were the first to give a polynomial-time algorithm for this problem in unweighted graphs, while the first strongly polynomial algorithm for weighted graphs was obtained by Frank~\cite{Frank92}. Since then, several algorithms~\cite{CaiS89,NaorGM97,Gabow16,Gabow94,NagamochiI97,BenczurK00,CenLP22a} have progressively improved the running time for this problem until it was recently shown to be solvable in $\tO(m)$ time~\cite{CenLP22b}.\footnote{$\tO(\cdot)$ hides polylogarithmic factors.} Another natural special case is when $T = \{s, t\}$. In this case, the optimal solution is to add an edge of weight $\tau - \lambda(s, t)$ between $s$ and $t$, where $\lambda(s, t)$ is the $(s,t)$-connectivity. Thus, for $|T| = 2$, the problem is solvable in $F(m,n)$ time. But for general $T$ not equal to $V$ or $\{s,t\}$, it is not a priori obvious if the problem is polynomial-time tractable. In fact, by changing the problem slightly where only a given set of edges can be added, the problem becomes NP-hard. Therefore, it is somewhat surprising that the Steiner connectivity augmentation problem can be solved for general $T$ not just in polynomial-time, but with only a $\polylog(n)$ overhead over the trivial case of $|T| = 2$.

A related problem to connectivity augmentation is that of {\em splitting-off}. In this problem, we are given a graph $G = (V, E)$ and a vertex $x\in V$ that has to be split-off. A splitting-off operation at $x$ pairs two edges $(u, x)$ and $(x, v)$ and replaces them by the shortcut edge $(u, v)$ in the graph. The goal is to split-off all edges incident on $x$ 
while preserving some connectivity property of the remaining graph. 
Lov\'asz~\cite{Lovasz79} introduced the splitting-off operation and showed that any vertex can be split-off while preserving the {\em global} connectivity of the remaining graph. Since then, many splitting-off theorems have been shown (e.g.,~\cite{Mader78,Szigeti08}) that preserve a rich set of connectivity properties.
The splitting-off operation has been an influential inductive tool for both extremal graph theory and graph algorithms with applications in network design~\cite{ChekuriS2009, ChanFLY11}, tree packing~\cite{Lau2007approximate, BhalgatHKP08, ChekuriK14}, tree decomposition~\cite{Merker17}, graph orientation~\cite{FrankK02, FrankK03, BernathIKKS08}, and metric TSP~\cite{Gao18}.

Preserving a rich set of connectivity properties, however, comes at the cost of slower algorithmic implementation. While an $\tO(m)$-time implementation of Lov\'asz's theorem has recently been shown~\cite{CenLP22b}, the best implementations of stronger splitting-off theorems are much slower. In this paper, we extend our Steiner connectivity augmentation algorithm to obtain an algorithm for Steiner splitting-off, namely split-off a vertex $x$ while preserving the Steiner connectivity of any set of terminals $T\subseteq V$. (The Steiner connectivity of a set of terminals $T$ is the minimum $(s,t)$ connectivity for any terminals $s, t\in T$.) Our Steiner splitting-off algorithm has the same running time as the Steiner connectivity augmentation algorithm, namely $\polylog(n)$ calls to any maximum flow algorithm. To the best of our knowledge, prior to our work, the only known splitting-off algorithms that preserve a more general property than global connectivity actually preserve all pairwise connectivities at a much slower running time. 
The current best algorithm for this latter problem is due to Lau and Yung \cite{LauY13}, whose running time is dominated by $n$ calls to a Gomory-Hu tree subroutine. The paper describes the running time of the algorithm as is $\tO(m+n\tau^3)$ for unweighted graphs where $\tau$ is the maximum pairwise connectivity. This is based on the best known (partial) Gomory-Hu tree algorithm at the time of the Lau-Yung paper~\cite{HariharanKP07, BhalgatHKP08}. Using the more recent $\tO(n^2)$-time Gomory-Hu tree algorithm for weighted graphs or the $m^{1+o(1)}$-time Gomory-Hu tree algorithm for unweighted graphs~\cite{AbboudKLPST22} improves the running time of the Lau-Yung algorithm to $\tO(n^3)$ in weighted graphs and $mn^{1+o(1)}$ in unweighted graphs. 
In this paper, we improve the running time for both weighted and unweighted graphs further from $\tO(n^3)$ and $mn^{1+o(1)}$ respectively to $m^{1+o(1)}$, but only for the Steiner splitting-off problem.


\begin{theorem}\label{thm:main-intro}
    There are $m^{1+o(1)}$-time algorithms for the Steiner connectivity augmentation and Steiner splitting-off problems on undirected graphs with $m$ edges. The algorithms are randomized and succeed with high probability.\footnote{Following standard convention, \emph{with high probability} is used to mean with probability at least $1-n^{-c}$ for an arbitrarily large constant $c>0$, in this paper.}
\end{theorem}

\subsection{Our Techniques}
In this overview, we will only consider unweighted graphs; the same ideas also extend to weighted graphs. We will use $d(X)$ to denote the {\em degree} of any set of vertices $X$, i.e., the number of edges with exactly one endpoint in $X$.

\paragraph{Supreme Sets.}
A key role in (global) connectivity augmentation (e.g., \cite{WatanabeN87,NaorGM97,Gabow16,Benczur94,BenczurK00,CenLP22a,CenLP22b}) is played by {\em extreme sets} that are defined as follows: a set of vertices $X$ is extreme if for every proper subset $Y\subset X$, we have $d(Y) > d(X)$. Intuitively, the extreme sets represent the ``bottlenecks'' in connectivity augmentation, and we want to ensure that satisfying their deficit also satisfies that of all other sets. To this end, the algorithms first perform a step called {\em external augmentation}. In this step, an external vertex $x$ is added to the graph and the algorithm finds a minimum set of edges with one endpoint at $x$ that satisfies the deficits for all the extreme sets. It is not difficult to see that these edges also satisfy the connectivity requirements for all other sets of vertices. Now, the second step of the algorithm replaces these edges by a new set of edges only on the actual vertices (i.e., not incident to $x$) in a careful manner so as to preserve the augmented connectivity. The prior algorithms differ in how these two steps are implemented.

For Steiner connectivity augmentation, we can similarly define extreme sets, although we now need to restrict to sets that contain at least one terminal vertex (since sets not containing terminals do not have any connectivity requirement). But, this creates more extreme sets than in the global case. Consider a set $X$ containing terminals such that all its proper subsets $Y$ satisfying $d(Y) \le d(X)$ do not contain any terminals. In the global case, $X$ would not be an extreme set but now we need to declare $X$ as extreme since its requirement will not be met by any of its subsets. In fact, the number of extreme sets can now be exponentially large (see Fig.~\ref{fig:exponential} in Appendix~\ref{sec:exponential} for an example). In contrast, one of the key properties of extreme sets in global connectivity augmentation is that they are laminar in structure (and hence linear in number), which allows the design of efficient algorithms that find all the extreme sets.

Since we cannot find all extreme sets explicitly, we first partition them into equivalence classes based on their projection onto the terminals. Now, consider all the extreme sets with the same terminal projection. We show that the extreme sets are closed under union (but not intersection!) and therefore, we can define a unique maximum extreme set in this class. We call these maximum extreme sets the {\em supreme sets} and show that they form a laminar family. This also means that the terminal projections themselves also form a laminar family, and our first goal is to find this tree of terminal projections. The key property that we prove is the following:
\begin{lemma}
    A supreme set with terminal projection $R$ is also the {\em earliest} minimum-value cut that bipartitions the terminals $T$ as $R$ and $T\setminus R$.
\end{lemma}    
This connection allows us to design a divide and conquer algorithm that uses polylogarithmic calls to a maximum flow subroutine and returns a laminar family of vertex sets that includes all the supreme sets. Using a sequence of postprocessing steps, we eliminate the spurious sets from this laminar family and extract the tree of terminal projections. Since extreme sets (for global connectivity) are an important concept in their own right that has been studied extensively in the literature, we hope this structural and algorithmic exploration of extreme sets for Steiner connectivity and the definition of supreme sets is of independent value. We present this in \Cref{sec:extreme-sets}.

\bigskip

The augmentation algorithm (\Cref{sec:augmentation}) has three parts. The first part performs external augmentation on the tree of terminal projections (\Cref{sec:external}). Note that this means that all the edges we add are incident to the terminals. (This is the key property that makes the Steiner connectivity augmentation problem simpler than the Steiner splitting-off problem that we will consider later; this property will not hold in the latter problem.) We now need to replace the edges added by external augmentation with edges between terminals. We view this as a splitting-off process (\Cref{sec:splitoffview}). The splitting-off proceeds in two steps. The first step (\Cref{sec:chains}) achieves Steiner connectivity of $\tau-1$ by adapting an iterative framework of adding {\em augmentation chains} while respecting the vertex degrees (on the terminals) given by external augmentation~\cite{BenczurK00}. To implement this framework efficiently, we adapt the use of {\em lazy} data structures from \cite{CenLP22a}, and show that the overall implementation of this step takes only near-linear time. 

\paragraph{Augmenting Connectivity by Random Matchings.}
Finally, we need to augment Steiner connectivity by one, from $\tau-1$ to $\tau$ (\Cref{sec:augment-1}).
  For global connectivity, augmenting connectivity from $\tau-1$ to $\tau$ uses the cactus data structure for all minimum cuts which can be computed in $\tO(m)$ time~\cite{KargerP09}. For minimum Steiner cuts, however, there is no data structure that can be computed as efficiently (the only known data structure is called a {\em connectivity carcass}~\cite{DinitzV94} whose computation is highly inefficient). To remedy this, we revisit the algorithm for augmenting global connectivity from $\tau-1$ to $\tau$. Consider the simple case of augmenting connectivity of a cycle on $n$ vertices from $2$ to $3$. Now, suppose we iteratively add edges that form a matching with the following property: when adding an edge $(u, v)$, ensure that they are non-adjacent on the cycle, and after adding the edge, shortcut the two neighbors of $u$ and the two neighbors of $v$. We show this in Figure~\ref{fig:augment-cycle-1}. It is not difficult to show that this gives an optimal augmentation for the cycle. 
 
 Crucially, this algorithm can be implemented even without knowing the structure of the cycle, since a random pair of vertices is unlikely to be adjacent (except at the end). Indeed, even if we add a random matching on half the vertices instead of single edge to the cycle, it will not contain any adjacent pair with constant probability. Thus, we can implement this algorithm using $O(\log n)$ random matchings, as long as we can certify correctness of a single matching (i.e., whether it contains an adjacent pair or not). We use this intuition to design an algorithm for augmenting the (Steiner) connectivity by one. The algorithm first partitions the set of terminals based on the minimum Steiner cuts (of value $\tau-1$), and then iteratively adds a sequence of random matchings between the terminal sets in this partition. We also give an algorithm that can certify whether a random matching can be extended to an optimal solution, and if not, we simply discard the matching and repeat the step again. Given the interest in the literature in the special case of augmenting the connectivity by one and the simplicity of our algorithm, we believe this algorithm is of independent interest.
 
 
 \paragraph{Steiner Splitting-Off.} 
 The Steiner splitting off algorithm is presented in \Cref{sec:splitting}. 
 This problem can be viewed as Steiner connectivity augmentation, but where the degree of every vertex in terms of the added edges is constrained by their original degree to the vertex being split-off. This creates additional complications because the connectivity augmentation algorithm that we just sketched always adds edges to terminals whereas the splitting-off algorithm might require us to add edges to nonterminals as well. Suppose we are augmenting the connectivity of a set of extreme sets that share the same terminal projection. Earlier, we could add edges to the terminals in the projection and this benefits all the extreme sets; so we did not need the explicit extreme sets structure (and it is too expensive to compute it in any case). Now, we need to add edges incident to nonterminals, but these edges will benefit some but not all of the extreme sets with the same terminal projection. To solve the problem of adding edges to nonterminals optimally, we encode this problem as a sequence of maximum flow problems on suitably defined auxiliary graphs for each heavy path in a {\em heavy-light decomposition} of the supreme sets forest. The external edges are tied to the heavy paths, and this allows us to track the contribution of edges incident to nonterminals. The details of this encoding and the consequent algorithm are given in \Cref{sec:deg-external}.

 Another difficulty arises in the maintenance of the (lazy) data structures for tracking the cut values of the supreme sets while adding augmentation chains. Our data structures do not explicitly track the supreme sets, but rather their intersections with the terminals. This was not a problem in augmentation since the edges we add are always between terminals. But, now that we are forced to add edges incident to nonterminals, we do not immediately know whether a supreme set's cut value was increased by the addition of an edge. We bypass this problem by showing that there is a small subset of {\em representative} supreme sets such that it is sufficient for our data structures to only maintain the cut values of the representatives. The details appear in \Cref{sec:deg-splitting}. Finally, we adapt our algorithm for augmenting Steiner connectivity by one using random matchings to the degree-constrained case (\Cref{sec:deg-augment-1}).

\subsection{Related Work}
Connectivity augmentation problems are well-studied in many settings. Since we already mentioned the line of work for global connectivity augmentation earlier, we describe some other connectivity augmentation problems here. Jackson and Jord\'an \cite{JacksonJ05} proposed a polynomial time algorithm for vertex connectivity augmentation in undirected graphs when the connectivity target is a constant. This result was generalized to directed graphs by Frank and Jord\'an \cite{FrankJ99} who gave a polynomial-time algorithm, and was further improved by V\'egh and Bencz\'ur \cite{VeghB08} to a strongly polynomial-time algorithm. Another variant of connectivity augmentation considered in the literature is that of inclusion-wise monotone connectivity targets on vertex sets (Ishii and Hagiwara \cite{IshiiH06}, Ishii~\cite{Ishii09}). 
They showed that this problem is NP-hard in general, but gave polynomial-time algorithms when the targets are strictly greater than $1$.

A popular line of work in connectivity augmentation involves restricting the set of edges to a given set, or more generally, allowing the edges to have nonuniform costs. This makes the problem NP-hard to even augment (global) connectivity by one~\cite{FredericksonJ81}. Nevertheless, Marx and V\'egh \cite{MarxV15} showed that with edge costs, it is fixed parameter tractable (FPT) to increment connectivity with the number of new edges as the parameter. In a similar vein, Klinkby et al.~\cite{KlinkbyMS21} showed an FPT result for the problem of augmenting a directed graph to make it strongly connected even with edge costs. There is also a long line of work in approximation algorithms for adding a minimum number of edges to a graph to increase the connectivity to $2$. Jain's iterative rounding framework~\cite{jain2001factor} yields a $2$-approximation for this problem. The approximation factor of $2$ was subsequently improved in a series of works for the case where the input graph is a tree \cite{grandoni2018improved,byrka2020breaching,cecchetto2021bridging,traub2022better}. Recently, Grandoni et al. \cite{grandoni2022breaching} gave an algorithm breaking the approximation factor of $2$ even when the input graph is a forest (which is equivalent to general input graphs for this problem).

\section{Preliminaries}
\label{sec:prelim}
For any nonempty vertex set $X\subsetneqq V$, the cut $(X, V\setminus X)$ (or cut $X$ in short) is the set of edges connecting $X$ and $V\setminus X$.
Denote $d(X)$ to be the cut value of $X$, defined as the sum of weights of edges in cut $X$. In particular, $d(\emptyset)=d(V)=0$.

For any $s, t\in V$, $s\ne t$, an $s$-$t$ cut is a cut $(S, T)$ separating $s$ and $t$. An $s$-$t$ min cut is an $s$-$t$ cut with minimum value among all $s$-$t$ cuts. Denote $\lambda(s,t)$ to be the value of the $s$-$t$ min cut. Among the (possibly many) $s$-$t$ min cuts, the one whose $s$-side is minimal (as a vertex set) is called the earliest $s$-$t$ min cut, and the one whose $t$-side is minimal (or equivalently, maximal in $s$-side) is called the latest $s$-$t$ min cut. (They are unique because the intersection and union of two $s$-$t$ min cuts are also $s$-$t$ min cuts.)

For disjoint vertex sets $X$ and $Y$, similarly define an $X$-$Y$ min cut to be the minimum cut separating $X$ and $Y$, and $\lambda(X, Y)$ be the value of $X$-$Y$ min cut. The earliest (resp.\ latest) $X$-$Y$ min cut is the $X$-$Y$ min cut whose $X$-side (resp.\ $Y$-side) is minimal. These notations about $X$-$Y$ cuts can also be understood as the corresponding notations of $s$-$t$ cuts by contracting $X$ and $Y$ to be $s$ and $t$.

We now define two key terms that we use throughout the paper: crossing and submodularity.
\begin{Definition}[Crossing sets]
We say that vertex sets $X$ and $Y$ \emph{cross} if $X\cap Y,\ X\setminus Y$ and $Y\setminus X$ are all non-empty.

A family of sets is \emph{laminar} if no two sets in the family cross. Such a family can be represented by a collection of disjoint rooted trees, where each node corresponds to a set in the family, so that (the set corresponding to) an ancestor node of a tree contains (the set corresponding to) a descendant node, and two nodes that are incomparable in a single tree (or are in different trees) are disjoint sets. Formally, for each set $R$ in the family, the node of $R$ is a root iff it is maximal in the family, and $R$ is a child of $X$ iff $R\subseteq X$ and there is no other set $Y$ in the family such that $R\subseteq Y\subseteq X$.
\end{Definition}

\begin{fact}[Submodularity and posi-modularity of cuts]
\label{fact:submodularity}
For any vertex sets $X$ and $Y$,
\begin{equation}\label{eq:submodularity-1}
    d(X\cap Y)+d(X\cup Y)\le d(X)+d(Y)
\end{equation}
\begin{equation}\label{eq:submodularity-2}
    d(X\setminus Y)+d(Y\setminus X)\le d(X)+d(Y)
\end{equation}
The following immediate corollaries of \emph{(\ref{eq:submodularity-1})} will be useful.
\begin{equation}\label{eq:submodularity-3}
    d(X\cap Y)\ge d(X) \implies d(X\cup Y)\le d(Y)
\end{equation}
\begin{equation}\label{eq:submodularity-4}
    d(X\cap Y)> d(X) \implies d(X\cup Y)< d(Y)
\end{equation}
\end{fact}

\subsection{Proof of Tractability}
Before discussing our efficient algorithm, we first observe that polynomial-time tractability of the Steiner connectivity augmentation and splitting-off problems follow easily from previous works. Indeed, these problems are special cases of pairwise connectivity augmentation and splitting off problems, for which we give short proofs of polynomial tractability based on previous work below. 

For the splitting off problem preserving pairwise connectivities, the running time of the current best algorithm by Lau and Yung \cite{LauY13} is dominated by $n$ calls to a Gomory-Hu tree subroutine. Although the result stated in \Cref{lem:LauY} is for unweighted multigraphs, the same algorithm works for integer-weighted graphs if we conceptually regard an edge of weight $w$ as $w$ parallel edges. 

\begin{lemma}[\cite{LauY13}]\label{lem:LauY}
There exists an algorithm that, given an unweighted undirected graph with a special vertex $x$ to be split off where $d(x)$ is even and there is no cut edge incident to $x$, solves splitting off preserving all-pairs connectivities of $V\setminus \{x\}$ in $O(n)$ attempts. Each attempt picks a pair of vertices $(u, v)$ and splits $(x, u)$ and $(x, v)$ for some amount, and the amount can be calculated by a Gomory-Hu tree subroutine.
\end{lemma}

Using the $\tO(n^2)$ time Gomory-Hu tree algorithm \cite{AbboudKLPST22}, we can solve Steiner splitting-off as a special case of splitting-off preserving all pairwise connectivities in $\tO(n^3)$ time.

For the Steiner connectivity augmentation problem, we first perform external augmentation, and then split-off the external augmentation solution to get a solution to the original problem. This transformation is discussed in \Cref{sec:splitoffview}. Here, we show in \Cref{lem:simple-augment-time} that the external augmentation solution can be computed in $\tO(n^3)$ time using prior work.



We introduce a greedy external augmentation algorithm proposed by Frank \cite{Frank92}. It works for the all-pairs connectivity setting, where every pair of vertices $u, v\in V\setminus\{x\}$ has a $u$-$v$ connectivity requirement. In external augmentation, we add a new vertex $x$.
First, we add sufficiently many edges (called external edges) from every vertex to the external vertex $x$, so that the connectivity requirement is satisfied. Then, we repeatedly remove any external edge as long as the connectivity requirement is satisfied, until no external edge can be removed. The remaining external edges form a feasible external augmentation solution. 

\begin{lemma}[Lemma 5.6 of \cite{Frank92}]\label{lem:frank-external}
If the optimal value of an all-pairs connectivity augmentation instance is $\gamma$, and no vertex pair has connectivity requirement 1, then the greedy external augmentation algorithm outputs a feasible external augmentation solution with total weight at most $2\gamma$.
\end{lemma}

\begin{lemma}\label{lem:greedy-ext-time}
The greedy external augmentation algorithm can be implemented in $\tO(n^3)$ time.
\end{lemma}
\begin{proof}
In the removal stage, we process one vertex at a time. Each iteration picks a vertex $v$, and removes the maximum number of $(v, x)$ edges such that no connectivity requirement is violated. After vertex $v$ has been processed, none of the $(v, x)$ edges can be removed without violating a connectivity requirement in the future. This property holds irrespective of what other edges are removed in the future, because removing other edges cannot increase cut values. Therefore, it suffices to process every vertex $v$ exactly once, and the algorithm terminates in $n$ iterations.

In a single iteration, for a vertex $v$, we need to decide the maximum number of removable $(v, x)$ edges. To this end, we compute all-pairs min-cuts in the graph after removing all $(v, x)$ edges. Suppose there are $w$ parallel $(v, x)$ edges. Let $\Delta$ be the maximum violation of a connectivity requirement after removing these $w$ edges. That is, $\Delta$ is the maximum difference between the $(s, t)$-connectivity requirement and the value of the $(s, t)$ min-cut after removing the $(v, x)$ edges, among all vertex pairs $s, t$. Since all cut values satisfy the connectivity requirements before removing the $(v, x)$ edges, the $(v, x)$ edges must cross all violated cuts. Therefore, by adding back $\Delta$ parallel $(v, x)$ edges, all violations can be resolved. In other words, removing $w-\Delta$ edges is feasible. Moreover, removing more than $w-\Delta$ edges is infeasible because adding back less than $\Delta$ edges cannot resolve the maximum violation of $\Delta$. Therefore the maximum number of removable edges is $w-\Delta$.
%
 
The algorithm runs $n$ iterations. Each iteration calls an all-pairs min-cuts subroutine, which takes $\tO(n^2)$ time \cite{AbboudKLPST22}. The total running time is $\tO(n^3)$.
\end{proof}


\begin{lemma}\label{lem:simple-augment-time}
The Steiner connectivity augmentation problem can be solved in $\tO(n^3)$ time.
\end{lemma}
\begin{proof}
When $\tau=1$ the optimal solution is simply a tree over the connected components containing terminals. The main case is $\tau\ge 2$. Assume the optimal solution has total weight $\gamma$.

Run the greedy external augmentation algorithm to get an external augmentation solution $F^{ext}$. (To convert Steiner connectivity augmentation into all-pairs connectivity augmentation, set the connectivity requirement for all pairs of terminals to $\tau$ and that for all other pairs to $0$.)  By \Cref{lem:frank-external}, the total weight of $F^{ext}$ is at most $2\gamma$. If the total weight is odd, we increase the weight of an arbitrary edge in the solution by 1 so that the total weight is even and still at most $2\gamma$. Adding these edges increases the Steiner connectivity to at least $\tau$.

Then, we split-off the external vertex using \Cref{lem:LauY} to get an edge set $F$ with total weight at most $\gamma$. (There is no cut edge incident to the external vertex because the Steiner connectivity is at least 2.) Splitting-off preserves pairwise connectivity; therefore, the Steiner connectivity is still at least $\tau$. Thus, $F$ is a feasible solution to the Steiner connectivity augmentation problem, and its total weight is no more than the optimal value $\gamma$. Hence, $F$ it is an optimal solution.

The running time is $\tO(n^3)$ in both steps by \Cref{lem:LauY,lem:greedy-ext-time}. Therefore, the total running time is $\tO(n^3)$.
\end{proof}

\section{Extreme Sets for Steiner Connectivity}
\label{sec:extreme-sets}
For a graph $G$ on vertex set $V$ and a set $T\subseteq V$ of terminals, we define (one side of) a {\it Steiner cut} as a vertex set $X$ such that $X\cap T\ne\emptyset $ and $T\not\subseteq X$. The concept of {\it extreme sets} in global connectivity setting can be generalized in the following natural way.

\begin{Definition}[Extreme sets]\label{def:extreme}
A Steiner cut $X$ is \emph{extreme} if and only if for any Steiner cut $Y$ that is also a proper subset of $X$, $d(Y)>d(X)$.
\end{Definition}

If a Steiner cut $X$ is not extreme, by definition there exists a Steiner cut $Z\subsetneqq X$ with $d(Z)\le d(X)$. Such a set $Z$ is called a \emph{violator}. 

Let $\mathcal{X}$ be the family of all extreme sets in $G$.
The following facts are immediate from the definition.

\begin{fact}
\label{fact:subset-of-extreme}
If $X$ is an extreme set and $Y\subseteq X$ is a Steiner cut, then $d(Y)\ge d(X)$, and equality only holds when $Y=X$.
\end{fact}

\begin{fact}
\label{fact:extreme-violator}
If a Steiner cut $X$ is not extreme, then there exists a violator of $X$ that is extreme.
\end{fact}
\begin{proof}
Let $Z$ be the violator of $X$ of minimum size. By definition, $d(Z)\le d(X)$. We prove that $Z$ is extreme, which implies the statement.

Let $W$ be any Steiner cut that is a proper subset of $Z$. Because $|W|<|Z|$ and $Z$ is the minimum-size  violator of $X$, $W$ is not a violator of $X$. That is, $d(W)>d(X)\ge d(Z)$. Therefore $Z$ is extreme by \Cref{def:extreme}.
\end{proof}

Finally, we define the projection of a Steiner cut to be the set of terminals in it.
\begin{Definition}
Define $\rho$ to be the \emph{projection} from vertex sets to terminal sets. For any $X\subseteq V$, $\rho(X)=X\cap T$.
\end{Definition}
\begin{Remark}
$X$ is a Steiner cut if and only if $\rho(X)\ne \emptyset$ and $\rho(X)\ne T$.
\end{Remark}

\subsection{Structure of Extreme Sets and Supreme Sets}
In the Steiner setting, the structure of extreme sets is much weaker than in the global connectivity setting. Most importantly, we are no longer guaranteed that extreme sets never cross. Fortunately, the extreme sets are still closed under set union, which we establish in the following chain of lemmas.

\begin{lemma}
\label{lem: extreme-difference-terminal}
If extreme sets $X$ and $Y$ cross, then at least one of $X\setminus Y$ or $Y\setminus X$ contains no terminal.
\end{lemma}
\begin{proof}
Assume for contradiction that both $X\setminus Y$ and $Y\setminus X$ contain terminals. Then they are Steiner cuts. Because $X$ and $Y$ cross, $X\setminus Y$ is a proper subset of $X$, and $Y\setminus X$ is a proper subset of $Y$. Then $d(X\setminus Y)>d(X)$ and $d(Y\setminus X)>d(Y)$ since $X$ and $Y$ are extreme (\Cref{def:extreme}). Adding these two inequalities gives \[d(X\setminus Y)+d(Y\setminus X)>d(X)+d(Y)\] which contradicts (\ref{eq:submodularity-2}) of \Cref{fact:submodularity}.
\end{proof}

\begin{lemma}
\label{lem: extreme-intersect-terminal}
For any extreme sets $X$ and $Y$, if $X\cap Y$ is nonempty, then $X\cap Y$ is a Steiner cut.
\end{lemma}

\begin{proof}
When $X\subseteq Y$ (resp.\ $Y\subseteq X$), $X\cap Y$ is a Steiner cut because it is $X$ (resp.\ $Y$). The remaining case is that both $X\setminus Y$ and $Y\setminus X$ are nonempty. The statement also assumes $X\cap Y$ is nonempty, so $X$ and $Y$ cross. By \Cref{lem: extreme-difference-terminal}, one of $X\setminus Y$ or $Y\setminus X$ contains no terminal. Without loss of generality, assume $\rho(X\setminus Y)=\emptyset$.  Then 
\begin{equation*}
\rho(X\cap Y)=\rho(X\setminus(X\setminus Y))=\rho(X) \setminus \rho(X\setminus Y)=\rho(X)
\end{equation*}
Because $X$ is a Steiner cut and $\rho(X\cap Y)=\rho(X)$, $X\cap Y$ is also a Steiner cut.
\end{proof}

\begin{lemma}
\label{lem:extreme-union}
If two extreme sets $X$ and $Y$ cross and $X\cup Y$ is a Steiner cut, then $X\cup Y$ is extreme.
\end{lemma}
\begin{proof}
Assume for contradiction that there exist crossing extreme sets $X$ and $Y$ such that $X\cup Y$ is a Steiner cut but not extreme.
By \Cref{fact:extreme-violator}, there exists an extreme violator of $X\cup Y$. Let $Z$ be the one with minimum cut value among all extreme violators of $X\cup Y$. \begin{equation}\label{eq:extreme-union-1}  d(Z)\le d(X\cup Y) \end{equation}

Because $X, Y$ cross, $X\cap Y\ne\emptyset$. By \Cref{lem: extreme-intersect-terminal}, $X\cap Y$ is a Steiner cut. Again because $X, Y$ cross, both $X\setminus Y$ and $Y\setminus X$ are non-empty, so $X\cap Y$ is a proper subset of $X$ and $Y$. Since $X$ and $Y$ are extreme (\Cref{def:extreme}), 
\[d(X\cap Y)>\max\{d(X),d(Y)\}\]
Applying (\ref{eq:submodularity-4}) of \Cref{fact:submodularity} on $X$ and $Y$, we have 
\begin{equation}\label{eq:extreme-union-2}
    d(X\cup Y)<\min\{d(X),d(Y)\}
\end{equation}
(\ref{eq:extreme-union-1}) and (\ref{eq:extreme-union-2}) imply $d(Z)<\min\{d(X), d(Y)\}$, which means $Z$ cannot be a subset of $X$ or $Y$ by \Cref{fact:subset-of-extreme}. Therefore $X\cap Z\ne\emptyset,\ Y\cap Z\ne\emptyset$.
By \Cref{lem: extreme-intersect-terminal}, $X\cap Z$ and $Y\cap Z$ are Steiner cuts, so
\[d(X\cap Z)\ge d(X), d(Y\cap Z)\ge d(Y)\]
by \Cref{fact:subset-of-extreme}. Moreover, one of the two inequalities must be strict. Because $Z$ is a proper subset of $X\cup Y$, $X\cap Z=X$ and $Y\cap Z=Y$ cannot hold simultaneously. Without loss of generality, assume $Y\cap Z\subsetneqq Y$; then since $Y$ is extreme, \[d(Y\cap Z)>d(Y)\]

Applying (\ref{eq:submodularity-3}) of \Cref{fact:submodularity} on $X$ and $Z$, and (\ref{eq:submodularity-4}) on $Y$ and $Z$, we have \begin{equation}\label{eq:extreme-union-3}
    d(X\cup Z)\le d(Z)
\end{equation}
\begin{equation}\label{eq:extreme-union-4}
    d(Y\cup Z)<d(Z)
\end{equation}
Adding (\ref{eq:extreme-union-3}) and (\ref{eq:extreme-union-4}) and using (\ref{eq:extreme-union-1}) give
\begin{equation}\label{eq:extreme-union-5}
    d(X\cup Z)+d(Y\cup Z)<d(Z)+d(Z)\le d(Z)+d(X\cup Y)
\end{equation}
(\ref{eq:submodularity-1}) of \Cref{fact:submodularity} on $X\cup Z$ and $Y\cup Z$ says
\begin{equation}\label{eq:extreme-union-6}
    d(X\cup Z)+d(Y\cup Z)\ge d((X\cup Z)\cap (Y\cup Z))+d((X\cup Z)\cup (Y\cup Z))=d((X\cap Y)\cup Z)+d(X\cup Y)
\end{equation}
Finally, (\ref{eq:extreme-union-5}) and (\ref{eq:extreme-union-6}) imply $d((X\cap Y)\cup Z)<d(Z)$.

If $(X\cap Y)\cup Z$ is extreme, let $Z'=(X\cap Y)\cup Z$; else let $Z'$ be an extreme violator of $(X\cap Y)\cup Z$ (which exists by \Cref{fact:extreme-violator}). Then $d(Z')\le d((X\cap Y)\cup Z) < d(Z)\le d(X\cup Y)$, $Z'\subseteq (X\cap Y)\cup Z \subseteq X\cup Y$, so $Z'$ is an extreme violator of $X\cup Y$ with cut value less than $Z$, which contradicts the assumption that $Z$ has minimum cut value among all extreme violators of $X\cup Y$. 
\end{proof}

\begin{Remark}
We note that $X\cap Y$ is not necessarily extreme, unlike in the global connectivity setting.
\end{Remark}

Since extreme sets are closed under set union, the concept of maximal extreme sets is well-defined. This motivates one of the key new concepts of the paper, the supreme sets, which we show form a laminar family.
\begin{Definition}[Supreme sets]\label{def:supreme}
Partition the extreme sets into equivalence classes according to projection: two extreme sets $X$ and $Y$ are equivalent if their projections $X\cap T$ and $Y\cap T$ are equal.

For each equivalence class with a common terminal set $R$, define the \emph{supreme set} of $R$, denoted $\super(R)$, to be the union of all extreme sets in the equivalence class, i.e.\ $\super(R)=\bigcup_{X\in\mathcal{X}:\rho(X)=R} X$. The set $\super(R)$ is defined only if $R\subseteq T$ is the projection of some extreme set; otherwise we say $\super(R)$ is undefined. In particular, $\super(\emptyset)$ and $\super(T)$ are undefined.
\end{Definition}
\begin{fact}\label{fact:supreme-extreme}
If $\super(R)$ is defined, then $\super(R)$ is an extreme set and $\rho(\super(R))=R$.
\end{fact}
\begin{proof}
$R\subsetneqq T$ because $\super(R)$ is defined. By definition, $\super(R)$ is the union of extreme sets $\{X_1,\ldots,X_r\}$, where $\rho(X_i)=R$ for each $X_i$. For any $1\le k\le r$, let $Y_k=\bigcup_{i=1}^{k} X_i$. By a simple inductive argument, $\rho(Y_k)=R$ and $Y_k$ is a Steiner cut for each $k$.

We prove by induction that $Y_k$ is extreme. The base case $k=1$ is trivial. As an induction step, assume $Y_k$ is extreme. If $Y_k\subseteq X_{k+1}$ (resp.\ $X_{k+1}\subseteq Y_k$), $Y_{k+1}=X_{k+1}$ (resp.\ $Y_{k+1}=Y_k$) is an extreme set. The remaining case is that both $Y_k\setminus X_{k+1}$ and $X_{k+1}\setminus Y_k$ are nonempty. $Y_k\cap X_{k+1}\supseteq R$ is also nonempty, so $Y_k$ and $X_{k+1}$ cross. We have shown $Y_{k+1}$ is a Steiner cut. By \Cref{lem:extreme-union}, $Y_{k+1}=Y_k\cup X_{k+1}$ is extreme.

In conclusion, $\super(R)=Y_r$ is extreme.
\end{proof}

\begin{lemma}
\label{lem:supreme-laminar}
The supreme sets form a laminar family.
\end{lemma}
\begin{proof}
Assume for contradiction that two supreme sets $X_1=\super(R_1)$ and $X_2=\super(R_2)$ cross. 
By \Cref{fact:supreme-extreme}, $X_1$ and $X_2$ are extreme sets. By \Cref{lem: extreme-difference-terminal}, one of $X_1\setminus X_2$ and $X_2\setminus X_1$ contains no terminal. Therefore $R_1\subseteq R_2$ or $R_2\subseteq R_1$. Without loss of generality, assume $R_1\subseteq R_2$. Then
$$\rho(X_1\cup X_2)=\rho(X_1)\cup\rho(X_2)=R_1\cup R_2=R_2.$$
By \Cref{def:supreme} of $\super(R_2)$, $X_1\cup X_2\subseteq \super(R_2)=X_2$, which contradicts the assumption that $X_1$ crosses $X_2$.

There is no crossing supreme sets, so they form a laminar family.
\end{proof}


Finally, we show that supreme sets can be efficiently computed as long as we know their projections, a key ingredient in our supreme sets algorithm.
\begin{lemma}
\label{lem:supreme-mincut}
For any supreme set $X=\super(R)$, $X$ is the earliest $R$-$(T\setminus R)$ mincut.
\end{lemma}
\begin{proof}
Because $\mu(R)$ is defined, $R\ne \emptyset$ and $R\ne T$. Let $S$ be the $R$-side of the earliest $R$-$(T\setminus R)$ min cut, which is a Steiner cut because $\rho(S)=R$. Because $X$ is an $R$-$(T\setminus R)$ cut and $S$ is an $R$-$(T\setminus R)$ min cut, 
\begin{equation}\label{eq:supreme-mincut-1}
    d(S)\le d(X)
\end{equation}

First we prove that $S$ is extreme. Assume otherwise; then by \Cref{fact:extreme-violator}, there exists an extreme violator $Z\subsetneqq S$ such that
\begin{equation}\label{eq:supreme-mincut-2}
    d(Z)\le d(S)
\end{equation}

Because $Z\subsetneqq S$, $\rho(Z)\subseteq \rho(S)=R$. We divide into two cases $\rho(Z)=R$ and $\rho(Z)\subsetneqq R$ and derive a contradiction for both cases.

When $\rho(Z)=R$, $Z$ is an $R$-$(T\setminus R)$ cut, and also a proper subset of the earliest $R$-$(T\setminus R)$ min cut $S$. So $d(Z)>d(S)$, which contradicts (\ref{eq:supreme-mincut-2}).

The remaining case is $\rho(Z)\subsetneqq R$. Let $Y=\super(\rho(Z))$. Because $Z$ is an extreme set, $Z\subseteq Y$ by \Cref{def:supreme}. Using \Cref{fact:subset-of-extreme},
\begin{equation}\label{eq:supreme-mincut-3}
    d(Z)\ge d(Y)
\end{equation}
By \Cref{lem:supreme-laminar}, $Y=\super(\rho(Z))$ and $X=\super(R)$ cannot cross. Because $X\setminus Y\supseteq R\setminus \rho(Z)\ne\emptyset$ and $X\cap Y\supseteq \rho(Z)\ne\emptyset$, we have $Y\setminus X=\emptyset$ and $Y\subsetneqq X$. Since $X$ is extreme,
\begin{equation}\label{eq:supreme-mincut-4}
    d(Y) > d(X)
\end{equation}

(\ref{eq:supreme-mincut-3}) and (\ref{eq:supreme-mincut-4}) give $d(Z)>d(X)$, but (\ref{eq:supreme-mincut-1}) and (\ref{eq:supreme-mincut-2}) give $d(Z)\le d(X)$, a contradiction.
In conclusion, $S$ is extreme.

By \Cref{def:supreme}, $S\subseteq \super(R)=X$.
$S\subsetneqq X$ would imply $d(S)>d(X)$ since $X$ is extreme, which contradicts (\ref{eq:supreme-mincut-1}). Therefore $S=X$.
\end{proof}

\eat{\begin{lemma}\label{lem:mincut-max-supreme}
A minimal Steiner min cut is a maximal supreme set.
\end{lemma}
\begin{proof}
Let $X$ be any minimal Steiner min cut, and $\lambda=d(X)$ be the Steiner min cut value. For any Steiner cut $W\subsetneqq X$, $d(W)>d(X)$ because $W$ is not a Steiner min cut. So $X$ is extreme.

If there exists an extreme set $Y\supsetneqq X$, then $d(X)>d(Y)$, which contradicts that $X$ is a Steiner min cut. So $X$ is a maximal extreme set. 

$\super(\rho(X))$ is an extreme set containing $X$, so $\super(\rho(X))=X$. $X$ has no extreme proper superset, so no supreme proper superset. In conclusion $X$ is a maximal supreme set.
\end{proof}}

\begin{lemma}\label{lem:supreme-contract}
For any supreme set $X=\super(R)$ in $G$, and any set $K$ which is either a subset of $X$ or disjoint from $X$, consider a graph $H$ by contracting $K$ from $G$. Then $X$'s image $X'$ is still a supreme set in $H$.
\end{lemma}
\begin{proof}
The contraction homomorphism $\varphi$ maps vertices in $K$ to a contracted node $k$, and maps vertices outside $K$ to themselves. Define its inverse homomorphism $\varphi^{-1}$ that maps $k$ to $K$ and vertices outside $K$ to themselves. We have $d_H(Z)=d(\varphi^{-1}(Z))$ for any vertices set $Z$ in $H$. In particular, the assumption means $X=\varphi^{-1}(X')$, so $d_H(X')=d(X)$.

Define extreme sets and supreme sets in $H$ with terminal set $\varphi(T)$.
Let $R'=X'\cap \varphi(T)$. Our goal is to prove that $X'=\super_H(R')$.

First we prove that $X'$ is extreme in $H$. Intuitively, by contracting a subset or disjoint set, the extreme condition of $X$ can only become weaker. Formally, for any Steiner cut $Z'\subsetneqq X'$, let $Z=\varphi^{-1}(Z')$, so that $Z\subsetneqq X$. Because $Z'$ is a Steiner cut in $H$, $\rho_H(Z')$ is not empty nor $\varphi(T)$, so $\rho(Z)=\varphi^{-1}(\rho_H(Z'))$ is not empty nor $T$. That is, $Z$ is a Steiner cut and also a proper subset of $X$.
Because $X$ is extreme,
\[d_H(Z')=d(Z)>d(X)=d_H(X')\]
So $X'$ is extreme in $H$ by \Cref{def:extreme}.
By \Cref{def:supreme} of $\super_H(R')$, $X'\subseteq \super_H(R')$.


Assume for contradiction that $X'\subsetneqq\super_H(R')$. 
Let $Y=\varphi^{-1}(\super_H(R'))$, so that $X\subsetneqq Y$. If there exists a terminal $t\in Y\setminus X$, then $t$ is mapped to $R'$, which means $t$ and an vertex in $X$ have the same image. However, this is impossible because $K$ is either contained in $X$ or disjoint from $X$. Therefore $\rho(Y)=\rho(X)=R$.

Because $\super_H(R')$ is extreme in $H$,
\begin{equation}\label{eq:supreme-contract-1}
    d(X)=d_H(X')>d_H(\super_H(R'))=d(Y)
\end{equation}
If $Y$ is extreme, then $Y$ is a subset of $X=\super(R)$ and $d(Y)\le d(X)$, which contradicts (\ref{eq:supreme-contract-1}). Therefore $Y$ is not extreme and there exists a violator $Z$ of $Y$ with
\begin{equation}\label{eq:supreme-contract-2}
    d(Z)\le d(Y)
\end{equation}
Because $\rho(Z)\subseteq \rho(Y)=R$ and $Z$ is extreme, $Z\subseteq \super(\rho(Z))\subseteq X$. Then $d(Z)\le d(X)$ by \Cref{fact:subset-of-extreme}. However, (\ref{eq:supreme-contract-1}) and (\ref{eq:supreme-contract-2}) give $d(X)>d(Z)$, a contradiction.
In conclusion, $X'=\super_H(R')$.
\end{proof}

\subsection{Algorithm for Finding Supreme Sets}

This section presents an algorithm that constructs a laminar forest $L$\eat{laminar tree $L$ of subsets of $T$} representing the terminal sets of all extreme sets, and calculates $c(T_i)=d(\super(T_i))$ for each $T_i\in L$ as a byproduct. This is a little weaker than the laminar forest of supreme sets guaranteed by \Cref{lem:supreme-laminar}. Instead of finding all supreme sets $\super(T_i)$, we only find their terminal sets $T_i$ and cut values $c(T_i)$, which are enough for our augmentation algorithm.

The algorithm consists of three phases. Phase 1 performs perturbation on edge weights to make cut values distinct, but also distorts the supreme sets. Phase 2 applies a divide and conquer algorithm to find a laminar forest containing all supreme sets of the perturbed graph. Finally in Phase 3, several rounds of post-processing extract the laminar forest on terminals of original graph from the laminar forest of supreme sets of perturbed graph.


\subsubsection{Phase 1: Perturbation}
We apply the same perturbation used in \cite{CenLP22a}. That is, for each edge $(u, v)$, let the new weight $\widetilde{w}(u,v)=mN\cdot w(u,v)+r(u,v)$, where $N$ is an integer larger than the sum of all edge weights but still polynomial in $n$, and $r(u,v)$ are uniformly and independently drawn from $\{1,2,\ldots, N\}$.
Let $\widetilde{G}$ be the graph after perturbation. Use $\widetilde{d}(\cdot)$ and $\widetilde{\lambda}(\cdot,\cdot)$ to denote cut values under edge weights $\widetilde{w}$.

Phase 1 takes $O(m)$ time to generate the random numbers.

\begin{fact}[\cite{CenLP22a}]\label{fact:perturb-order}
Perturbation does not change order of two unequal cut values.
\end{fact}
\begin{fact}\label{fact:perturb-extreme}
If $X$ is an extreme set under edge weights $w$, then $X$ remains extreme under edge weights $\widetilde{w}$.
\end{fact}
\begin{proof}
For every Steiner cut $Y$ that is a proper subset of $X$, we have $d(Y)>d(X)$ by \Cref{def:extreme}. By \Cref{fact:perturb-order}, $\widetilde{d}(Y)>\widetilde{d}(X)$. Therefore $X$ satisfies  \Cref{def:extreme} under edge weights $\widetilde{w}$.
\end{proof}
\begin{lemma}[\cite{CenLP22a}]
\label{lem:perturb}
After the perturbation, the following holds with high probability.
Fix any vertex $s$. For every vertex $t\ne s$, the minimum $s$-$t$ cut under $\widetilde{w}$ is unique. Moreover, for every vertex pair $t,t'\in V\setminus \{s\}$, either $\widetilde{\lambda}(s,t)\ne \widetilde{\lambda}(s,t')$, or the unique minimum $s$-$t$ cut is identical to the unique minimum $s$-$t'$ cut under $\widetilde{w}$.
\end{lemma}

\subsubsection{Phase 2: Divide and Conquer}

\begin{algorithm2e}[t]
\caption{Find the laminar forest of supreme sets.}
\label{alg:extreme}
\SetKwInOut{Input}{Input}
\SetKwInOut{Output}{Output}
\setcounter{AlgoLine}{0}
\Input{Graph $\widetilde{G}=(V,E)$ with edge weights $\widetilde{w}$, terminal set $T$.}
\smallskip
\If{$|T|\le 16$}{
Calculate the earliest $T_1$-$T_2$ min cut for every bipartition $(T_1,T_2)$ of $T$.\\
For any pair of crossing sets $(X, Y)$ in the family, either $d(X\setminus Y)\le d(X)$ or $d(Y\setminus X)\le d(Y)$ by (\ref{eq:submodularity-2}). In the former case remove $X$, and in the latter case remove $Y$.\\\label{line:base-check}
Return the laminar forest of remaining sets. 
}\Else{
\Repeat{$|T_1|\ge \frac{1}{16}|T|, |T_2|\ge \frac{1}{16}|T|$\label{line:repeat-condition}}{
 Uniformly sample two terminals $s,t\in T$. Let $\phi=\tilde{\lambda}(s,t)$.\\
 Compute $T_1=ct(s,\phi)\cap T$. Let $T_2=T\setminus T_1$.\\
 }
 Compute the earliest $T_1$-$T_2$ min cut. Let the two sides of the min cut be $S_1$ and $S_2$ where $T_1\subseteq S_1$, $T_2\subseteq S_2$.\\ \label{line:ct-mincut}  
 Form $\widetilde{G}_1$ by contracting $S_1$ in $\widetilde{G}$. Let $t_1$ be the contracted node of $S_1$. Recursively run the algorithm on input $(\widetilde{G}_1, T_2\cup\{t_1\})$. Let $L_1$ be the returned laminar forest.\\
 Form $\widetilde{G}_2$ by contracting $S_2$ into $t_2$. Recurse on $(\widetilde{G}_2, T_1\cup\{t_2\})$ to get laminar forest $L_2$.\\
 After deleting the leaf of $\{t_1\}$ in the forest $L_1$, attach all trees in $L_1$ to be children of the leaf of $\{t_2\}$ in $L_2$.\\ 
 \label{line:merge}
 Output the merged laminar forest.
}
\end{algorithm2e}

This section works on the perturbed graph $\widetilde{G}$ assuming the properties of \Cref{lem:perturb} hold. We remark that for simplicity, the reader can assume that the original graph $G$ has all $s$-$t$ min cuts unique, and simply pretend that no perturbation exists, i.e.\ $\widetilde G=G$.

We present in \Cref{alg:extreme} a divide and conquer algorithm that outputs a laminar family containing all supreme sets of $\widetilde{G}$. The algorithm is based on uncrossing property about cut thresholds (\Cref{lem:ct-uncross-terminal}).

\begin{Definition}\label{def:ct}
The cut threshold of source $s$ and threshold $\phi$ is defined as $$ct(s,\phi)=\{t\in V:\widetilde{\lambda}(s,t)\ge \phi\}.$$ We also use the convention that $s\in ct(s,\phi)$.
\end{Definition}

Given graph $\widetilde{G}$, the algorithm samples two terminals $s, t\in T$ uniformly, and uses cut threshold $ct(s,\tilde{\lambda}(s,t))$ to partition the terminals into $T_1$ and $T_2$. By \Cref{lem:ct-balance} proved later on, this partition is balanced on terminals with constant probability, which is crucial for efficiency. Let $(S_1, S_2)$ be the earliest $T_1$-$T_2$ min cut. We form two subproblems by contracting $S_1$ in one and $S_2$ in the other. \Cref{lem:ct-mincut-uncross-supreme} guarantees that the cut $(S_1, S_2)$ does not cross any supreme set. Then for any supreme set $X$, one of $S_1$ and $S_2$ is a subset of $X$ or disjoint from $X$, and $X$ remains supreme in one of the recursive calls by \Cref{lem:supreme-contract}. Therefore, all supreme sets can be found in the recursive calls.

The recursion reaches the base case when $|T|\le 16$, which is still nontrivial because the number of non-terminals is unbounded. By \Cref{lem:supreme-mincut}, every supreme set is the earliest min cut of a subset of $T$, so finding all $2^{|T|-1}$ such earliest min cuts produces a family containing all supreme sets. However, this family may not be laminar. Fortunately, for any pair of crossing sets $(X, Y)$, either $d(X\setminus Y)\le d(X)$ or $d(Y\setminus X)\le d(Y)$ by (\ref{eq:submodularity-2}), which means one of $X$ and $Y$ is violated by its subset and hence not extreme. By checking all pairs of sets in the family and remove a non-extreme set in each crossing pair, we get a laminar family containing all supreme sets.

Finally, we merge the two laminar forests returned by the recursive calls in \Cref{line:merge}, without losing any supreme sets.
By attaching the forest $L_1$ to the leaf of $\{t_2\}$ in $L_2$
and deleting the leaf of $\{t_1\}$, we preserve all sets in $L_1$ and $L_2$, except the sets in $L_1$ that contain $t_1$. \Cref{lem:ct-mincut-uncross-supreme} guarantees that these sets cannot be supreme, so we do not lose any supreme set in the merging step, and the final output is a laminar forest containing all supreme sets of $\widetilde{G}$.


\begin{lemma}
\label{lem:ct-uncross-terminal}
For any cut threshold $X=ct(s,\phi)$ and any extreme set $Y$, either $Y\setminus X$ contains no terminal, or at least one of $X\cap Y$ and $X\setminus Y$ is empty.
\end{lemma}
\begin{proof}
Assume for contradiction that there exists $X=ct(s,\phi)$ and extreme $Y$ such that $\rho(Y\setminus X)\ne\emptyset,\ X\cap Y\ne\emptyset,\ X\setminus Y\ne\emptyset$. Pick $u\in X\setminus Y$ if $s\in Y$, and pick $u\in X\cap Y$ if $s\notin Y$. Then $Y$ separates $s$ from $u \in X$. Also pick terminal $v\in \rho(Y\setminus X)$. By \Cref{def:ct} on $X$,
\begin{equation}\label{eq:ct-uncross-terminal-1}
    \widetilde{\lambda}(s,u)\ge \phi,
\end{equation}
\begin{equation}\label{eq:ct-uncross-terminal-2}
    \widetilde{\lambda}(s,v)< \phi.
\end{equation}

Let $Z$ be the $v$-side of the $s$-$v$ min cut (which is unique by \Cref{lem:perturb}).
\begin{equation}\label{eq:ct-uncross-terminal-3}
    \widetilde{d}(Z)=\widetilde{\lambda}(s,v)<\phi
\end{equation}
Then for all $z\in Z,\ \widetilde{\lambda}(s,z)\le \widetilde{d}(Z)<\phi$ and $z\notin X$ by \Cref{def:ct} on $X$, so $Z\cap X=\emptyset$. 
Therefore, $Y\cup Z$ is a $s$-$u$ cut, and
\begin{equation}\label{eq:ct-uncross-terminal-4}
\widetilde{d}(Y\cup Z) \ge  \widetilde{\lambda}(s,u)
\end{equation}

$Y\cap Z$ is a proper subset of $Y$ containing terminal $v$, and since $Y$ is extreme,
$\widetilde{d}(Y\cap Z)>\widetilde{d}(Y)$.
Use (\ref{eq:submodularity-4}) of \Cref{fact:submodularity} and then (\ref{eq:ct-uncross-terminal-3}),
\begin{equation}\label{eq:ct-uncross-terminal-5}
\widetilde{d}(Y\cup Z) < \widetilde{d}(Z) < \phi
\end{equation}
But (\ref{eq:ct-uncross-terminal-1}) and (\ref{eq:ct-uncross-terminal-4}) give $\widetilde{d}(Y\cup Z)\ge \phi$, which contradicts (\ref{eq:ct-uncross-terminal-5}).
\end{proof}

\begin{figure}
\centering
\begin{tikzpicture}
\draw {(0,.5) ellipse (1 and 2)}
{(.6,-0.3) ellipse (1.5 and 1)}
{(.6,1.9) ellipse (1.2 and .5)};
\draw (2.3,-.5) node{$X$} (2.1,2) node{$Z$} (-1.3,1) node{$Y$};
\draw (.1, 1.9) node{$v$} (1.5, -.3) node{$u$($s$)}  (.1,-.3) node {$s$($u$)};
\end{tikzpicture}
\caption{Proof of \Cref{lem:ct-uncross-terminal}}
\end{figure}

\begin{corollary}\label{cor:ct-uncross-terminal}
For any cut threshold $X$ and any extreme set $Y$,  $\rho(X)$ does not cross $\rho(Y)$. 
\end{corollary}
\begin{Remark}
Recall that in the setting of global min cuts, a cut threshold does not cross any extreme set. However, in Steiner min cuts setting, it is possible that a cut threshold crosses an extreme set in non-terminals, so \Cref{cor:ct-uncross-terminal} only gives the weaker property of uncrossing in terminal.
\end{Remark}

\begin{lemma}
\label{lem:ct-mincut-uncross-supreme}
Suppose $ct(s,\phi)$ partitions the terminals into two nonempty sides $T_1=\{t\in T:\widetilde{\lambda}(s,t)\ge \phi\}$ and $T_2=\{t\in T:\widetilde{\lambda}(s,t)<\phi\}$. Let $S$ be the $T_1$-side of the earliest $T_1$-$T_2$ mincut. Then $S$ does not cross any supreme set.
\end{lemma}
\begin{proof}
Assume for contradiction that $S$ crosses some supreme set $X=\widetilde{\super}(R)$.
By \Cref{cor:ct-uncross-terminal}, $R$ does not cross $T_1$, which is the terminal set of a cut threshold. Then there are three cases of the relation between $R$ and $T_1$:
\begin{enumerate}
\item[(1)] $R\subseteq T_1$. Then $ \rho(X\cap S)=R$ and $X\cap S$ is a Steiner cut. $X\cap S$ is also a proper subset of $X$ because $X$ and $S$ cross. Since $X$ is extreme, 
$\widetilde{d}(X\cap S)>\widetilde{d}(X)$. Using (\ref{eq:submodularity-4}) of \Cref{fact:submodularity}, $\widetilde{d}(X\cup S)<\widetilde{d}(S)$. This contradicts the the fact that $S$ is a $T_1$-$T_2$ min cut because $X\cup S$ is a $T_1$-$T_2$ cut.
\item[(2)] $R\cap T_1=\emptyset$. Then $ \rho(X\setminus S)=R$ and $X\setminus S$ is a Steiner cut. $X\setminus S$ is also a proper subset of $X$ because $X$ and $S$ cross. Since $X$ is extreme, 
$\widetilde{d}(X\setminus S)>\widetilde{d}(X)$. Using (\ref{eq:submodularity-2}) of \Cref{fact:submodularity}, $\widetilde{d}(S\setminus X)<\widetilde{d}(S)$. This contradicts the fact that $S$ is a $T_1$-$T_2$ min cut because $S\setminus X$ is a $T_1$-$T_2$ cut.

\item[(3)] $T_1 \subsetneqq R$. $X\cap S$ is a proper subset of $S$ because $X$ and $S$ cross, and also a $T_1$-$T_2$ cut. Because $S$ is the earliest $T_1$-$T_2$ cut, $\widetilde{d}(X\cap S)>\widetilde{d}(S)$. Using (\ref{eq:submodularity-4}) of \Cref{fact:submodularity},
\begin{equation}\label{eq:ct-mincut-uncross-supreme-1}
    \widetilde{d}(X\cup S)<\widetilde{d}(X)
\end{equation}
By \Cref{lem:supreme-mincut}, $X$ is an $R$-$(T\setminus R)$ min cut. Notice that $X\cup S$ is a $R$-$(T\setminus R)$ cut. Therefore, $\widetilde{d}(X\cup S)\ge \widetilde{d}(X)$, which contradicts (\ref{eq:ct-mincut-uncross-supreme-1}).
\end{enumerate}
All three cases derive contradiction, so $S$ does not cross any supreme set.
\end{proof}

\begin{lemma}\label{lem:divide-conquer-correct}
\Cref{alg:extreme} outputs a laminar forest of Steiner cuts and the family contains all supreme sets of $\widetilde{G}$.
\end{lemma}
\begin{proof}
First we prove that all sets in the output family are Steiner cuts. In the base case, we only calculate min cuts between terminal sets. In the recursive case, the output is merged from the output of subproblems. So the algorithm only outputs Steiner cuts.

Next we prove that every supreme set is in the output by induction on recursion depth. In the base case, we find all sets in the form of $T_1$-$T_2$ earliest min cut, which contains all supreme sets by \Cref{lem:supreme-mincut}. 
\Cref{line:base-check} removes a set $X$ only when $X$ crosses $Y$ and $d(X\setminus Y)\le d(X)$. If $X\setminus Y$ is a Steiner cut, then it is a violator of $X$. If $\rho(X\setminus Y)=\emptyset$, then $\rho(X\cup Y)=\rho(Y)$, and $d(X\cup Y)\ge d(Y)$ because $Y$ is a min cut between terminals. By (\ref{eq:submodularity-1}), $d(X\cap Y)\le d(X)$, and $X\cap Y$ is a violator of $X$. In both cases $X$ does not satisfy the extreme property, so the supreme sets are not removed.

As an inductive step, consider any supreme set $X$ and its relation to the cut $S_1$ in \Cref{line:ct-mincut} of \Cref{alg:extreme}. By \Cref{lem:ct-mincut-uncross-supreme}, one of the following holds: (a) $X\subseteq S_1$, (b) $X\supseteq S_1$, or (c) $X\cap S_1=\emptyset$. By \Cref{lem:supreme-contract}, in cases (b) and (c) $X$ is a supreme set in the subproblem contracting $S_1$, and in case (a) $X$ is a supreme set in the subproblem contracting $S_2$. The two subproblems return two laminar forests $L_1$ and $L_2$. Using the induction hypothesis, $X$ is found in $L_1$ for cases (b) and (c), and found in $L_2$ for case (a).

In the merge step (\Cref{line:merge}), we form a merged laminar forest $L$ by attaching $L_1$ to the leaf of contracted node in $L_2$, 
and deleting the leaf of contracted node of $L_1$. For case (a), $X$ is a set in $L_2$ disjoint from the contracted $S_2$, so $X$ is preserved in $L$. For case (b), $X$ is a set in $L_1$ containing the contracted $S_1$, and is preserved in $L$ because we expand the contracted node of $L_1$ to be $S_1$. Finally for case (c), $X$ is a set in $L_2$ disjoint from the contracted $S_1$, so $X$ is preserved in $L$. In conclusion the merge step does not lose supreme sets. The merged forest represents a laminar family because what we do is just expanding the contracted $S_1$ by a subtree of subsets of $S_1$.

By induction, the overall merged laminar forest contains all supreme sets of $\widetilde{G}$.
\end{proof}

Next we analyze the running time of \Cref{alg:extreme}. The analysis has the same spirit as \cite{CenLP22a}. We first show the recursion depth is $O(\log n)$, then bound the running time of each recursive call to be poly-logarithmic many max flows on the contracted graph of the subproblem, and finally prove that the total size of subproblems in each recursion layer is $\tO(m)$ edges and $\tO(n)$ vertices. It follows that the total running time is $\tO(F(m,n))$.

\begin{lemma}[\cite{LiP21approximate}]\label{lem:ct-time}
Given a graph $G=(V,E)$, a source vertex $s\in V$ and an integer $\phi>0$, cut threshold $ct(s,\phi)$ can be computed in $\tO(F(m,n))$ time with high probability.
\end{lemma}

\begin{lemma}[\cite{AbboudKT20focs}]
\label{lem:tournament-count}
For any graph $G=(V,E)$ and terminal set $T\subseteq V$, there exists $|T|/2$ nodes $u\in T$ such that at least $|T|/4$ other nodes $w\in T\setminus\{u\}$ satisfy $|U\cap T|>|W\cap T|$, where $(U,W)$ is the $u$-$w$ min cut. (If not unique, choose the same cut for $u$-$w$ min cut and $w$-$u$ min cut.)
\end{lemma}

\begin{lemma}\label{lem:ct-balance}
Uniformly sample two terminals $s,t\in T$ and let $\phi=\widetilde{\lambda}(s,t)$. 
After perturbation, use cut threshold $ct(s,\phi)$ to partition terminals into two sides $T_1=\{t\in T:\widetilde{\lambda}(s,t)\ge \phi\}$ and $T_2=\{t\in T:\widetilde{\lambda}(s,t)< \phi\}$. Then with  constant probability, $|T_1|\ge \frac{1}{16}|T|, |T_2|\ge \frac{1}{16}|T|$.
\end{lemma}
\begin{proof}
Uniformly sample $s\in T$. Assume the property of \Cref{lem:perturb} holds, which happens with high probability. Then the $s$-$t$ min cut for any other terminal $t$ is unique. Further assume $s$ satisfies the property of $u$ in \Cref{lem:tournament-count}, which holds with probability $\frac 12$.

We claim that there does not exist $\frac 34|T|$ terminals $t\in T\setminus\{s\}$ with the same $\widetilde{\lambda}(s, t)$ value. Assume for contradiction that there exists such a subset $R_1$. By \Cref{lem:perturb}, the unique $s$-$t$ min cuts are the same for all $t\in R_1$. Then for each $t\in R_1$, the $t$-side of $s$-$t$ min cut has more than $\frac 12|T|$ terminals. By property of \Cref{lem:tournament-count}, there exist $\frac 14|T|$ terminals $R_2\subseteq T\setminus\{s\}$ such that for each $t\in R_2$, the $t$-side of $s$-$t$ min cut has less than $\frac 12|T|$ terminals. $R_1$ and $R_2$ are disjoint, so $|T\setminus\{s\}|\ge |R_1|+|R_2| = \frac 34 |T|+ \frac 14 |T|=|T|$, which is impossible.

List all other terminals $t\in T\setminus\{s\}$ in non-increasing order of $\widetilde{\lambda}(s, t)$. 
Sample $t$ uniformly from the list. Further assume that $t$ has index between $\frac{1}{16} |T|$ and $\frac{1}{8}|T|$ in the list, which holds with probability $\frac{1}{16}$. Then $|T_1|=|\{t\in T:\widetilde{\lambda}(s,t)\ge \phi\}|\ge \frac{1}{16}|T|$. Also $|T_1|=|\{t:\widetilde{\lambda}(s,t)> \phi\}| + |\{t:\widetilde{\lambda}(s,t)= \phi\}| \le \frac 18 |T| + \frac 34 |T| = \frac 78 |T|$ by the claim. So $|T_2|=|T|-|T_1|\ge \frac{1}{16}|T|$. The probability that all three assumptions hold is at least $\frac{1}{32}(1-n^{-D})$ for a large enough constant $D$.
\end{proof}

\begin{lemma}
\Cref{alg:extreme} runs in $\tO(F(m,n)) $ time with high probability.
\end{lemma}
\begin{proof}
The recursion depth is $O(\log n)$ because each recursive call reduces the size of $T$ by a constant factor, and the recursion terminates when $|T|\le 16$. Note that each recursive call only introduces one contracted node to a subproblem, so in each subproblem, there are $O(\log n)$ contracted nodes. 

Let $(m', n')$ be the size of a subproblem. Let $m_i$ and $n_i$ be the total edge and vertex size of the depth-$i$ recursion layer respectively. Let $k_i$ be the number of subproblems in the depth-$i$ recursion layer. We first prove that $m_i=O(m\log^2 n),\ n_i=O(n\log n)$ for every layer, and then bound the running time of a subproblem to be $\tO(F(m', n'))$. It follows that the total running time is $\tO(F(m,n))$.

To prove the bounds on $m_i$ and $n_i$, we begin by bounding $k_i\le n$. We first claim that each instance has at least 3 vertices. It holds for recursive case $|T|>16$. For base case $|T|\le 16$, consider its parent instance with terminal set $T'$. (Assuming the original graph is sufficiently large, so the root is not in base case.)
$|T'|>16$, so $|T'_1|,|T'_2|\ge \frac{1}{16}|T'|>1$, and $|T'_1|, |T'_2|\ge 2$ as integers. The recursive calls have terminal sets $T'_1\cup \{t_2\}$ or $T'_2\cup \{t_1\}$, which contain at least 3 vertices.

Let $\alpha_i=n_i-2k_i$. Initially $\alpha_0=n-2$. Consider the difference of $\alpha_{i+1}$ from $\alpha_i$. A base case contributes to $\alpha_i$ (positively by the claim) but not to $\alpha_{i+1}$. A recursive call creates two new contracted nodes, so $n_{i+1}$ increases by 2, but $k_{i+1}$ also increases by 1. In conclusion $\alpha_{i+1}\le \alpha_i$.
Therefore for every depth $i$, $\alpha_i=n_i-2k_i\le n$. $n_i\ge 3k_i$ by the claim, so $k_i\le n$.

The uncontracted nodes are disjointly partitioned in a recursion layer. For contracted nodes, each of the $O(n)$ subproblems in a layer contains $O(\log n)$ contracted nodes, so the total number is $O(n\log n)$. Therefore, $n_i=O(n\log n)$.

Edges between uncontracted nodes are disjointedly partitioned. Edges between contracted and uncontracted nodes are repeated at most twice in a recursion layer because the uncontracted endpoint only occurs in one subproblem. For edges between contracted nodes, there are $O(\log^2 n)$ in each subproblem and $O(n\log^2 n)$ in a layer. In conclusion, $m_i=O(m+n\log^2 n)$.

Next we bound the running time of a recursive call. In the base case, $|T|\le 16$. The algorithm runs $2^{|T|-1}=O(1)$ max flows to calculate the min cuts. In the checking step of \Cref{line:base-check}, the algorithm calculates the cut values of $O(1)$ cuts, which can be done by enumerating all edges. In conclusion, the running time is $\tO(F(m',n'))$.


In the recursive case, the bottleneck of running time is calculating cut threshold, and calculating earliest $T_1$-$T_2$ min cut. By \Cref{lem:ct-time}, each cut threshold call costs $\tO(F(m',n'))$. The condition of \Cref{line:repeat-condition} is satisfied with constant probability by \Cref{lem:ct-balance}, so with high probability the \textbf{repeat} loop terminates in $O(\log n)$ rounds for all $O(n\log n)$ recursive calls. In conclusion, the running time is $\tO(F(m', n'))$.
\end{proof}

\subsubsection{Phase 3: Post-processing}

In Phase 2, we get a laminar forest $\widetilde{L}$ containing all supreme sets of $\widetilde{G}$. This forest differs from our goal in two ways. First, $\widetilde{L}$ may contain non-supreme sets of $\widetilde{G}$, because the contraction in divide and conquer may generate new extreme sets; second, the supreme sets of $\widetilde{G}$ may differ from the supreme sets of $G$, because perturbation may distort the relation of equal-valued cuts and generate new extreme sets. In post-processing, we apply several post-order traversal on each tree in the forest to remove redundant sets, and only keep the weaker forest structure of terminals and cut values.

With some abuse of notation, we refer to a tree node and its corresponding set interchangeably. When we remove a node $R$ during the algorithm, we assign its children to its parent if $R$ is not a root, and we designate its children to be new roots of the forest if $R$ is a root. In this way the new forest represents the laminar family after removing set $R$.

The algorithm works as follows. Compute the cut values under $w$ as well as $\widetilde{w}$ of each set in $\widetilde{L}$. Record the number of terminals inside each set in $\widetilde{L}$. Then run three post-order traversals:
\begin{enumerate}
 \item In the first post-order traversal, when visiting a node of set $R$, compare its cut value under $\widetilde{w}$ to all its children. If a child $W$ has $\widetilde{d}(W)\le \widetilde{d}(R)$, remove $R$.
 \item In the second post-order traversal, when visiting a node of set $R$, compare its number of terminals to its parent $W$ (if exists). If the numbers are equal, i.e.\ all terminals in $W$ are contained in $R$, then remove $R$.
 \item In the third post-order traversal, when visiting a node of set $R$, compare its cut value under $w$ to all its children. If a child $W$ has $d(W)\le d(R)$, remove $R$.
\end{enumerate}
The algorithm only removes sets from the forest, so all queried cut values are recorded. It is not hard to combine the three steps into one post-order traversal; we keep them separated for clarity.

After the second post-order traversal, we get a laminar family of all supreme sets under $\widetilde{w}$ (\Cref{lem:traversal2-supreme}). This is very close to supreme sets under $w$, but not identical. We show in \Cref{lem:traversal3} that after the third post-order traversal, the laminar forest is the same as the family of all supreme sets of $G$ in terms of terminal sets and cut values. We discard the information about non-terminals, and only record the projection on terminals and the cut value of each set. The output is a laminar family of terminal sets, and the cut value $c(R)=d(\super(R))$ of each terminal set $R$.

\begin{fact}\label{fact:traversal1-order}
If $W$ is an ancestor of $U$ on the forest, and $\widetilde{d}(W)\ge \widetilde{d}(U)$, then $W$ will be removed in the first post-order traversal.
\end{fact}
\begin{proof}
Let the tree path from $U$ to root be $(W_0=U, W_1, W_2,\ldots)$ after the first post-order traversal. We claim that $\widetilde{d}(U)>\widetilde{d}(W_1)>\widetilde{d}(W_2)>\ldots$. Assume otherwise; then for some $k$, $\widetilde{d}(W_k)\le \widetilde{d}(W_{k+1})$. When visiting $W_{k+1}$, $W_k$ is a child, so $W_{k+1}$ should be removed according to the description of the first post-order traversal. This contradicts the assumption of tree path.

The claim implies that any ancestor $W$ with $\widetilde{d}(W)\ge \widetilde{d}(U)$ is removed.
\end{proof}
\begin{fact}\label{fact:traversal3-order}
If $W$ is an ancestor of $U$ on the forest after the second post-order traversal, and $d(W)\ge d(U)$, then $W$ will be removed in the third post-order traversal.
\end{fact}
\begin{proof}
Identical to \Cref{fact:traversal1-order}.
\end{proof}

\begin{lemma}\label{lem:traversal2-supreme}
After the second post-order traversal, the laminar forest represents the family of all supreme sets of $\widetilde{G}$.
\end{lemma}
\begin{proof}
We first prove that after the second post-order traversal, all supreme sets are in the forest. By \Cref{lem:divide-conquer-correct}, this holds before the post-processing, so only need to show that any supreme set $X=\super(R)$ cannot be removed. Because $X$ is extreme (\Cref{fact:supreme-extreme}) and all sets in the forest are Steiner cuts (\Cref{lem:divide-conquer-correct}), no descendant of $X$ can be a violator of $X$, so $X$ is not removed in the first post-order traversal. By \Cref{lem:supreme-mincut}, $X$ is a $R$-$(T\setminus R)$ min cut. If $\rho(W)=R$, $\widetilde{d}(X)\le \widetilde{d}(W)$. This implies $X$'s parent $W$ cannot have $\rho(W)=R$, otherwise $W$ would be removed in the first post-order traversal. So $X$ cannot be removed in the second post-order traversal.

It remains to prove that after the second post-order traversal, all sets in the laminar forest are supreme sets.
We prove by bottom-up induction that after visiting a node $U$, all nodes in the subtree of $U$ are supreme sets under $\widetilde{w}$. Consider any node $U$. By induction, before visiting $U$, all nodes except $U$ in the subtree are supreme sets. We only need to prove that if $U$ is not removed, it is supreme. 
Let $R=\rho(U)$.

We first prove that $U$ is extreme. Suppose not; then there exists an extreme violator $Z$. Let $Z'=\super(\rho(Z))$, which is defined because $Z$ is extreme. As a supreme set, $Z'$ is in the forest. If $Z'$ is in the subtree of $U$, then $Z'$ is a subset of $U$, and $\widetilde{d}(Z')\le \widetilde{d}(Z)\le \widetilde{d}(U)$, so $U$ would be removed in the first post-order traversal by \Cref{fact:traversal1-order}. If $Z'$ is outside the subtree of $U$, then $Z'$ is an ancestor of $U$ because $\rho(Z'\cap U) \ne\emptyset$, which means $Z'$ is a superset of $U$ and $\rho(Z')\supseteq R$. We also have  $\rho(Z')=\rho(Z)\subseteq R$, so $\rho(Z')=R$. This means all sets on the path from $U$ to $Z'$ have projection $R$, and $U$ will be removed in the second post-order traversal. In conclusion, both cases contradict the assumption that $U$ is not removed.

Now that $U$ is extreme, let $W=\super(R)$. As a supreme set, $W$ is in the forest. Because $W\supseteq U$, $W$ is an ancestor of $U$. Then all sets on the path from $U$ to $W$ have projection $R$. The only case when $U$ is not removed in the second post-order traversal is $U=W$.

In conclusion, after the second post-order traversal, the laminar forest is exactly the family of all supreme sets in $\widetilde{G}$.
\end{proof}

\begin{lemma}\label{lem:traversal3}
Let $\mathcal{R}$ be the family of supreme sets in $G$ projected onto terminals. After the third post-order traversal, the laminar forest is $\{R\in \mathcal{R}:\widetilde{\super}(R)\}$. Moreover, $d(\super(R))=d(\widetilde{\super}(R))$ for all $R\in\mathcal{R}$. Note that the family of supreme sets in $G$ is $\{R\in \mathcal{R}:\super(R)\}$.
\end{lemma}
\begin{proof}
First we prove that $d(\widetilde{\super}(R))=d(\super(R))$ if $\super(R)$ is defined. By \Cref{fact:perturb-extreme}, perturbation cannot destroy extreme sets, so $\super(R)$ is also extreme under $\widetilde{w}$,  and $\super(R)\subseteq \widetilde{\super}(R)$. If $\widetilde{\super}(R)$ is extreme before perturbation, $\widetilde{\super}(R)\subseteq \super(R)$, so $\widetilde{\super}(R)=\super(R)$ and $d(\widetilde{\super}(R))=d(\super(R))$. The remaining case is that $\widetilde{\super}(R)$ is not extreme before perturbation,
so there exists an extreme violator $Z$ of $\widetilde{\super}(R)$ by \Cref{fact:extreme-violator}. Because $Z\subseteq \widetilde{\super}(R)$, $\rho(Z)\subseteq R$, and $Z\subseteq \super(\rho(Z))\subseteq \super(R)$ by laminarity of supreme sets (\Cref{lem:supreme-laminar}). So \[d(\super(R))\le d(Z)\le d(\widetilde{\super}(R))\]
(The first inequality uses \Cref{fact:subset-of-extreme} on $\super(R)$, and the second follows that $Z$ is a violator.)
Because $\widetilde{\super}(R)$ is extreme after perturbation, \[\widetilde{d}(\super(R))>\widetilde{d}(\widetilde{\super}(R)).\] By \Cref{fact:perturb-order}, perturbation does not change order of unequal cut values, so the two inequalities imply $d(\super(R))=d(\widetilde{\super}(R))$.

Next we prove that if $X=\widetilde{\super}(R)$ is defined but $\super(R)$ is not defined, $X$ will be removed during the third post-order traversal. Because $\super(R)$ is not defined, there is no extreme set with terminal set $R$, and $X$ is not extreme under $w$. So there exists an extreme violator $Z$ such that $Z\subsetneqq X$ and $d(Z)\le d(X)$. Because $X$ is extreme under $\widetilde{w}$, $\widetilde{d}(Z)>\widetilde{d}(X)$. This is impossible if $d(Z)<d(X)$ because perturbation does not change order of unequal cut values (\Cref{fact:perturb-order}), so $d(Z)=d(X)$. Because $\super(R)$ is not defined, $\rho(Z)\subsetneqq R$. Because $Z$ is extreme, $\super(\rho(Z))$ is defined and $Y=\widetilde{\super}(\rho(Z))$ is in the forest after the second post-order traversal. Because $X\setminus Y\ne\emptyset$, and $X, Y$ are in the laminar family, $Y\subsetneqq X$. We have $d(Y)\le d(Z)=d(X)$, so $X$ will be removed in the third post-order traversal by \Cref{fact:traversal3-order}.

Third we prove that for all $R\in\mathcal{R}$, $\widetilde{\super}(R)$ is in the family after the third post-order traversal. $\super(R)$ is defined because $R\in \mathcal{R}$. By \Cref{fact:perturb-extreme}, $\super(R)$ is extreme under $\widetilde{w}$, and $X=\widetilde{\super}(R)$ is defined. By \Cref{lem:traversal2-supreme}, $X$ is in the laminar forest before the third post-order traversal. It remains to show that $X$ is not removed in the third post-order traversal. Assume otherwise; then there exists a child $W\subsetneqq X$ such that $d(W)\le d(X)$.
Because $X$ is extreme under $\widetilde{w}$, $\widetilde{d}(W)>\widetilde{d}(X)$, so $d(W)=d(X)$.
By \Cref{lem:traversal2-supreme}, $W$ is another supreme set under $\widetilde{w}$. Let $W=\widetilde{\super}(P)$, then $P\subsetneqq R$. We proved $d(W)=d(\super(P)), d(X)=d(\super(R))$, so $d(\super(P))=d(\super(R))$. Because the supreme sets under $w$ are laminar, $\super(P)\subsetneqq \super(R)$, so $\super(P)$ violates the extreme property of $\super(X)$. This contradicts the fact that $\super(X)$ is extreme under $w$.
\end{proof}

\begin{lemma}
The post-processing takes $\tO(m)$ time.
\end{lemma}
\begin{proof}
Given the cut values, the post-order traversal can be executed in $O(n)$ time. Next show that cut values can be computed in $\tO(m)$ time. We describe the process under $w$, and $\widetilde{w}$ is the same.

We use a dynamic tree data structure on the laminar forest. For the sake of analysis, add a dummy root representing $V$ and let all trees in the forest be its children.  
In this laminar tree, every vertex is contained in a path from some node to the root. Attach each vertex $u$ to the lowest node $l(u)$ of the path. For each edge $u, v$, add $w(u,v)$ to the path from the node $l(u)$ to $l(v)$, excluding the lowest common ancestor (LCA) of $l(u)$ and $l(v)$. This costs $\tO(m)$ time.

To see the correctness of cut values, we prove that for any set $X$ in the tree and any edge $(u,v)$, $(u, v)$ is in cut $X$ if and only if $X$ is on the path from $l(u)$ to $l(v)$ excluding the LCA. The edge $(u,v)$ is in cut $X$ iff $u\in X,\, v\notin X$ or $v\in X,\, u\notin X$, which is equivalent to that exactly one of $l(u)$ and $l(v)$ is in the subtree of $X$. This is further equivalent to $X$ being on the path from $l(u)$ to $l(v)$ excluding the LCA.
\end{proof}


\section{Algorithm for Steiner Connectivity Augmentation}
\label{sec:augmentation}
In this section, we give the algorithm for Steiner connectivity augmentation. We will assume throughout that the target connectivity $\tau \ge 2$. If not, an optimal solution is a spanning tree on the components containing terminals. This algorithm has three parts which we present in different subsections below:
\begin{enumerate}
    \item[(a)] An algorithm for external augmentation. (\Cref{sec:external})
    \item[(b)] An algorithm for adding augmentation chains to achieve Steiner connectivity of $\tau-1$. (\Cref{sec:chains})
    \item[(c)] An algorithm to augment the Steiner connectivity from $\tau-1$ to $\tau$ using random matchings. (\Cref{sec:augment-1})
\end{enumerate}

\subsection{External Augmentation}
\label{sec:external}
This section considers the external augmentation problem. That is, given an undirected graph $G=(V,E)$ with integer edge weights $w\ge 0$, terminal set $T\subseteq V$ and a connectivity target $\tau$, insert a new node $x$, and add edges connecting $x$ and vertices in $V$ with minimum total weight
(the weights of new edges are chosen by the algorithm), such that the Steiner connectivity of $T$ is increased to $\tau$.

Given the laminar forest $L$ of supreme sets on terminals, external augmentation can be solved by a simple greedy algorithm in \Cref{alg:external}. For a rooted forest containing a node $R$, define $ch(R)$ as the set of children of $R$.  Define the recursive demand $\rdem(R)$ of terminal set $R$ to be the minimum total weight of edges across the cut $\super(R)$ 
to satisfy all supreme sets in the subtree of $R$ in external augmentation. To satisfy $\super(R)$, the demand is $\tau-c(R)$; to satisfy other supreme sets in the subtree, take the sum of recursive demands of $R$'s children. Therefore $\rdem(R)=\max\{\tau-c(R), \sum_{P\in ch(R)} \rdem(P)\}$ and can be computed by tree DP.
The algorithm performs a post-order traversal on each tree of $L$. When visiting each node $R$, add edges from $x$ to any terminal in $R$ such that the new edges crossing $R$ reach total weight of $\rdem(R)$. That is, add a new edge of weight $\rdem(R)-\sum_{P\in ch(R)} \rdem(P)$.

\begin{algorithm2e}[t]
\caption{External augmentation.}
\label{alg:external}
\SetKwInOut{Input}{Input}
\SetKwInOut{Output}{Output}
\setcounter{AlgoLine}{0}
\Input{A laminar forest $L$ of supreme sets on terminals, a function $c(R)=d(\super(R))$ for every node $u\in L$.}
\smallskip
\ForEach{$u\in L$ in post-order traversal}{
 Calculate $\rdem(R)=\max\{\tau-c(R), \sum_{W\in ch(R)} \rdem(W)\}$.\\
 \label{line:rdem}
 Add an edge of weight $\rdem(R)-\sum_{W\in ch(R)} \rdem(W)$ from $x$ to any terminal in $R$.
}
 Output the new edges.
\end{algorithm2e}

\begin{lemma}
\label{lem:external-rdem}
For any node $R\in L$, at the time \Cref{alg:external} finishes visiting $R$, the algorithm adds edges of total weight $\rdem(R)$ across $R$.
\end{lemma}
\begin{proof}
We prove by induction on decreasing depth of $R$ in the rooted forest. As a base case, when $R$ is a leaf, the sum on children is empty, so the algorithm adds an edge of weight $\rdem(R)$ across $R$. As an inductive step, consider an internal node $R$. When visiting $R$, every child $P\in ch(R)$ does not get new edges after visiting it because the algorithm only add edges to terminals in the visited node during post-order traversal. Therefore the new edges crossing each child $P$ has a total weight of $\rdem(P)$ when visiting $R$. We add edges of weight $\rdem(R)-\sum_{P\in ch(R)} \rdem(P)\ge 0$ across $R$, so the total weight of new edges across $R$ is $\rdem(R)$ after visiting $R$. 
\end{proof}

\begin{lemma}
\label{lem:external-feasible}
\Cref{alg:external} outputs a feasible solution to the external augmentation problem.
\end{lemma}
\begin{proof}
Consider any Steiner cut $X$. If $X$ is not extreme, by \Cref{fact:extreme-violator} there exists an extreme violating set $Z\subsetneqq X$ with $d(Z)\le d(X)$. If $Z$ is augmented to $\tau$, so is $X$ because any new edge across $Z$ also crosses $X$. Therefore we only need to show that every extreme set $X$ is augmented to value $\tau$.

Let $R=\rho(X)$. $R\in L$. By \Cref{lem:external-rdem}, when visiting $R$, we add $\rdem(R)$ edges across $R$. The algorithm only adds edges, so the final cut value of $X$ is at least $d(X)+\rdem(R)\ge d(\super(R)) + \tau-c(R)=\tau$.
\end{proof}

\begin{lemma}
\label{lem:external-optimal}
\Cref{alg:external} outputs an optimal solution of external augmentation. The optimal value is $\sum_{R\in L_1}\rdem(R)$, where $L_1$ is the set of all roots in $L$.
\end{lemma}
\begin{proof}
We claim that for each $R\in L$, we need to add at least $\rdem(R)$ edges across $\super(R)$. First, $\super(R)$ requires $\tau-c(R)$ edges because $c(R)=d(\super(R))$. Second, by induction, for each child $W\in ch(R)$, $\super(W)$ requires $\rdem(W)$ edges, so $\super(R)$ requires $\sum_{W\in ch(R)}\rdem(W)$ edges as a super set of all $\super(W)$'s.

The overall lower bound on the total number of new edges is $\sum_{R\in L_1}\rdem(R)$. \Cref{alg:external} adds $\rdem(R)$ edges across $R$ when visiting $R$ by \Cref{lem:external-rdem}, and does not add any new edge across $R$ after visiting $R$ if $R\in L_1$. Therefore the output of \Cref{alg:external} has size $\sum_{R\in L_1}\rdem(R)$. As a feasible solution by \Cref{lem:external-feasible}, the output is an optimal solution.
\end{proof}

\begin{lemma}
\Cref{alg:external} takes $O(n)$ time.
\end{lemma}
\begin{proof}
The size of forest $L$ is $O(n)$. Calculating $\rdem(R)$ takes $|ch(R)|$ time, so \Cref{line:rdem} takes $O(n)$ time in total. When visiting each node, we only add edge to one terminal, so inserting edges takes $O(n)$ time in total.
\end{proof}

\subsection{Splitting off View}\label{sec:splitoffview}
The external augmentation problem is closely related to Steiner connectivity augmentation problem. In our algorithm for Steiner connectivity augmentation problem, assume the input is graph $G$, terminals $T$ and target $\tau$. Construct the corresponding external augmentation problem with the same graph $G$, terminals $T$ and target $\tau$. After finding the structure of supreme sets in \Cref{sec:extreme-sets}, \Cref{alg:external} produces the optimal external augmentation solution $F^{\text{ext}}$.

Assume the total weight of $F^{\text{ext}}$ is $k$. For the sake of analysis, regard the weighted edges as parallel unweighted edges. If $k$ is odd, add an arbitrary external edge to make the degree of $x$ even. By Mader's splitting off theorem \ref{thm:mader}, we can split $x$ off to get $\lceil\frac k2\rceil$ edges $F$ on $V$ which preserves the Steiner connectivity to be $\tau$. (There is no cut edge incident on $x$ because the Steiner connectivity is at least 2 after external augmentation and the degree of $x$ is always even during the splitting off.) This means $F$ is a feasible solution to the Steiner connectivity augmentation problem. On the other hand, $\lceil\frac k2\rceil$ is a lower bound on the value of any feasible solution by \Cref{lem:external-lower-bound}. Therefore, $F$ is an optimal solution to Steiner connectivity augmentation problem. This means we can get an optimal solution to augmentation problem by splitting off the external augmentation output.

In general, splitting off is not easier than augmentation. However, our external augmentation solution only consists of edges connecting $x$ and $T$, which makes splitting off easier. Going through this reduction also provides useful information. In our augmentation problem, we now know how an optimal solution looks like, in terms of the degree of every vertex in the new edges. We call it degree constraint of vertices, and say a terminal $u$ has vacant degree if the number of new edges incident to $u$ has not reached the degree constraint. Our algorithm will only add edges to terminals with vacant degree.

\begin{theorem}[Mader's theorem \cite{Mader78}]\label{thm:mader}
If a vertex $x$ is not incident to a cut edge and $d(x)\ne 3$, then there exist a pair of edges incident to $x$ such that splitting them preserves the $s$-$t$ minimum cut value for each pair $s,t\in V\setminus \{x\}$.
\end{theorem}

\begin{lemma}\label{lem:external-lower-bound}
If the optimal external augmentation solution has value $k$, then any feasible solution to the corresponding Steiner augmentation problem has value at least $\lceil\frac k2\rceil$.
\end{lemma}
\begin{proof}
Let $F$ be any feasible solution to Steiner augmentation problem. Let $p$ be the total weight of $F$. Construct $F'$ by replacing every edge $(u,v)\in F$ by two edges $(u,x)$ and $(v,x)$ with the same weight. Then $F'$ is a feasible solution to external augmentation with total weight $2p$. Because $k$ is the optimal value of external augmentation, $k\le 2p$, which implies $p\ge \frac k2$. Because we work on integer weights, $p$ and $k$ are integers, and $p\ge \lceil\frac k2\rceil$.
\end{proof}

\subsection{Greedily Adding Augmentation Chains}
\label{sec:chains}
In this section, we split off the external augmentation output.

For a terminal set $R$ of some extreme set, define its demand $\dem(R)=\tau-d(\super(R))$, and the cut value will change as we add edges.
List the maximal supreme sets $(X_1,\ldots,X_r)$ with demands at least 2, such that $X_1$ and $X_r$ have the minimum cut value. Let their terminal sets be $(R_1, \ldots, R_r)$. Because a Steiner min cut of minimum size is an extreme set with no extreme superset, $X_1$ and $X_r$ are Steiner min cuts. Define an augmentation chain to be a set of edges $F=\{(a_i,b_{i+1})\}_{1\le i\le r-1}$ such that $a_i, b_i\in R_{i}$, and every vertex's degree does not exceed its degree in the external augmentation solution.

A basic algorithm is to repetitively add augmentation chains until all demands of extreme sets are reduced to at most 1. We solve the case of demand 1 in \Cref{sec:augment-1}. During the algorithm, the extreme sets structure may change and need to be maintained. Our key observation \Cref{thm:chain-no-new-extreme} guarantees that the algorithm does not create new extreme sets, so we only need to handle the original supreme sets in $L$. To maintain the value $c(R)=d(\super(R))$ for every $R\in L$, notice that the effect of adding an edge $(a, b)$ in augmentation chain is increasing the cut value of all supreme sets on the $L$ tree path from $a$ to $b$ excluding the LCA, so we can maintain $c(R)$ by a dynamic tree data structure. For a node $R\in L$, if $c(R)$ is increased to make $c(R)\ge c(P)$ for some child $P$ of $R$, then $R$ is no longer the terminal set of any extreme set and need to be removed from $L$.

The running time of the basic algorithm is $O(n\tau)$, that is  $O(n)$ for each augmentation chain. To speed it up, we lazily use the same edges to form the chain until they become invalid. We use a priority queue to maintain the time when each edge will become invalid, which is further maintained by efficient data structures. Then we can add the same chain while repeatedly replacing the first invalid edge. This speedup is identical to \cite{CenLP22a} and discussed in \Cref{sec:lazy}. 

\begin{theorem}\label{thm: aug-chain-runtime}
    Given the supreme sets forest and the cut values of the supreme sets, we can add all augmentation chains so that all supreme sets have demand at most 1 in time $\tO(n)$. 
\end{theorem}

\subsubsection{Property of Augmentation Chains}
We now prove the key property of augmentation chains, stated below.
\begin{theorem}
\label{thm:chain-no-new-extreme}
Adding an augmentation chain does not create new extreme set.
\end{theorem}
\begin{proof}
Assume for contradiction that $U$ is a new extreme set after adding $F=\{(a_i,b_{i+1})\}_{1\le i\le r-1}$, and $W$ is an extreme violating set of $U$ before augmentation.

Let $d(\cdot)$ and $d'(\cdot)$ be the cut values before and after augmentation respectively. Then $d(W)\le d(U)$ but $d'(W)>d'(U)$. So $d_F(W)>d_F(U)$. Because we add at most 2 edges to any extreme set, $2\ge d_F(W)>d_F(U)\ge 0$. So $d_F(U)$ can only be 0 or 1. $d_F(W)>d_F(U)$ also implies that there must be a new edge crossing $W$ but not crossing $U$. Such an edge $e_0=(a_i,b_{i+1})$ must have $a_i, b_{i+1}\in U$. Also, no extreme set can contain all terminals in $U$ because $e_0$ crosses two maximal supreme sets.

Starting from $e_0=(a_i,b_{i+1})$, we have $b_i$ and $a_{i+1}$ also belong to $U$ by Lemma \ref{lem:u-walk}. Because $d_F(U)\le 1$, one of $a_{i-1}$ and $b_{i+1}$ is contained in $U$. By repeating this step, we can always walk along $F$ in one of the two directions, until reaching $a_1$ or $b_r$. By Lemma \ref{lem:u-cross-x}, $U$ cannot cross $X_1$ or $X_r$ in terminal because they only have one new edge. Therefore $\rho(X_1)\subsetneqq U$ or $\rho(X_r)\subsetneqq U$. Assume wlog $\rho(X_1)\subsetneqq U$.

$X_1\cap U$ is a proper subset of $U$ that contains terminals, so $d'(X_1\cap U)>d'(U)$. Then
\[d(X_1\cap U)\ge d'(X_1\cap U)-1\ge d'(U)=d(U)+d_F(U)\ge d(U)\]
By (\ref{eq:submodularity-3}) of \Cref{fact:submodularity}, $d(X_1\cup U)\le d(X_1)$. $X_1\cup U$ has the same terminal set as $U$, so it is a Steiner cut and $d(X_1\cup U)\ge d(X_1)$. ($X_1$ is a Steiner min cut.) Therefore $d(X_1\cup U)=d(X_1)$ and $X_1\cup U$ is also a Steiner min cut. Now the chain of inequality becomes equality, so $d_F(U)=0$.

When $d_F(U)=0$, we can continue the walking on both directions of $F$, so $U$ contains all edges of the chain $F$. By Lemma \ref{lem:mincut-uncross-extreme}, because $X_1\cup U$ is a Steiner min cut, it cannot cross any extreme set. Therefore $U$ must contain all terminals which is in an extreme set with demand at least 2. However, $X_1\cup U$ is a Steiner min cut, so there is terminal in $\overline{X_1\cup U}$ which has demand $\dem(X_1)\ge 2$, a contradiction.
\end{proof}

\begin{lemma}
\label{lem:u-cross-x}
Assume $U$ is a new extreme set after adding $F$. If $U$ crosses an extreme set $X$ in terminal, then $d_F(X)=2$, and the two new edges across $X$ are both contained in $U$.
\end{lemma}
\begin{proof}
Because $U$ crosses $X$ in terminal, $U\cap X, X\setminus U$ and $U\setminus X$ are Steiner cuts. Because $X$ is extreme, $d(X\setminus U)>d(X)$. By (\ref{eq:submodularity-4}) of \Cref{fact:submodularity}, $d(U\setminus X)<d(U)$. Because $U$ is a new extreme set, $d'(U\setminus X)>d'(U)$. Therefore
\[d(U\setminus X)+d_F(U\setminus X) = d'(U\setminus X) \ge d'(U)+1=d(U)+d_F(U)+1\ge d(U\setminus X)+d_F(U)+2\]
\[d_F(U\setminus X)\le d_F(U)+d_F(U\setminus X, U\cap X)\le d_F(U)+d_F(X)=d_F(U)+2\]
Comparing these two inequalities, we have $d_F(U\setminus X)=d_F(U)+2$, and all inequalities are equality. Therefore $d_F(X)=2$, and the two new edges incident to $X$ must both connect $U\setminus X$ and $U\cap X$.
\end{proof}

\begin{lemma}
\label{lem:u-walk}
Assume $U$ is a new extreme set after adding $F$. If $X_i$ has two new edges incident on $a_i$ and $b_i$, and $a_i\in U$, then $b_i\in U$ as well.
\end{lemma}
\begin{proof}
Assume for contradiction that $b_i\notin U$. Then $U$ crosses $X_i$. ($a_i, b_i$ are terminals in $U\cap X_i$ and $X_i\setminus U$. $U\setminus X_i$ is also a Steiner cut because no extreme set can contain all terminals in $U$.) By Lemma \ref{lem:u-cross-x}, $b_i\in U$, a contradiction.
\end{proof}

\begin{lemma}
\label{lem:mincut-uncross-extreme}
A Steiner min cut cannot cross any extreme set.
\end{lemma}
\begin{proof}
Assume for contradiction that a Steiner min cut $Y$ crosses an extreme set $X$.

If $\rho(X\setminus Y)\ne\emptyset$ and $\rho(Y\setminus X)\ne\emptyset$, we have $d(X\setminus Y)>d(X)$ by extreme property of $X$, and $d(Y\setminus X)\ge d(Y)$ because $Y$ is a Steiner min cut. Adding them contradicts (\ref{eq:submodularity-2}) of \Cref{fact:submodularity}.

The remaining case is that $X\setminus Y$ or $Y\setminus X$ contains no terminal. Because $X$ and $Y$ are Steiner cuts, $\rho(X\cap Y)\ne\emptyset$. $d(X\cap Y)>d(X)$ by extreme property of $X$. By submodularity $d(X\cup Y)<d(Y)$, so $X\cup Y$ is not a Steiner cut and $\rho(X\cup Y)=T$. Because $X$ and $Y$ are Steiner cuts, they cannot contain all the terminals. So $\rho(X\setminus Y)\ne\emptyset, \rho(Y\setminus X)\ne\emptyset$, which contradicts the assumption.
\end{proof}

\subsection{Augmenting Connectivity by One via Random Matchings}
\label{sec:augment-1}
After adding as many augmentation chains as possible, there is no extreme set with demand at least 2. If all demands are at most 0, the solution is complete. The remaining case is that there are extreme sets with demand 1, or equivalently cut value $\tau-1$. This means the Steiner minimum cut is now $\tau-1$. This section solves the subproblem that we augment Steiner connectivity by 1, from $\tau-1$ to $\tau$. The new edges need to augment or cross every Steiner mincut.

The optimal external augmentation is adding a star connecting $x$ and every demand-1 extreme sets. (These extreme sets must be maximal and disjoint.) Contract these extreme sets, and let $K$ be the set of contracted extreme sets. In the splitting off view, the optimal solution is a perfect matching of $K$. (If $k=|K|$ is odd, we match an arbitrary node twice.)

\subsubsection{Cactus Structure}

In global connectivity setting, previous algorithms construct the matching using the cactus structure of global minimum cuts. A cactus is a graph where any two simple cycles intersect in at most one node. A min cut cactus is a weighted cactus where the nodes correspond to a partition of the vertices of the original graph, and every min cut of the original graph is represented by a min cut with the same vertex partition in the cactus. The demand-1 extreme sets are minimal min cuts, so they are represented by all degree-2 nodes of the cactus. Given the min cut cactus, one can construct in $O(n)$ time a matching on all leaves to augment every min cut \cite{NaorGM97}.

In contrast, Steiner min cuts do not form such a neat structure. In the special case of $T=\{s, t\}$, the Steiner min cuts are just $s$-$t$ min cuts, which can only be represented by a directed acyclic graph (DAG). Fortunately, when projected onto terminals, the Steiner min cuts also admit a cactus structure. This is known as the `skeleton' part of connectivity carcass structure \cite{DinitzV94}.

Calculating the cactus structure of Steiner min cuts on terminals can be time-consuming. We circumvent the construction, and only exploit its existence to bound the success probability of our random matching algorithm. The key lemma is as follows.

\begin{lemma}
\label{lem:cyclic-order}
There exists a cyclic order of $K$ such that any Steiner min cut partitions the order into two intervals.
\end{lemma}
\begin{proof}
Take any Euler tour\footnote{An Euler tour is a (not necessarily simple) cycle that contains every edge of the graph exactly once. Here, the graph is the underlying cactus of the cactus structure.} of a cactus structure of Steiner min cuts on terminals. Take the subsequence of all degree-2 nodes. Any Steiner min cut has value $2$ so it cuts the Euler tour in two edges, breaking it into two intervals. Restricted to the degree-2 nodes, it also partitions the cyclic order into two intervals.
\end{proof}

\begin{algorithm2e}[t]
\caption{Augment Steiner connectivity by 1.}
\label{alg:augment-1}
\SetKwInOut{Input}{Input}
\SetKwInOut{Output}{Output}
\setcounter{AlgoLine}{0}
\Input{Graph $G$, terminal set $T$, demand-1 extreme sets family $K$ as a partition of $T$.}
\smallskip
Contract every set in $K$ to be a node. (When adding edge to a node, add to an arbitrary terminal in the node.)\\
\If{$|K|$ is odd}{Pick arbitrary $u,v\in K$, add edge $\{u, v\}$, and remove $u$ from $K$.}
\While{$|K|\ge 4$}{
\Repeat{$M$ is a feasible partial solution\label{line:feasible}}
{Randomly match $\lfloor \frac{|K|}{4}\rfloor$ pairs of nodes in $K$. Let $M$ be the partial matching.}
Remove endpoints in $M$ from $K$.
}
Match the last two nodes of $K$.
\end{algorithm2e}

\subsubsection{Algorithm Description and Correctness}
As presented in \Cref{alg:augment-1}, the algorithm runs in phases, each phase adding a partial matching. A matching of $K$ is feasible if it augments every Steiner min cut. Call a partial matching feasible if it is contained in a feasible matching of $K$. Because we always end a phase with a feasible partial solution, there always exists a feasible matching containing the current partial solution. This invariant holds when $|K|$ is reduced to 2, in which case matching the last two nodes is the only solution and must be a feasible solution.

The ignored detail about checking whether $M$ is a feasible partial solution in \Cref{line:feasible} is explained in \Cref{lem:matching-feasible}.

\begin{lemma}
\label{lem:matching-feasible}
Form a new graph $G'$ by adding a node $x$ and adding a star connecting $x$ and $K\setminus V(M)$. $M$ is a feasible partial solution if and only if the Steiner min cut of $G'$ is $\tau$.
\end{lemma}
\begin{proof}
If the Steiner min cut of $G'$ is $\tau$, by Mader's splitting off theorem \ref{thm:mader} we can split off $x$ to get a feasible matching of $K$.

If the Steiner min cut is not $\tau$, it must be $\tau-1$. (The Steiner min cut of $G'$ is at least $\tau-1$ because the Steiner min cut of $G$ is $\tau-1$. The Steiner min cut of $G'$ is at most $\tau$ because the degree cut of any $u\in K$ is $\tau$.) Let $(S,\Sbar)$ be a Steiner cut of value $\tau-1$. Without loss of generality, assume $x\in S$. Because $S$ has cut value at least $\tau-1$ in $G$, there can be no star edge across the cut. Therefore all unmatched nodes are in $S$, and the cut $S$ will not be augmented for any matching containing $M$. 
\end{proof}

\subsubsection{Running Time}

The algorithm runs in phases, each phase reduce $|K|$ from $2t$ to $2t-2\lfloor \frac{2t}{4}\rfloor \le t+1$ ($|K|$ is always even), so the algorithm runs $O(\log n)$ phases. Each phase makes several trials until getting a feasible partial solution. By \Cref{lem:matching-success-prob}, the success probability is at least $\frac 13$, so each phase ends in $O(\log n)$ trials whp. In each trial, sampling the random matching takes $O(n)$ time. Checking the feasibility calls a Steiner min cut according to \Cref{lem:matching-feasible}, which runs in poly-logarithmic many max flows \cite{LiP20deterministic}. In conclusion, the total running time is $\tO(F(m,n))$.

\begin{lemma}
\label{lem:matching-success-prob}
The success probability that $M$ is a feasible partial matching is at least $\frac 13$.
\end{lemma}
\begin{proof}
Generate $M$ by picking one uniformly random unmatched pair of nodes at each step. In each step, match the picked pair, and remove them from the cyclic order given by \Cref{lem:cyclic-order}. Let $|K|=k$ initially, then $k-2i$ unmatched nodes remain after the $i$-th step. 

Let $A_i$ be the event that in the $i$-th step, the picked pair $(u,v)$ is not adjacent in the cyclic order. $\Pr[A_i]=\frac{k-2i-1}{k-2i+1}$. Because each step is independent,
\[\Pr[A_1\land A_2\land\ldots \land A_{\lfloor k/4\rfloor}]= \frac{k-3}{k-1}\times \frac{k-5}{k-3}\times \ldots \times \frac{k-2\lfloor k/4\rfloor-1}{k-2\lfloor k/4\rfloor+1} = \frac{k-2\lfloor k/4\rfloor-1}{k-1}\ge \frac 13\]
when $k\ge 4$.

Finally, we prove that when no step matches adjacent nodes, the partial solution is feasible. Assume for contradiction that all events $A_i$ happen, but the output partial matching is not feasible. Then by \Cref{lem:matching-feasible}, there exists a Steiner min cut $(S,\Sbar)$ such that all leaves in its one side $S$ are matched. Consider the last matched leaf $u$ in $S$. Assume $u$ is matched with $v$. Because cut $(S,\Sbar)$ is not augmented, $v\in S$. $S$ forms an interval in the cyclic order by \Cref{lem:cyclic-order}. $u, v$ are the last two matched nodes in $S$, so they are adjacent in cyclic order when matched, which contradicts the assumption.
\end{proof}

\section{Algorithm for Steiner Splitting-Off}
\label{sec:splitting}
In this section, we generalize our Steiner augmentation algorithm to solve the Steiner splitting-off problem.
We assume there is no cut edge incident to $x$ and the degree of $x$ is even, which a standard assumption for splitting-off theorems including Lov\'asz's theorem and Mader's theorem.

The splitting-off problem reduces to a degree-constrained external augmentation problem by removing the external vertex $x$ and setting the degree constraint $\beta(u)$ to be the weight of removed $(u,x)$ edge for each vertex $u$. The connectivity target $\tau$ of the degree-constrained problem is the Steiner connectivity in the original graph. 

The algorithm follows the same framework as \Cref{sec:augmentation}. However, degree constraints introduce a significant new complication: that we now need to add edges incident to nonterminals. (Recall that in the Steiner connectivity augmentation algorithm, all edges added to the graph were incident to terminals only.)  As before, we have three parts:
\begin{enumerate}
    \item[(a)] For external augmentation, we use maximum flow calls on suitably defined auxiliary graphs for each heavy path in a heavy-light decomposition of the trees of the supreme sets forest. The external edges are tied to the heavy paths, and this allows us to track the contribution of edges connecting nonterminals. (\Cref{sec:deg-external})
    \item[(b)] Next, we adapt the definition of augmentation chains to incorporate degree constraints in a way that retains its key properties and allows us to design a greedy algorithm. (\Cref{sec:deg-splitting})
    \item[(c)] Finally, for augmenting Steiner connectivity from $\tau-1$ to $\tau$, we incorporate the notion of {\em surrogates} in our random matching algorithm. (\Cref{sec:deg-augment-1})
\end{enumerate}

\subsection{Degree-Constrained External Augmentation}
\label{sec:deg-external}
In degree-constrained augmentation, we have a degree constraint $\beta(u)$ on each vertex $u$, and the aim is to augment Steiner connectivity to $\tau$ while maintaining that the degree of $u$ in the new edges is at most $\beta(u)$. Although our unconstrained augmentation algorithm only adds edges to terminals, in degree-constrained setting we may have to add to non-terminals as well because of the terminals' degree constraints. This makes the problem much harder because now a new edge does not augment all extreme sets in a supreme set, and there may be exponentially many crossing extreme sets in a supreme set.

In degree-constrained setting, we also define external augmentation, that is inserting a new vertex $x$ and adding edges from $V$ to $x$, such that each vertex $u\in V$ has at most $\beta(u)$ new edges, and the Steiner connectivity of $T$ is augmented to at least $\tau$.
A degree-constrained external augmentation instance is feasible if there exists such an edge set. Equivalently, the instance is feasible if adding a $(u,x)$ edge of capacity $\beta(u)$ for each vertex $u$ augments the Steiner connectivity to at least $\tau$.
Because our degree constraints are generated from the splitting-off problem, we can assume the problem is feasible. The rest of the section constructs an optimal solution with minimum total weight.




Recall that in \Cref{sec:external}, for external augmentation without degree constraint, the optimal total weight is $\sum_{R\in L_1}\rdem(R)$, where the sets in $L_1$ are the terminal projection of maximal supreme sets. This is a lower bound on the value of degree-constrained external augmentation solution, because any feasible degree-constrained solution is also feasible for the unconstrained problem. Surprisingly, \Cref{lem:deg-external-optimal} shows that the optimal degree-constrained solution matches this lower bound regardless of the constraints, as long as the problem is feasible.

We restate the related notations here. A node $R\in L$ represents the terminal set of a supreme set. $\super(R)$ is the supreme set of terminal set $R$. The recursive demand of $R\in L$ is
$\rdem(R)=\max\{\tau-d(\super(R)), \sum_{W\in ch(R)} \rdem(W)\}$.

If some $R\in L$ has $d(\super(R))\ge\tau$, then all extreme sets with terminal set $R$ are already satisfied by extreme property of $\super(R)$, and $R$ can be contracted to its parent in $L$. After this step, we always have $\rdem(R)>0$ for $R\in L$.
Call a node $R$ {\it critical} if $$\rdem(R)=\tau-d(\super(R))>\sum_{W\in ch(R)}\rdem(W).$$ In particular, leaves are critical because $\rdem(R)>0$.

Next we describe the external augmentation algorithm. Apply heavy-light decomposition separately on each tree of the forest $L$. This partitions the forest into a vertex-disjoint 
collection of paths (called heavy paths). Define the depth of a heavy path $P$ to be the number of heavy paths intersecting the path from $P$ to the root of the tree containing $P$.
A key property of heavy-light decomposition is that the maximum depth of heavy paths is $O(\log n)$.
Process the heavy paths in a bottom-up manner by indexing the paths in non-increasing 
order of depth, and treating the paths one by one in that order. 

Let $G_\ell$ be the external augmentation graph after processing the first $\ell$ heavy paths. Initially $G_0=(V\cup\{x\}, E)$ is the original graph after adding an external vertex $x$.
In the $\ell$-th iteration, the algorithm adds external edges into $G_{\ell -1}$ to form a graph $G_\ell$, and the aim is to satisfy all extreme sets with terminal sets on the $\ell$-th heavy path $P_\ell$. To determine these new external edges, we run a max flow on an auxiliary graph $H_\ell$. The algorithm only adds edges, so all extreme sets will be satisfied after processing all heavy paths.

We construct a directed graph $H_\ell$ based on $G_{\ell -1}$ and run $R_1$-$x$ max flow on $H_\ell$, where $R_1$ is the leaf of the $\ell$-th heavy path $P_\ell$. In the construction, all extreme sets with terminal sets on $P_\ell$ are $R_1$-$x$ cuts, so their cut values on $H_\ell$ are at least the $R_1$-$x$ max flow value, which turns out to be $\tau$ by \Cref{lem:deg-external-flow-value}. We construct $G_\ell$ from $H_\ell$ such that these extreme sets have the same cut values in $G_\ell$ as their flow values in $H_\ell$ for the $R_1$-$x$ max flow, so these extreme sets are augmented to at least $\tau$. Let $P_\ell = (R_1^{(\ell)},R_2^{(\ell)},\ldots,R_{k_\ell}^{(\ell)})$, $R_1^{(\ell)}\subsetneqq R_2^{(\ell)}\subsetneqq\ldots\subsetneqq R_{k_\ell}^{(\ell)}$. For simplicity, we will drop the superscripts of $R_i^{(\ell)}$ in the rest of the section, as long as the path is clear. The following are the detailed steps for constructing $H_\ell$ from $G_{\ell -1}$, illustrated in \Cref{fig:deg-external}.
\begin{enumerate}
    \item Replace every undirected edge $(u,v)$ by two directed edges $(u,v)$ and $(v, u)$.
    \item Add an external vertex $y$.
    \item Add a directed edge from $u$ to $y$ with capacity $\beta(u)-w_{G_{\ell -1}}(u,x)$ for every $u\in V$, where $w_{G_{\ell -1}}$ is the edge capacity in $G_{\ell -1}$.
    \item  For each $1\le i\le k$, contract $R_i\setminus R_{i-1}$. Also contract $V\cup\{x\}\setminus \tilde{\super}(R_k)$. ($\tilde{\super}(\cdot)$ is the supreme sets in $\tilde{G}$, the graph after perturbation.) For the sake of analysis, we explain contraction as adding a cycle of infinite capacity connecting the contracted set, so that the vertices set of $H_\ell$ is $V\cup\{x,y\}$.
    \item  Add a directed edge of infinite capacity from $s_{i+1}$ to $s_i$ for each $1\le i\le k-1$, where $s_i$ is an arbitrary vertex in the contracted $R_i\setminus R_{i-1}$.
    \item  Add a directed edge from $y$ to $x$ of capacity $\Delta_\ell=\rdem(R_k)-\sum_{W\in B(R_k)} \rdem(W)$, where $B(R_k)$ is the set of roots of heavy paths in the subtree of $R_k$ except $R_k$ itself.
\end{enumerate} 

Run an $R_1$-$x$ max flow on $H_\ell$ to get an integer max flow $f_\ell$. 
Form $G_\ell$ by adding the flow values on $(u, y)$ edges to the capacity of $(u, x)$ edges in $G_{\ell -1}$ for each $u\in V$. That is, $w_{G_\ell}(u,x)=w_{G_{\ell -1}}(u,x)+f_\ell(u,y)$ for each $u\in V$.
\begin{figure}[t]
    \centering
    \includegraphics[width=.4\textwidth]{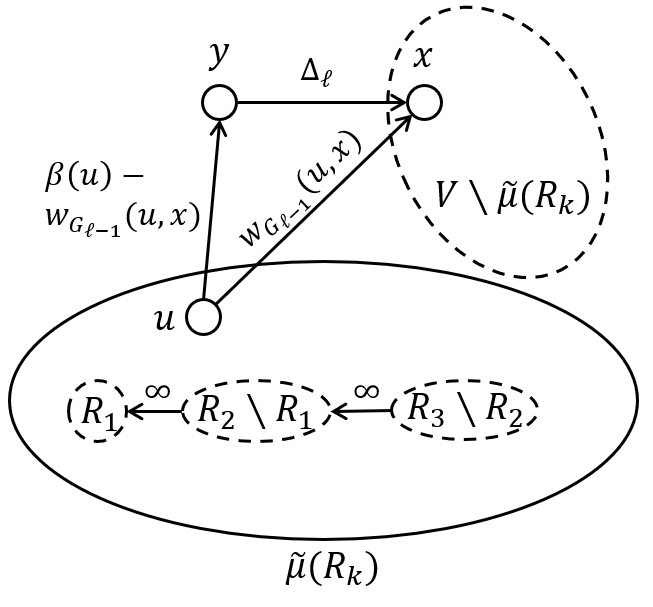}
    \caption{Construction of the flow network $H_\ell$. Dashed ovals are contracted sets.}
    \label{fig:deg-external}
\end{figure}

After processing all $p$ heavy paths, the resulting graph $G_p$ consists of the original graph $G$ and external edges connecting $V$ and $x$. Return these external edges as the output. We prove in \Cref{lem:deg-external-optimal} that this output is feasible and optimal. The proof is inductive and based on analysis of the flow $f_\ell$.

\begin{lemma}\label{lem:deg-external-flow-value}
Assume the external augmentation is feasible. Then in each $H_\ell$, the $R_1^{(\ell)}$-$x$ min cut value is $\tau$.
\end{lemma}


We defer the proof of \Cref{lem:deg-external-flow-value} and use it to show the optimality of the external augmentation algorithm.
Use $\delta_{H_\ell}(\cdot)$ to denote the outgoing cut value in $H_\ell$, emphasizing that $H_\ell$ is directed. 

\begin{lemma}\label{lem:deg-external-cut-value}
If $S\subseteq \tilde{\super}(R_{k_\ell}^{(\ell)})$ and $\rho(S)=R_i^{(\ell)}$ is a node in path $P_\ell$, then
\begin{itemize}
 \item $\delta_{H_\ell}(S)=d(S)+\beta(S)$, where $\beta(S)=\sum_{u\in S}\beta(u)$.
 \item $\delta_{H_\ell}(S\cup\{y\})=d_{G_{\ell -1}}(S)+\Delta_\ell$.
\end{itemize}
\end{lemma}
\begin{proof}
Because $S\subseteq \tilde{\super}(R_{k_\ell})$ and $\rho(S)=R_i$, $S$ does not cross the contracted $R_j\setminus R_{j-1}$ and $V\cup \{x\}\setminus \tilde{\super}(R_k)$, 
and the infinite capacity edges in step 5 do not contribute to $\delta_{H_\ell}(S)$ or $\delta_{H_\ell}(S\cup \{y\})$.

In $\delta_{H_\ell}(S)$, the edges on $V$ contribute $d(S)$. For each $u\in V$, $(u,x)$ and $(u,y)$ contribute $w_{G_{\ell -1}}(u,x)+\beta(u)-w_{G_{\ell -1}}(u,x)=\beta(u)$ by step 3. So $\delta_{H_\ell}(S)=d(S)+\beta(S)$.

In $\delta_{H_\ell}(S\cup \{y\})$, the edges on $V$ and $(u,x)$ edges contribute $d_{G_{\ell -1}}(S)$ in total. The $(y,t)$ edge contributes $\Delta_\ell$. The $(u,y)$ edges do not cross the cut. So $\delta_{H_\ell}(S\cup\{y\})=d_{G_{\ell -1}}(S)+\Delta_\ell$.
\end{proof}

\begin{lemma}\label{lem:deg-external-optimal}
The degree-constraint external augmentation algorithm outputs a feasible solution of optimal value $\sum_{R\in L_1}\rdem(R)$.
\end{lemma}
\begin{proof}
First we prove that every Steiner cut is augmented to at least $\tau$.
We only need to prove this for extreme sets. For other Steiner cuts $U$, consider their extreme violators $W$ in \Cref{fact:extreme-violator}. If $W$ is satisfied, we add $\tau-d(W)\ge \tau-d(U)$ edges connecting $W$ and $x$, which also cross $U$ because $W\subseteq U$. So satisfying all extreme sets implies satisfying all Steiner cuts.

For any extreme set $X$ with terminal set $R$, suppose $R$ is processed in the $\ell$-th heavy path $P_\ell=(R_1,R_2,\ldots,R_{k_\ell})$ and $R=R_i$. We have $R_1\subseteq R\subseteq X$. By definition of supreme sets, $X\subseteq \super(R)\subseteq \super(R_k)$. By \Cref{fact:perturb-extreme}, $\super(R_k)\subseteq \tilde{\super}(R_k)$. So $R_1\subseteq X\subseteq \tilde{\super}(R_k)$ and $X$ is an $R_1$-$x$ cut that is not contributed by infinite capacity edges in $H_\ell$.
Let $\lambda_\ell$ be the flow value of $f_\ell$. Then the flow value on cut $X$ is 
\[\lambda_\ell=\sum_{u\in X, v\in V\cup\{x,y\}\setminus X}f_\ell(u,v)-f_\ell(v,u).\] In this flow value, the flow in edges on $V$ and $(u,x)$ edges is upper bounded by their total capacity $d_{G_{\ell -1}}(X)$. Combined with the flow in $(u,y)$ edges, \[\lambda_\ell\le \sum_{u\in X, v\in V\cup\{x\}\setminus X}f_\ell(u,v)+\sum_{u\in X}f_\ell(u,y)\le d_{G_{\ell -1}}(X)+\sum_{u\in X}f_\ell(u,y).\]
Because we define $w_{G_\ell}(u,x)=w_{G_{\ell -1}}(u,x)+f_\ell(u,y)$, we have $d_{G_\ell}(X)=d_{G_{\ell -1}}(X)+\sum_{u\in X}f_\ell(u,y)\ge \lambda_\ell$.
Since
$d_{G_\ell}(X)\ge \lambda_\ell= \tau$ by \Cref{lem:deg-external-flow-value}, we conclude that $X$ is satisfied.

Next we prove that the output satisfies degree constraints, so it is a feasible solution. Fix any vertex $u\in V$.
For each $\ell$, $w_{G_\ell}(u,x) = w_{G_{\ell -1}}(u,x)+f_\ell(u,y) \le w_{G_{\ell -1}}(u,x)+\beta(u)-w_{G_{\ell -1}}(u,x)=\beta(u)$. The output uses $w_{G_p}(u,x)$ degrees after processing the last heavy path $p$. By induction, $w_{G_p}(u,x)\le \beta(u)$, so the edge weight does not exceed $\beta(u)$.

By \Cref{lem:external-optimal}, the output value is lower bounded by $\sum_{R\in L_1}\rdem(R)$, which is the optimal value if we remove the degree constraint.
Because the $(y,x)$ edge in $H_\ell$ has capacity $\Delta_\ell$ and we only add new edge for flows through $y$, each heavy path adds at most $\Delta_\ell=\rdem(R_k)-\sum_{W\in B(R_k)}\rdem(W)$. So the total value of output is upper bounded by the sum $\sum_{R\in L_1}\rdem(R)$. ($L_1$ nodes have rdem added once, and other nodes have rdem added once and subtracted once.) In conclusion the output value is $\sum_{R\in L_1}\rdem(R)$ and is optimal.
\end{proof}


Instead of proving \Cref{lem:deg-external-flow-value} by itself, we prove it in combination with the following statement. 

\begin{lemma}\label{lem:deg-external-mincut}
For each $\ell$, let $R_m^{(\ell)}$ be the highest critical node on the $\ell$-th heavy path. Then $\super(R_m^{(\ell)})\cup\{y\}$ is an $R_1$-$x$ min cut in $H_\ell$.
\end{lemma}

We give corollaries of the statements of \Cref{lem:deg-external-flow-value} and \Cref{lem:deg-external-mincut} to facilitate the inductive proof.

\begin{lemma}\label{lem:deg-external-new-edges}
If the statement of \Cref{lem:deg-external-mincut} holds for $\ell$, then $G_\ell$ is constructed from $G_{\ell -1}$ by adding $\Delta_\ell$ edges from $\super(R_m^{(\ell)})$ to $x$.
\end{lemma}
\begin{proof}
By the statement, $\super(R_m)\cup\{y\}$ is an $R_1$-$x$ min cut in $H_\ell$. Also, $f_\ell(y,x)=\Delta_\ell$ because $(y,x)$ is a min cut edge. By flow conservation, $f_\ell(y,x)=\sum_{u\in V}f_\ell(u,y)$. The algorithm adds the integer flow on $(u,y)$ edges as new edges to $G_\ell$.

In an $R_1$-$x$ max flow, no flow goes from sink side to source side of an $R_1$-$x$ min cut. Because $y$ is in the source side, all flow through $y$ are from the source side, so all new edges are from $\super(R_m)$.
\end{proof}

\begin{lemma}\label{lem:deg-external-cut-increase}
If the statement of \Cref{lem:deg-external-flow-value} and \Cref{lem:deg-external-mincut} hold for $\ell$, then for any $R\subsetneqq T$ such that $R\supseteq R_k^{(\ell)}$, the $R$-$(T\setminus R)$ min cut value increases by $\Delta_\ell$ from $G_{\ell -1}$ to $G_\ell$.
\end{lemma}
\begin{proof}
Let the $R$-$(T\setminus R)$ min cut values in $G_{\ell -1}$ and $G_\ell$ be $\lambda_{\ell -1}$ and $\lambda_\ell$ respectively. Let $U$ be the $R$-side of the latest $R$-$(T\setminus R)$ min cut in $G_\ell$. By \Cref{lem:deg-external-new-edges}, the algorithm adds $\Delta_\ell$ edges to $\super(R_m)$ in $G_\ell$. We claim that $U\supseteq \super(R_m)$, which implies $\lambda_\ell=d_{G_\ell}(U)=d_{G_{\ell -1}}(U)+\Delta_\ell\ge \lambda_{\ell -1}+\Delta_\ell$. The increment is exactly $\Delta_\ell$ because iteration $\ell$ only adds $\Delta_\ell$ units of edges.

To prove the claim, assume for contradiction that $\super(R_m)\setminus U\ne\emptyset$. Because $U$ is the latest $R$-$(T\setminus R)$ min cut and $R_m\subseteq R_k\subseteq R$, $d_{G_\ell}(U\cup \super(R_m))>d_{G_\ell}(U)$. By submodularity,
\begin{equation}\label{eq:deg-external-cut-increase-1}
    d_{G_\ell}(U\cap \super(R_m))<d_{G_\ell}(\super(R_m))
\end{equation}
By the statement of \Cref{lem:deg-external-mincut} for $\ell$, $\super(R_m)\cup\{y\}$ is an $R_1$-$x$ min cut in $H_\ell$, so the statement of \Cref{lem:deg-external-flow-value} for $\ell$ gives
\begin{equation}\label{eq:deg-external-cut-increase-2}
    \tau = \delta_{H_\ell}(\super(R_m)\cup\{y\})=d_{G_{\ell -1}}(\super(R_m))+\Delta_\ell
 = d_{G_\ell}(\super(R_m)) 
\end{equation}
(The second step uses \Cref{lem:deg-external-cut-value}, and the third step uses \Cref{lem:deg-external-new-edges}.)

Because $R_m\subseteq U\cap \super(R_m)\subseteq \super(R_m)$, $U\cap \super(R_m)$ is an $R_1$-$x$ cut that is not contributed by infinite capacity edges in $H_\ell$. We have $d_{G_\ell}(U\cap \super(R_m))=d_{G_{\ell-1}}(U\cap \super(R_m))+\sum_{u\in U\cap \super(R_m)}f_\ell(u,y)$, which is at least the flow value of $f_\ell$ on cut $U\cap \super(R_m)$. The $R_1$-$x$ max flow value on $H_\ell$ is $\tau$ by the statement of \Cref{lem:deg-external-flow-value} for $\ell$, so $d_{G_\ell}(U\cap \super(R_m))\ge \tau$. But (\ref{eq:deg-external-cut-increase-1}) and (\ref{eq:deg-external-cut-increase-2}) gives $d_{G_\ell}(U\cap \super(R_m))<\tau$, a contradiction.
\end{proof}

Before proving \Cref{lem:deg-external-flow-value}, we first introduce two helper statements below.

\begin{fact}\label{fact:yt-capacity}
In any path $P_\ell$, for each $i>1$,
\[\rdem(R_i)-\sum_{W\in B(R_i)}\rdem(W)\ge \rdem(R_{i-1})-\sum_{W\in B(R_{i-1})}\rdem(W)\]
Moreover, if $R_i$ is not critical, then the inequality can be replaced by equality.
\end{fact}
\begin{proof}
By definition of rdem,
\[\rdem(R_i)\ge \sum_{W\in ch(R_i)}\rdem(W) = \rdem(R_{i-1})+\sum_{W\in ch(R_i)\setminus \{R_{i-1}\}}\rdem(W).\]
When $R_i$ is not critical, the inequality can be replaced by equality.
Note that $ch(R_i)\setminus \{R_{i-1}\}=B(R_i)\setminus B(R_{i-1})$, so the statement holds.
\end{proof}
\begin{corollary}\label{cor:highest-black}
Let $R_m$ be the highest critical node in a path $P_\ell$. Then
\[\rdem(R_{k_\ell})-\sum_{W\in B(R_{k_\ell})}\rdem(W)=\rdem(R_m)-\sum_{W\in B(R_m)}\rdem(W).\]
\end{corollary}

Now we are prepared to prove \Cref{lem:deg-external-flow-value} in combination with \Cref{lem:deg-external-mincut}.

\begin{proof}[Proof of \Cref{lem:deg-external-flow-value} and \Cref{lem:deg-external-mincut}]
Because every heavy path stops at a leaf, 
$R_1$ is critical. So the highest critical node $R_m$ exists.

We prove by induction on $\ell$. Fix the path $P_\ell$, and assume the statement holds for all paths $P_j$ with $j<\ell$.

Let $S$ be the $R_1$-side of the latest $R_1$-$x$ min cut in $H_\ell$. Because of the infinite capacity edges, $S$ contains a prefix of $(R_1,R_2\setminus R_1,\ldots, R_{k_\ell}\setminus R_{k_\ell-1})$, and $S\subseteq \tilde{\super}(R_{k_\ell})$. That is, $\rho(S)=R_i$ for some node $R_i$ in $P_\ell$. 

First we prove that $\delta_{H_\ell}(S)\ge \tau$. If $y\notin S$, $\delta_{H_\ell}(S)=d(S)+\beta(S)$ by \Cref{lem:deg-external-cut-value}, which is at least $\tau$ because the problem is feasible. The remaining case is $y\in S$. Let $S'=S\setminus\{y\}$. Now $\delta_{H_\ell}(S)=d_{G_{\ell -1}}(S')+\Delta_\ell$ by \Cref{lem:deg-external-cut-value}.

$S'$ is an $R_i$-$(T\setminus R_i)$ cut. Initially, the $R_i$-$(T\setminus R_i)$ min cut value is $d(\super(R_i))$ by \Cref{lem:supreme-mincut}. For every other heavy path $P_j$ in the subtree of $R_i$, iteration $j$ adds $\Delta_j$ to the min cut value by \Cref{lem:deg-external-cut-increase}. Before iteration $\ell$, all other paths in the subtree of $R_i$ are already processed, adding a total of $\sum_{W\in B(R_i)}\rdem(W)$ to the min cut value. So $d_{G_{\ell -1}}(S')\ge d(\super(R_i))+\sum_{W\in B(R_i)}\rdem(W)$. Also, by applying \Cref{fact:yt-capacity} iteratively from $k_\ell$ to $i$, $\Delta_\ell = \rdem(R_{k_\ell}) - \sum_{W\in B(R_{k_\ell})} \rdem(W)\ge \rdem(R_i) - \sum_{W\in B(R_i)} \rdem(W)$. Therefore,
\begin{align*}
    \delta_{H_\ell}(S) &= d_{G_{\ell -1}}(S')+\Delta_\ell\\
    &\ge d(\super(R_i))+\sum_{W\in B(R_i)}\rdem(W) + \rdem(R_i) - \sum_{W\in B(R_i)} \rdem(W)\\
    &\ge d(\super(R_i)) + \tau-d(\super(R_i)) = \tau
\end{align*}

Next we prove that $\delta_{H_\ell}(\super(R_m)\cup\{y\})=\tau$, so it is an $R_1$-$x$ min cut.
\begin{align*}
    \delta_{H_\ell}(\super(R_m)\cup\{y\}) &= d_{G_{\ell -1}}(\super(R_m))+\Delta_\ell\\
    &= d_{G_{\ell -1}}(\super(R_m))+ \rdem(R_m) - \sum_{W\in B(R_m)} \rdem(W)\\
    &= d(\super(R_m))+\sum_{W\in B(R_m)}\rdem(W) + \rdem(R_m) - \sum_{W\in B(R_m)} \rdem(W)\\
    &= d(\super(R_m)) + \tau-d(\super(R_m)) = \tau
\end{align*}
The second equality uses \Cref{cor:highest-black}. The third equality is because each path $P_j$ in the subtree of $R_m$ adds $\Delta_j$ new edges (in $G_{\ell -1}$) to a supreme set in path $P_j$ by \Cref{lem:deg-external-new-edges}, the supreme sets are contained in $\super(R_m)$, and the total value is $\sum_{W\in B(R_m)}\rdem(W)$. The 
fourth equality is because $R_m$ is critical.
\end{proof}

\begin{lemma}
The running time of degree-constrained external augmentation algorithm is $\tO(F(m,n))$ for all heavy paths.
\end{lemma}
\begin{proof}
If two heavy paths $P_1, P_2$ have the same depth, then they are disjoint, and all supreme sets of the nodes on $P_1$ are disjoint from all supreme sets on $P_2$. By \Cref{lem:traversal2-supreme}, we computed the supreme sets in the perturbed graph $\tilde{G}$. The supreme sets in $\tilde{G}$ form a laminar family, so the heavy paths $P_\ell$ of the same depth have disjoint $\tilde{\super}(R_{k_\ell}^{(\ell)})$. We contract $V\setminus \tilde{\super}(R_{k_\ell}^{(\ell)})$ when constructing $H_\ell$, so all max flow calls for heavy paths of the same depth can run in parallel with total graph size $O(m)$. The maximum depth in heavy-light decomposition is $O(\log n)$, so the total running time is $O(F(m,n)\log n)$.
\end{proof}

\subsection{Degree-Constrained Augmentation}
\label{sec:deg-splitting}
Given an external augmentation solution, we split it off to get new edges on $V$ such that the Steiner connectivity is preserved to be at least $\tau$. They form a degree-constrained augmentation solution. Equivalently, we can first remove all external edges, use the degrees in the external edges to be degree constraints, and then add new edges under these degree constraints. Note that these degree constraints generated by external augmentation solution are tighter then the original degree constraints. With some abuse of notation, we forget about the original degree constraints and refer to these tight degree constraints as $\beta(\cdot)$.


We slightly modify the definition of augmentation chains in degree-constrained setting.
List all maximal supreme sets $Q=(X_1,\ldots,X_r)$ with demands at least 2, such that $X_1$ and $X_r$ have the minimum cut value. (Supreme sets are defined in the graph before adding chain.)
Define an augmentation chain as a set of edges $F=\{(a_i,b_{i+1})\}_{1\le i\le r-1}, a_i, b_i\in X_i$,
such that adding these edges satisfies the degree constraint. In contrast to the unconstrained setting, now we allow the endpoints to be non-terminals.

\subsubsection{Property of Augmentation Chains}

We prove in \Cref{lem:deg-chain-no-new-extreme} that adding an augmentation chain does not create new extreme sets, even if we allow nonterminal endpoints. This is a stronger version of \Cref{thm:chain-no-new-extreme}. It follows that splitting off by augmentation chains preserves the Steiner min cut value to be at least $\tau$.

\begin{fact}
If $Q$ is nonempty, then $r\ge 2$, and $X_1$ and $X_r$ are Steiner min cuts.
\end{fact}
\begin{proof}
Let $(S_1, S_2)$ be the two sides of a Steiner min cut and $\lambda$ be the Steiner min cut value. If $S_1$ is extreme, let $U_1=S_1$, otherwise let $U_1$ be the extreme violator of $S_1$ in \Cref{fact:extreme-violator}. Then $U_1$ is extreme and $d(U_1)\le d(S_1)=\lambda$. Also $d(U_1)\ge \lambda$ because $U_1$ is a Steiner cut. So $d(U_1)=\lambda$. If an extreme set $U'\supsetneqq U_1$, then $d(U')<d(U_1)$, which is impossible. So $U_1$ is a maximal extreme set. Applying the same argument to $S_2$ gives another maximal extreme set $U_2$ the is a Steiner min cut.

There exists at least two maximal supreme sets that are Steiner min cuts. When $Q$ is nonempty, they have minimum cut value $\lambda\le \tau-2$ among supreme sets in $Q$. So we choose $X_1$ and $X_r$ to be Steiner min cuts.
\end{proof}

\begin{lemma}\label{lem:deg-chain-intersection-Steiner}
Suppose $U$ is a new extreme set after adding an augmentation chain $F$ and $U$ crosses some $X_i\in Q$. If $d_F(X_i)-d_F(X_i\setminus U)\le 1$, then $\rho(U\cap X_i)\ne\emptyset$.
\end{lemma}
\begin{proof}
Assume for contradiction that $\rho(U\cap X_i)=\emptyset$, then $\rho(U\setminus X_i)=\rho(U)\ne\emptyset, \rho(X_i\setminus U)=\rho(X_i)\ne\emptyset$. Also $U\cap X_i\ne\emptyset$ because $U$ and $X_i$ cross, and the two difference sets are proper subsets of $U$ and $X_i$ respectively. Then because $U$ is extreme after adding $F$, and $X_i$ is extreme,
\begin{equation}\label{eq:deg-chain-intersection-Steiner-1}
d(U\setminus X_i)+d_F(U\setminus X_i)>d(U)+d_F(U)
\end{equation}
\begin{equation}\label{eq:deg-chain-intersection-Steiner-2}
d(X_i\setminus U)>d(X_i)
\end{equation}
By posi-modularity (\ref{eq:submodularity-2}),
\[d(U\setminus X_i)+d(X_i\setminus U)\le d(X_i)+d(U)\]
\begin{equation}\label{eq:deg-chain-intersection-Steiner-3}
 d_F(U\setminus X_i)+d_F(X_i\setminus U)\le d_F(X_i)+d_F(U)
\end{equation}
(\ref{eq:deg-chain-intersection-Steiner-3}) implies $d_F(U\setminus X_i)-d_F(U)\le d_F(X_i)-d_F(X_i\setminus U)\le 1$. So (\ref{eq:deg-chain-intersection-Steiner-1}) implies
\begin{equation}\label{eq:deg-chain-intersection-Steiner-4}
d(U\setminus X_i)\ge d(U)
\end{equation}
Together, (\ref{eq:deg-chain-intersection-Steiner-2})~and~(\ref{eq:deg-chain-intersection-Steiner-4}) violate posi-modularity (\ref{eq:submodularity-2}).
\end{proof}

\begin{lemma}\label{lem:deg-chain-mincut}
If $U$ is a new extreme set created by an augmentation chain $F$, then $U\cap X_1=U\cap X_r=\emptyset$.
\end{lemma}
\begin{proof}
$X_1$ is a Steiner min cut with one new edge endpoint $a_1$. Let $\lambda=d(X_1)$ be the Steiner min cut value.
Assume for contradiction that $U\cap X_1\ne\emptyset$.
$U\ne X_1$ because $U$ is not extreme before augmentation.

First we rule out $U\supsetneqq X_1$. If that happens, because $U$ is extreme after augmentation,
$$d(X_1)+d_F(X_1)>d(U)+d_F(U)$$
$$\lambda+1=d(X_1)+d_F(X_1)\ge d(U)+d_F(U)+1\ge \lambda+1$$
Therefore $U$ is a Steiner min cut and $d_F(U)=0$. A Steiner min cut cannot cross any supreme set by \Cref{lem:mincut-uncross-extreme}. $U$ also does not cross any edge in $F$ because $d_F(U)=0$, so $U$ contains all supreme sets in $Q$. But $U$ is a Steiner cut, so there is a terminal outside $U$. $U$'s demand is at least 2 because $U$ is a Steiner min cut and $F\ne\emptyset$, so $V\setminus U$ has demand at least 2, and there exists a supreme set outside $U$ with demand at least 2, which contradicts the definition of $Q$.

Next we rule out the case that $U$ and $X_1$ cross.
$d_F(X_1)-d_F(X_1\setminus U)\le d_F(X_1)=1$, so $U\cap X_1$ is a Steiner cut by \Cref{lem:deg-chain-intersection-Steiner}.
Because $U$ is extreme after augmentation, \[d(X_1\cap U)+d_F(X_1\cap U)>d(U)+d_F(U).\] Notice that $d_F(X_1\cap U)\le d_F(X_1)=1$. So $d(X_1\cap U)\ge d(U)+d_F(U)$. By submodularity (\ref{eq:submodularity-1}), $d(X_1\cup U)\le d(X_1)-d_F(U)$. If $X_1\cup U$ is a Steiner cut, then it is a Steiner min cut and $d_F(U)=0$, so we can use the same argument as the former case. If $X_1\cup U$ is not a Steiner cut, it must be $\rho(X_1\cup U)=T$. Because $X_1$ and $U$ are Steiner cuts, $\rho(X_1\setminus U)\ne\emptyset, \rho(U\setminus X_1)\ne\emptyset$. Because $X_1$ is extreme, $d(X_1\setminus U)>d(X_1)$. Because $U$ is extreme after adding $F$, $d(U\setminus X_1)+d_F(U\setminus X_1)>d(U)+d_F(U)$. Along with $d_F(U\setminus X_1)\le d_F(U)+d_F(X_1)=d_F(U)+1$, we obtain $d(U\setminus X_1)\ge d(U)$. Adding this inequality with $d(X_1\setminus U)>d(X_1)$ contradicts posi-modularity (\ref{eq:submodularity-2}).

Because $U$ is not contained in any extreme set, the only remaining case is that $U\cap X_1=\emptyset$. $U\cap X_r=\emptyset$ follows the same argument.
\end{proof}

\begin{lemma}\label{lem:deg-chain-no-new-extreme}
In degree constraint setting, adding an augmentation chain $F$ does not create any new extreme set.
\end{lemma}
\begin{proof}
In this proof, the cut value $d(\cdot)$ and supreme sets are defined in the graph before adding the augmentation chain $F$. Let $d_F(\cdot)$ be the cut value on the edges of $F$. Then after adding $F$, the cut value of $U$ is $d(U)+d_F(U)$.

Assume for contradiction that $U$ is a new extreme set. Because $U$ is not extreme before augmentation, let $W\subsetneqq U$ be its extreme violator in \Cref{fact:extreme-violator} with $d(W)\le d(U)$. Because $U$ is extreme after adding $F$, $d(W)+d_F(W)>d(U)+d_F(U)$.
Noticing that we add at most 2 edges to any extreme set, the above inequalities imply $$0\le d_F(U)<d_F(W)\le 2.$$ 
$d_F(W)>d_F(U)$ implies that there exists a new edge crossing $W$ but not crossing $U$. Suppose $(a_i, b_{i+1})\in F$ is such an edge, then $a_i, b_{i+1}\in U$. So $U$ intersects two maximal supreme sets $X_i, X_{i+1}$, and $U$ is not contained in any extreme set.

Let $Q_U$ be the sub-family of $Q$ that intersects $\rho(U)$. Because $d_F(W)>0$, $W$ is contained in a maximal supreme set $X_w$ in $Q$. Because $X_w\cap U\supseteq W$ is a Steiner cut, $X_w\in Q_U$, and $Q_U$ is nonempty.
Let $X_t=\arg\min_{Y\in Q_U}d(Y)$ be the set with minimum cut value in $Q_U$. Because $X_t\in Q_U$, $U\cap X_t$ is a Steiner cut.
Because $U$ is extreme after augmentation and not contained in $X_t$,
\[d(X_t\cap U)+d_F(X_t\cap U)>d(U)+d_F(U)\]
Using submodularity $d(X_t\cap U)+d(X_t\cup U)\le d(U)+d(X_t)$,
\begin{equation}\label{eq:deg-chain-no-new-extreme-1}
 d(U\cup X_t)< d(X_t)+d_F(X_t\cap U)-d_F(U)
\end{equation}

By \Cref{lem:deg-chain-mincut}, $1<t<r$. 
Divide $F$ into two halves $F_1=\{(a_i, b_{i+1})\}_{1\le i\le t-1}$ and $F_2=\{(a_i, b_{i+1})\}_{t\le i\le r-1}$. Then the extra term $d_F(X_t\cap U)-d_F(U)$ is
\[d_F(X_t\cap U)-d_F(U)=1[b_t\in U]+1[a_t\in U]-d_{F_1}(U)-d_{F_2}(U)\]
where $1[\cdot]$ is the indicator function that takes 1 if the event happens and 0 otherwise. Let $\theta_1=1[b_t\in U]-d_{F_1}(U)$, $\theta_2=1[a_t\in U]-d_{F_2}(U)$, so that $\theta_1,\theta_2\le 1$. Let $U_1=U\cup X_t$. (\ref{eq:deg-chain-no-new-extreme-1}) can be rewritten as $$d(U_1)<d(X_t)+\theta_1+\theta_2.$$

When $\theta_1\le 0$, let $U_2=U_1$. When $\theta_1=1$, we find a partner $X_p$ below, and let $U_2=U_1\cup X_p$. Because $\theta_1=1$, $b_t\in U$ and $d_{F_1}(U)=0$. Then for each $1\le i\le t-1$, either $a_i, b_{i+1}\in U$ or $a_i, b_{i+1}\notin U$. 
The two ends are $a_1\notin U$ by \Cref{lem:deg-chain-mincut} and $b_t\in U$, so there exists some $2\le p\le t-1$ such that $b_p\notin U, a_p\in U$. 
Then $d_F(X_p\setminus U)=1, d_F(X_p)=2$, so $U\cap X_p$ is a Steiner cut by \Cref{lem:deg-chain-intersection-Steiner} and $X_p\in Q_U$. 
Because $X_p$ is extreme, \[d(U\cap X_p)>d(X_p).\]
Applying submodularity on $U_1$ and $X_p$ ($U_1\cap X_p=U\cap X_p$  because $X_t$ is disjoint from $X_p$),
\[d(U_1\cup X_p)<d(U_1)\]
In both cases ($\theta_1\le0$ and $\theta_1=1$), we have
$$d(U_2)\le d(U_1)-\theta_1<d(X_t)+\theta_2.$$

The other direction is symmetric, except we start with $U_2$, not $U_1$. When $\theta_2\le 0$, let $U_3=U_2$. When $\theta_2=1$, we can find a partner $X_q$ such that $t\le q\le r-1$ and $b_q\in U$ but $a_q\notin U$. Let $U_3=U_2\cup X_q$. Using the same argument, we have
\begin{equation}\label{eq:deg-chain-no-new-extreme-2}
    d(U_3)\le d(U_2)-\theta_2<d(X_t)
\end{equation}

$U_3$ is the union of $U, X_t$ and $X_p, X_q$ if they exist. Recall that $X_t$ is the set with minimum cut value in $Q_U$, and $X_p, X_q\in Q_U$ if they exist. $U_3$ is a Steiner cut because $X_1$ is outside by \Cref{lem:deg-chain-mincut}.
$U_3$ is not extreme because it contains $U$, which is not contained in any extreme set. So there exists an extreme violator $W'$ of $U_3$ with $d(W')\le d(U_3)$ by \Cref{fact:extreme-violator}. Because $W'$ is extreme, it is contained in a  maximal supreme set $X_{w'}$. Because $X_{w'}\cap U_3\supseteq W'$ is a Steiner cut, either $X_{w'}\cap U$ is a Steiner cut, or $X_{w'}$ is one of $X_t, X_p, X_q$ if they exist. So $X_{w'}\in Q_U$, and $d(X_{w'})\ge d(X_t)$. Because $W'\subseteq X_{w'}$ and $X_{w'}$ is extreme, $d(W')\ge d(X_{w'})$. In conclusion
$$d(U_3)\ge d(W')\ge d(X_{w'})\ge d(X_t)$$
which contradicts (\ref{eq:deg-chain-no-new-extreme-2}).
\end{proof}

The algorithm is to simply keep adding augmentation chains. By \Cref{lem:deg-split-chain-feasible}, the partial solution is always feasible. Moreover, because we 
split off the optimal external augmentation solution, the partial solution is contained in an optimal solution. The algorithm stops when there is no supreme sets with demand at least 2, or equivalently, when the Steiner connectivity reaches at least $\tau-1$. This situation is solved in \Cref{sec:deg-augment-1}.

Let $G_\ell$ be the graph after adding $\ell$ chains, and $G^{ext}_\ell$ be the external graph after splitting off $\ell$ chains. (The definition of $G_\ell$ differs from \Cref{sec:deg-external}.) Initially $G_0=G$ is the original graph, and $G^{ext}_0=(V\cup\{x\}, E\cup F^{ext})$ is the original graph after adding the external augmentation solution. Let $\beta_\ell(u)$ be the remaining external degrees of $u$ in $G_\ell$, or equivalently the capacity of $(u,x)$ edge in $G^{ext}_\ell$. Let $\beta_\ell(S) = \sum_{u\in S} \beta_\ell(u)$. Then for any vertex set $S$, $d_{G^{ext}_\ell}(S)=d_{G_\ell}(S)+\beta_\ell(S)$. 

\begin{lemma}\label{lem:deg-split-chain-feasible}
The external graph $G^{ext}_\ell$ always has Steiner connectivity at least $\tau$.
\end{lemma}
\begin{proof}
We prove by induction on $\ell$ that the Steiner connectivity in $G^{ext}_\ell$ is at least $\tau$. The base case $\ell=0$ holds because the external augmentation solution is feasible by \Cref{lem:deg-external-optimal}.

Assume the statement holds up to $\ell$. Consider a minimal Steiner min cut $X$ in $G^{ext}_{\ell+1}$, taking the side $x\notin X$. By inductive hypothesis, $d_{G^{ext}_\ell}(X)\ge \tau$. If no edges of chain $\ell+1$ are inside $X$, then the splitting off does not decrease the cut value of $X$, and $d_{G^{ext}_{\ell+1}}(X)\ge d_{G^{ext}_\ell}(X)\ge \tau$. The remaining case is that there exists a chain edge inside $X$. A chain edge connects two maximal supreme sets (defined in $G_\ell$), so $X$ is not extreme in $G_\ell$ and hence not extreme in $G_{\ell+1}$ by \Cref{lem:deg-chain-no-new-extreme}, and there exists an extreme violator $W$ of $X$ in $G_{\ell+1}$, $d_{G_{\ell+1}}(W)\le d_{G_{\ell+1}}(X)$. Because $W\subsetneqq X$, $\beta_{\ell+1}(W)\le \beta_{\ell+1}(X)$. So $$d_{G^{ext}_{\ell+1}}(W)=d_{G_{\ell+1}}(W)+\beta_{\ell+1}(W)\le d_{G_{\ell+1}}(X)+\beta_{\ell+1}(X)= d_{G^{ext}_{\ell+1}}(X)$$ which contradicts the assumption that $X$ is a minimal Steiner min cut in $G^{ext}_{\ell+1}$. So this case is impossible, and the Steiner connectivity is at least $\tau$.
\end{proof}

\subsubsection{Finding Augmentation Chains}
Next we describe the detailed steps to find an augmentation chain efficiently. We first establish the algorithm, with a slight modification in the heavy-light decomposition of external augmentation, then prove it always finds augmentation chains in \Cref{lem:deg-find-chains}.

Let $L_2$ be the collection of $R\in L$ such that $R$ is critical and all ancestors of $R$ are not critical. Let $L_{high}$ be the union of paths from nodes in $L_2$ to the roots, that is a sub-forest of $L$ with leaves $L_2$. In the external augmentation algorithm, modify the heavy-light decomposition step such that the heavy paths stop at nodes in $L_2$. More precisely, apply heavy-light decomposition on $L_{high}$, then separately apply heavy-light decomposition on each subtree rooted at a node in $L_2$, and then merge the heavy-light decompositions. This way, a leaf to root path still passes $O(\log n)$ heavy paths, so all proofs in \Cref{sec:deg-external} work for the new definition. But now, each heavy path stops at a leaf or a node in $L_2$, which is critical. We say that the heavy-light decomposition on $L_{high}$ is \emph{upper level}, and the ones on the subtrees rooted at $L_2$ are \emph{lower level}.

We maintain the forest $L_{high}$. For lower level heavy paths, regard them as contracted into their ancestors in $L_2$.
For each $R\in L_{high}$, maintain the cut value $c(R)$, which is initialized to $d(\super(R))$. 
We assign an external edge to a heavy path if the edge is added by the external augmentation algorithm when processing the heavy path. In each iteration, we list all roots of $L_{high}$, and pick the roots $R_i$ with $c(R_i)\le \tau-2$ to form a list $Q_T=(R_1,\ldots, R_r)$, such that $c(R_1)$ and $c(R_r)$ are the minimum two in $Q_T$. Pick a chain $F=\{(a_i, b_{i+1})\}_{1\le i\le r-1}$, such that $a_i, b_i$ are vacant vertices (having unused external edges) assigned to heavy paths in the tree rooted at $R_i$. (When $a_i=b_i$, it needs two vacant degrees.) We use the external edges assigned to lower level paths first. The external edges assigned to upper level paths are used only when all lower level degrees in the tree are used.

After picking a chain, we add it to the solution, and update the vacant degrees. 
Next we update $c(R)$ values. For each endpoint $u$ in the chain, define its \emph{representative leaf} $U\in L_2$ as follows. If $u$ is assigned to an upper level path, let $U$ be the leaf of the heavy path to which $u$ is assigned. If $u$ is assigned to a lower level path, let $U$ be the node to which that path is contracted. Add 1 to the values $c(R)$ for each $R$ on the path from $U$ to the root.
Next we update the structure of $L_{high}$. If some node $R\in L_{high}\setminus L_2$ and its child $W$ have $c(R)\ge c(W)$, delete $R$ in $L_{high}$ and assign its children to its parent (when $R$ is a root, its children become new roots). 

For efficiency, we use the same chain edges until they become invalid. This occurs when either some endpoint uses up its vacant degree, or adding the chain causes some forest node $R$ to be deleted from $Q_T$. This second event occurs either when the supreme set corresponding to that forest node $R$ has demand at most 1 or when one of its children $W$ satisfies $c(R)\geq c(W)$. The $c(R)$ values are maintained by a dynamic tree data structure. We can use priority queues to maintain the expiration time of each endpoint and each forest node. This speedup is identical to \cite{CenLP22a} and discussed in \Cref{sec:lazy}. 

\begin{theorem}\label{thm: runtime-deg-constr-aug-chain}
Given the supreme sets forest and the cut values of the supreme sets, we can add all augmentation chains so that all supreme sets have demand at most 1 in time $\tO(n)$. 
\end{theorem}

The rest of this section establishes correctness by proving the chain found by the algorithm is an augmentation chain. 
One difficulty of the degree-constrained setting is that the chain edges may connect nonterminals, so for a fixed terminal set $R$, the supreme set whose projection is $R$ may change as we add augmentation chains. This does not happen in the unconstrained setting where we only add edges to terminals.

In the proofs, $\super(\cdot)$ is the supreme set defined in the original graph, and $\super_\ell(\cdot)$ is the supreme set defined in $G_\ell$. That is, for any $R\subseteq T$, if $R$ is the terminal projection of an extreme set in $G_\ell$, then $\super_\ell(R)$ is defined to be the supreme set of terminals $R$ in $G_\ell$; otherwise, $\super_\ell(R)$ is undefined. In particular, $\super_0(\cdot)=\super(\cdot)$ because $G_0=G$. Recall that $\beta_\ell(u)$ is the remaining external degree after adding $\ell$ augmentation chain, so the cut value in $G^{ext}_\ell$ can be expressed as $d_{G^{ext}_\ell}(U)=d_{G_\ell}(U)+\beta_\ell(U)$, where $\beta_\ell(U)=\sum_{u\in U}\beta_\ell(u)$.

\begin{lemma}\label{lem:deg-maintained-cut-value}
For $R\in L_{high}$, the $c(R)$ value maintained by the algorithm is equal to $d_{G_\ell}(\super(R))$ at every time $\ell$.
\end{lemma}
\begin{proof}
Consider the contribution of an edge $(u, w)$ in the $\ell$-th augmentation chain. Let $U$ be the representative leaf of $u$. If $u$ is assigned to a upper level heavy path, $U$ is the leaf of the path, and by \Cref{lem:deg-external-new-edges} $u\in \super(U)$. If $u$ is assigned to a lower level heavy path, $U$ is the parent of the path, and we also have $u\in \super(R_m)\subseteq \super(U)$ for some node $R_m$ in the path. Therefore $u$ is in $\super(R)$ for every $R$ on the path from $U$ to the root. Also these supreme sets do not contain $w$, because $w$ is assigned to a different tree from $u$, and $(u,w)$ connects two maximal supreme sets. 
In the algorithm, $c(R)$ is increased for $R$ on the path from $U$ to the root, which correctly maintains $d_{G_\ell}(\mu(R))$.
\end{proof}

\begin{fact}\label{fact:supreme-shrink-after-chain}
If $\super_i(R)$ is defined, then for any $0\le j<i$, $\super_j(R)$ exists and $\super_i(R)\subseteq \super_j(R)$.
\end{fact}
\begin{proof}
$\super_j(R)$ is extreme in $G_i$, so it is also extreme in $G_j$ by repeatedly applying \Cref{lem:deg-chain-no-new-extreme}. The statement then follows the definition of $\super_j(R)$.
\end{proof}

\begin{lemma}\label{lem:L2-tight}
For any $R\in L_2$, $d_{G^{ext}_\ell}(\super(R))=\tau$ holds for every $\ell$.
\end{lemma}
\begin{proof}
By \Cref{lem:deg-external-optimal}, the external augmentation solution is feasible, so for each $R\in L_2$, $\beta(\super(R))\ge \tau-d(\super(R))$. Moreover, \Cref{lem:deg-external-optimal} says the total external degree is $\sum_{R\in L_1}\rdem(R)$, which equals $\sum_{R\in L_2}\tau-d(\super(R))$ because all internal nodes in $L_{high}\setminus L_2$ are not critical. Therefore $\beta(\super(R))=\tau-d(\super(R))$ for each $R\in L_2$. 

Initially, $d_{G^{ext}_0}(\super(R))=\beta(\super(R))+d(\super(R))=\tau$.
$G^{ext}_{\ell+1}$ is formed by splitting off edges for an augmentation chain from $G^{ext}_\ell$, so $d_{G^{ext}_{\ell+1}}(\super(R))\le d_{G^{ext}_\ell}(\super(R))$, and $d_{G^{ext}_\ell}(\super(R))\le \tau$ for each $\ell$ by induction. Also $d_{G^{ext}_\ell}(\super(R))\ge \tau$ by \Cref{lem:deg-split-chain-feasible}, so the statement holds.
\end{proof}

\begin{lemma}\label{lem:focus-on-Lhigh}
For any $R\in L_2$, if $\super(R)$ has vacant degree after adding the $\ell$-th chain, then $\super_\ell(R)$ is defined, and $d_{G_\ell}(\super(R))=d_{G_\ell}(\super_\ell(R))$.
\end{lemma}
\begin{proof}

Fix some $R\in L_2$. Assume first that $\super_\ell(R)$ is defined. By \Cref{fact:supreme-shrink-after-chain}, $\super_\ell(R)\subseteq \super(R)$, so
\begin{equation}\label{eq:focus-on-Lhigh-3}
    \beta_\ell(\super_\ell(R))\le \beta_\ell(\super(R))
\end{equation}
By \Cref{lem:L2-tight},
$\tau=d_{G^{ext}_\ell}(\super(R))=\beta_\ell(\super(R))+d_{G_\ell}(\super(R))$. Combined with (\ref{eq:focus-on-Lhigh-3}), we have
\begin{equation}\label{eq:focus-on-Lhigh-4}
    \beta_\ell(\super_\ell(R))\le \tau-d_{G_\ell}(\super(R))
\end{equation}
By \Cref{lem:deg-split-chain-feasible}, $d_{G^{ext}_\ell}(\super_\ell(R))\ge \tau$. Combined with (\ref{eq:focus-on-Lhigh-4}), we have
$$\tau\le d_{G^{ext}_\ell}(\super_\ell(R))=\beta_\ell(\super_\ell(R))+d_{G_\ell}(\super_\ell(R))\le \tau-d_{G_\ell}(\super(R))+d_{G_\ell}(\super_\ell(R))$$
So $d_{G_\ell}(\super(R))\le d_{G_\ell}(\super_\ell(R))$. By \Cref{lem:supreme-mincut}, $\super_\ell(R)$ is an $R$-$(T\setminus R)$ min cut in $G_\ell$, so
$d_{G_\ell}(\super(R))= d_{G_\ell}(\super_\ell(R))$. 

It remains to prove that $\super_\ell(R)$ is defined.
We prove by induction on $\ell$. For the base case $\ell=0$, $\super(R)$ is defined because $R\in L$. Next consider the inductive case $\ell>0$. 

Assume for contradiction that there is no extreme set with terminal projection $R$ in $G_\ell$.
Let $X$ be an extreme violator of $\super(R)$ in $G_\ell$; then $X\subsetneqq \super(R)$ and \begin{equation}\label{eq:focus-on-Lhigh-1}
    d_{G_\ell}(X)\le d_{G_\ell}(\super(R))
\end{equation}
Because $X$ is extreme, $\rho(X)\ne R$, and $\rho(X)\subsetneqq R$ is a node in the subtree of $R$ in $L$. In the proof of \Cref{lem:deg-external-optimal}, we proved that for an extreme set $X$ with terminal set $\rho(X)\in P_j$ for some heavy path $P_j$, the cut value of $X$ is augmented to $\tau$ just after processing heavy path $P_j$. Because $X$ is in the subtree of $R$, and we stop the heavy paths at nodes in $L_2$, external edges assigned to $R$ are added after path $P_j$. Let $\beta'(\cdot)$ be the external degree in the external edges added for heavy paths up to $P_j$, and $\beta'_\ell(\cdot)$ be the remaining $\beta'$ degree in $G_\ell$, so that 
\begin{equation}\label{eq:focus-on-Lhigh-2}
    \beta'_\ell(X)+d_{G_\ell}(X)\ge \tau
\end{equation}
The algorithm uses lower level degrees first, and uses upper level external edges assigned to $R$ only after all lower level degrees in the subtree of $R$ are used. As long as $\super(R)$ has vacant degree, we have $\beta_\ell(\super(R))>\beta'_\ell(\super(R))\ge \beta'_\ell(X)$. Combined with \Cref{lem:L2-tight} and (\ref{eq:focus-on-Lhigh-1}),
$$\tau=\beta_\ell(\super(R))+d_{G_\ell}(\super(R)) >\beta'_\ell(X)+d_{G_\ell}(X),$$
which contradicts (\ref{eq:focus-on-Lhigh-2}).
\end{proof}

\begin{lemma}\label{lem:deg-vacant-in-current-supreme}
After adding the $\ell$-th chain, let $u$ be a vertex with vacant degree and $U\in L_2$ be its representative leaf, then $u\in \super_\ell(U)$.
\end{lemma}
\begin{proof}
$u\in \super(U)$ by \Cref{lem:deg-external-new-edges} and the definition of representative leaf. $\super(U)$ has a vacant degree at $u$ in $G_\ell$, so $\super_\ell(U)$ is defined by \Cref{lem:focus-on-Lhigh}.  Assume for contradiction that $u\notin \super_\ell(U)$. Then by \Cref{fact:supreme-shrink-after-chain}, $\super_\ell(U)\subsetneqq \super(U)$. Because one of the vacant degrees in $\beta_\ell(\super(U))$ is at $u\notin \super_\ell(U)$, $\beta_\ell(\super(U))>\beta_\ell(\super_\ell(U))$.

$\super(U)$ is not extreme in $G_\ell$, so there exists an extreme violator $W$. $\rho(W)\subseteq U$, so $W\subseteq \super_\ell(U)$, and $d_{G_\ell}(\super(U))\ge  d_{G_\ell}(W)\ge d_{G_\ell}(\super_\ell(U))$. 
By \Cref{lem:L2-tight}, $d_{G_\ell}(\super(U))+\beta_\ell(\super(U))=\tau$.
Therefore $$d_{G^{ext}_\ell}(\super_\ell(U))=  d_{G_\ell}(\super_\ell(U))+\beta_\ell(\super_\ell(U))<d_{G_\ell}(\super(U))+\beta_\ell(\super(U))=\tau,$$ which contradicts \Cref{lem:deg-split-chain-feasible}.
\end{proof}

\begin{lemma}\label{lem:supreme-equal-after-chain}
If $\super_\ell(R)$ is defined for some $R\in L_{high}\setminus L_2$, then $d_{G_\ell}(\super_\ell(R))= d_{G_\ell}(\super_{\ell-1}(R))$.
\end{lemma}
\begin{proof}
If $\super_\ell(R)=\super_{\ell-1}(R)$, the statement is trivial. By \Cref{fact:supreme-shrink-after-chain}, the remaining case is $\super_\ell(R)\subsetneqq \super_{\ell-1}(R)$.
$d_{G_{\ell-1}}(\super_\ell(R))> d_{G_{\ell-1}}(\super_{\ell-1}(R))$ since $\super_{\ell-1}(R)$ is extreme in $G_{\ell-1}$, but $d_{G_\ell}(\super_\ell(R))\le  d_{G_\ell}(\super_{\ell-1}(R))$ by \Cref{lem:supreme-mincut}. 
It remains to prove that $d_{G_\ell}(\super_\ell(R))\ge  d_{G_\ell}(\super_{\ell-1}(R))$.
An augmentation chain adds at most 2 edges to any supreme set. 
So the only nontrivial case is that the $\ell$-th chain adds 2 edges to $\super_{\ell-1}(R)\setminus \super_\ell(R)$ and no edge to $\super_{\ell}(R)$.

For each endpoint $u$ of the $\ell$-th chain, $u\in \super_{\ell-1}(U)$ by \Cref{lem:deg-vacant-in-current-supreme}, where $U\in L_2$ is $u$'s representative leaf. Assume the two edges are added to $X_1=\super_{\ell-1}(U_1)$ and $X_2=\super_{\ell-1}(U_2)$, $U_1, U_2\in L_2$. $U_1, U_2\subsetneqq R$ because $R\notin L_2$ and the two edges are added to $\super_{\ell-1}(R)$. 
Also the two endpoints are not in $\super_\ell(R)$, so $\super_{\ell}(R)$ crosses $X_1$ and $X_2$.

Case 1: $U_1\ne U_2$. Since $X_1$ and $X_2$ are extreme in $G_{\ell-1}$, 
\[d_{G_{\ell-1}}(X_1\cap \super_{\ell}(R))> d_{G_{\ell-1}}(X_1), d_{G_{\ell-1}}(X_2\cap \super_{\ell}(R))> d_{G_{\ell-1}}(X_2)\]
By submodularity (\ref{eq:submodularity-4}), $d_{G_{\ell-1}}(\super_{\ell}(R)\cup X_1)< d_{G_{\ell-1}}(\super_{\ell}(R))$, $d_{G_{\ell-1}}(\super_{\ell}(R)\cup X_1\cup X_2)<d_{G_{\ell-1}}(\super_{\ell}(R)\cup X_1)$ ($X_1$ and $X_2$ are disjoint by laminarity). So \[d_{G_\ell}(\super_\ell(R))=d_{G_{\ell-1}}(\super_{\ell}(R))\ge d_{G_{\ell-1}}(\super_{\ell}(R)\cup X_1\cup X_2)+2\ge d_{G_{\ell-1}}(\super_{\ell-1}(R))+2=d_{G_\ell}(\super_{\ell-1}(R))\]
where the equations use the assumption that the $\ell$-th chain adds 2 edges to $\super_{\ell-1}(R)$ but no edge to $\super_{\ell}(R)$.

Case 2: $U_1=U_2$. By \Cref{lem:focus-on-Lhigh}, either $\super_\ell(U_1)$ exists and $d_{G_\ell}(\super_\ell(U_1))=d_{G_\ell}(\super(U_1))$, or $\super(U_1)$ has no vacant degree at time $\ell$. Notice that $\rho(X_1\cap \super_\ell(R))=U_1$. In the former case, 
\[d_{G_{\ell-1}}(X_1\cap \super_\ell(R)) =d_{G_\ell}(X_1\cap \super_\ell(R))\ge  d_{G_\ell}(\super_\ell(U_1))=d_{G_\ell}(\super(U_1))=d_{G_{\ell-1}}(\super(U_1))+2=d_{G_{\ell-1}}(X_1)+2\]
In the latter case, $\super(U_1)$ has no vacant degree in $G_\ell$, so $d_{G_{\ell-1}}(X_1\cap \super_\ell(R)) =d_{G_\ell}(X_1\cap \super_\ell(R))\ge \tau$.
By \Cref{lem:L2-tight}, $\tau=d_{G_{\ell-1}}(\super(U_1))+\beta_{\ell-1}(\super(U_1))\ge d_{G_{\ell-1}}(\super(U_1))+2=d_{G_{\ell-1}}(X_1)+2$, so we also have $d_{G_{\ell-1}}(X_1\cap \super_\ell(R))\ge d_{G_{\ell-1}}(X_1)+2$.
By submodularity,
\[d_{G_{\ell-1}}(\super_{\ell}(R)\cup X_1)+2\le d_{G_{\ell-1}}(\super_{\ell}(R))\]
$\super_\ell(R)\subseteq \super_{\ell-1}(R)$ by \Cref{fact:supreme-shrink-after-chain}, and $X_1=\super_{\ell-1}(U_1)\subseteq \super_{\ell-1}(R)$ by laminarity.
Therefore, since $\super_{\ell-1}(R)$ is extreme,
\[d_{G_\ell}(\super_\ell(R))=d_{G_{\ell-1}}(\super_{\ell}(R))\ge d_{G_{\ell-1}}(\super_{\ell}(R)\cup X_1)+2\ge d_{G_{\ell-1}}(\super_{\ell-1}(R))+2=d_{G_\ell}(\super_{\ell-1}(R))\]

In conclusion we have $d_{G_\ell}(\super_\ell(R))\ge d_{G_\ell}(\super_{\ell-1}(R))$ in both cases, so $d_{G_\ell}(\super_\ell(R))= d_{G_\ell}(\super_{\ell-1}(R))$.
\end{proof}

\begin{lemma}\label{lem:deg-current-supreme-defined}
After processing the $\ell$-th chain, if a node $R$ in $L_{high}$ is not deleted, then either $\super_\ell(R)$ is defined and $d_{G_\ell}(\super(R))=d_{G_\ell}(\super_\ell(R))$, or $R\in L_2$ and $\super(R)$ has no vacant degree in $G_\ell$.
\end{lemma}
\begin{proof}
We prove by induction first on time $\ell$, then bottom-up on the forest. The base case $\ell=0$ holds trivially, and the base case $R\in L_2$ is handled by \Cref{lem:focus-on-Lhigh}. Next consider the inductive case $\ell>0$ and $R$ is an internal node in $L_{high}\setminus L_2$.

By inductive hypothesis, $\super_{\ell-1}(R)$ is defined in $G_{\ell-1}$, and $d_{G_{\ell-1}}(\super(R))=d_{G_{\ell-1}}(\super_{\ell-1}(R))$. 
$\super(R)\supseteq \super_{\ell-1}(R)$ by \Cref{fact:supreme-shrink-after-chain}. If chain $\ell$ does not add edges to $\super(R)$, then $d_{G_\ell}(\super(R))=d_{G_\ell}(\super_{\ell-1}(R))$. Next assume chain $\ell$ adds edges to $\super(R)$.  Because $R$ is a forest node, for each chain edge $(u,w)$, at most one endpoint $u$ is assigned to a leaf $U$ in the subtree of $R$. $u\in \super_{\ell-1}(U)\subseteq\super_{\ell-1}(R)$ by \Cref{lem:deg-vacant-in-current-supreme} and laminarity of supreme sets in $G_{\ell-1}$. The other endpoint $w$ is assigned to leaves outside the tree of $R$, so $w\notin\super(R)$. Therefore we also have $d_{G_\ell}(\super(R))=d_{G_\ell}(\super_{\ell-1}(R))$ because they increase by the same amount from $G_{\ell-1}$. In both cases we have
$d_{G_\ell}(\super(R))=d_{G_\ell}(\super_{\ell-1}(R))$.
By \Cref{lem:supreme-equal-after-chain}, if $\super_\ell(R)$ is defined, then $d_{G_\ell}(\super(R))=d_{G_\ell}(\super_{\ell-1}(R))=d_{G_\ell}(\super_\ell(R))$. It remains to prove that $\super_\ell(R)$ is defined.

Assume for contradiction that there is no extreme set with terminal set $R$ in $G_\ell$.
Let $X\subsetneqq \super(R)$ be an extreme violator of $\super(R)$ in $G_\ell$ with
$d_{G_\ell}(X)\le d_{G_\ell}(\super(R))$.
Let $W=\rho(X)$, then $W\subsetneqq R$ because $X$ is extreme.
$d_{G_\ell}(\super(R))<d_{G_\ell}(\super(W))$ because $R$ is not deleted, so 
\[d_{G_\ell}(X)<d_{G_\ell}(\super(W))\]
By inductive assumption, either  $\super_\ell(W)$ exists and $d_{G_\ell}(\super_\ell(W))=d_{G_\ell}(\super(W))$, or $\super(W)$ has no vacant degree.
In the former case,  $d_{G_\ell}(X)<d_{G_\ell}(\super_\ell(W))$, which contradicts the fact that $\super_\ell(W)$ is extreme. 
In the latter case, $X\subseteq \super(W)$ because $X$ is extreme, so $X$ has no vacant degree. 
But by \Cref{lem:L2-tight},
\[\tau=d_{G^{ext}_{\ell}}(\super(W))\ge d_{G_\ell}(\super(W))>d_{G_\ell}(X)\]
So $X$ should have vacant degree in a feasible external solution, a contradiction.
\end{proof}

\begin{lemma}\label{lem:deg-delete-non-extreme}
If a node $R$ is deleted after processing the $\ell$-th chain, then $\super_\ell(R)$ is not defined in $G_\ell$.
\end{lemma}
\begin{proof}
The algorithm only deletes internal nodes of $L_{high}$, so $R\in L_{high}\setminus L_2$.
Because $R$ is not deleted at time $\ell-1$, by \Cref{lem:deg-current-supreme-defined}, $d_{\ell-1}(\super(R))=d_{\ell-1}(\super_{\ell-1}(R))$. If $(u,w)$ is a chain edge and $u\in \super(R)$, then $w$ is not assigned to the tree of $R$ and $w\notin \super(R)$. By \Cref{lem:deg-vacant-in-current-supreme} and laminarity of supreme sets in $G_{\ell-1}$, all endpoints in $\super(R)$ in the $\ell$-th chain are in $\super_{\ell-1}(R)$. Therefore the cut values of $\super(R)$ and $\super_{\ell-1}(R)$ increase by the same amount from $G_{\ell-1}$ to $G_\ell$, and $d_{G_\ell}(\super(R))=d_{G_\ell}(\super_{\ell-1}(R))$.

Assume for contradiction that $\super_\ell(R)$ is defined. Then by \Cref{lem:supreme-equal-after-chain}, $d_{G_\ell}(\super_{\ell-1}(R))=d_{G_\ell}(\super_{\ell}(R))$, so $d_{G_\ell}(\super(R))=d_{G_\ell}(\super_{\ell}(R))$.
Because $R$ is deleted at time $\ell$, $d_{G_\ell}(\super(R))\ge d_{G_\ell}(\super(W))$ for some $W\in ch(R)$. 
By \Cref{lem:supreme-mincut}, $\super_\ell(R)$ is an $R$-$(T\setminus R)$ min cut in $G_\ell$, so $$d_{G_\ell}(\super_\ell(R)\cup\super(W))\ge d_{G_\ell}(\super_\ell(R))=d_{G_\ell}(\super(R))\ge d_{G_\ell}(\super(W))$$
Since $\super_\ell(R)$ is extreme in $G_\ell$ ($\rho(\super_\ell(R)\cap\super(W))=W\subsetneqq R$), $$d_{G_\ell}(\super_\ell(R)\cap\super(W))>d_{G_\ell}(\super_\ell(R))$$ Adding these two inequalities contradicts submodularity (\ref{eq:submodularity-1}).
\end{proof}




\begin{lemma}\label{lem:deg-find-chains}
The chains picked by the algorithm are augmentation chains.
\end{lemma}
\begin{proof}
By \Cref{lem:deg-maintained-cut-value}, the $c(R)$ value maintained by the algorithm equals $d_{G_\ell}(\super(R))$ for all $R\in L_{high}$ at each time $\ell$.
By \Cref{lem:deg-current-supreme-defined} and \Cref{lem:deg-delete-non-extreme}, the algorithm keeps a node $R$ in $L_{high}$ if and only if $R$ remains a projection of an extreme set, or $R\in L_2$ and $\super(R)$ has no vacant degree.
For the latter case, $c(R)=d_{G_\ell}(\super(R))=\tau$ at time $\ell$ by \Cref{lem:L2-tight}.
Therefore the roots $R_i$ with $c(R_i)\le \tau-2$ picked by the algorithm are exactly the terminal sets of all maximal supreme sets with demand at least 2. 

The algorithm picks endpoints with vacant degrees. By \Cref{lem:deg-vacant-in-current-supreme}, the endpoints are in the current supreme sets of their representative leaves, which are in corresponding maximal supreme sets by laminarity. In conclusion the algorithm picks an augmentation chain.
\end{proof}

\subsection{Augment Connectivity by One in the Degree-Constrained Setting}
\label{sec:deg-augment-1}
Similar to the unconstrained degree case, we are left with the problem of augmenting Steiner connectivity by 1 (from $\tau-1$ to $\tau$) after augmentation chains have been added. As input to this problem, we are given the graph $G$ with Steiner connectivity $\tau-1$, degree constraint $\beta$, and the family $K$ of the terminal sets of all supreme sets.
We solve this problem by a reduction to the unconstrained setting and a modification of \Cref{alg:augment-1}.

The sets in $K$ with demand $1$ correspond to the terminal sets of minimal Steiner minimum cuts in $G$. These are also the nodes in the skeleton structure described in \Cref{sec:augment-1}. 
Let $K'$ be the sub-family of $K$ of sets with demand 1. That is, each $k\in K'$ is the terminal set of a minimal Steiner min cut. We do not know the supreme sets of an arbitrary terminal projection, but in this special case we can find the supreme sets efficiently for terminal sets in $K'$ using isolating cuts as follows. Given a graph with terminals, isolating cuts framework finds the minimum cuts separating each terminal from all other terminals in $O(\log n)$ max flow calls \cite{LiP20deterministic, AbboudKT21}.

Contract each set $k\in K'$, and also contract $T\setminus K'$. Run isolating cuts with these contracted nodes as terminals. By \Cref{lem:supreme-mincut}, for each $k\in K'$, $\super(k)$ is the earliest $k$-$(T\setminus k)$ min cut. Therefore, if we find minimal isolating cuts, we get the supreme sets $\super(k)$ for each $k\in K'$.


\begin{lemma}\label{lem:deg-aug-1-surrogate}
If the degree-constrained augmenting connectivity by 1 problem is feasible, then for each $k\in K'$, there exists $v_k\in \super(k)$ such that $\beta(v_k)\geq 1$.
\end{lemma}
\begin{proof}
Because the problem is feasible and $\super(k)$ is a Steiner cut, $d(\super(k))+\beta(\super(k))\ge \tau$, where $\beta(\super(k))=\sum_{u\in \super(k)}\beta(u)$. $d(\super(k))=\tau-1$ by definition of $K'$, so there exists a $v_k\in \super(k)$ with $\beta(v_k)\ge 1$.
\end{proof}

Recall that in \Cref{alg:augment-1}, we randomly match a pair of demand-1 extreme sets $\super(k), \super(k')$, and add a matching edge connecting arbitrary terminals in $k$ and $k'$. Consider modifying the algorithm in the following way. For each $k\in K'$, choose a surrogate $v_k\in \super(k)$ according to \Cref{lem:deg-aug-1-surrogate}. When matching a pair of extreme sets $\super(k)$ and $\super(k')$, instead of connecting arbitrary terminals, we add an edge $(v_k, v_{k'})$ connecting the surrogates. \Cref{lem:surrogate-feasible} shows that this new matching on surrogates is an optimal solution.

\begin{lemma}\label{lem:surrogate-feasible}
Given an instance of degree-constrained augmentation by 1 problem, remove the degree constraint and run \Cref{alg:augment-1} to get an output $F$. If we replace every edge in $F$ by its surrogate edge to form a new matching $F'$, then $F'$ is an optimal solution to the degree-constrained augmentation by 1 problem.
\end{lemma}
\begin{proof}
For any Steiner min cut $S$, we prove that it is augmented by $F'$, so that $F'$ is a feasible solution.
Because $S$ is augmented by $F$, there exists an edge $(u, v)\in F$ across $S$, for some $u\in \super(k), v\in \super(k')$, where $\super(k), \super(k')$ are demand-1 supreme sets. Then $k,k'\in K'$. By \Cref{lem:mincut-uncross-extreme}, $S$ does not cross $\super(k)$ and $\super(k')$. So for surrogates $v_k\in \super(k)$ and $v_{k'}\in \super(k')$, edge $(v_k, v_{k'})\in F'$ also augments $S$. The surrogates have vacant degree by \Cref{lem:deg-aug-1-surrogate}. Therefore $F'$ is a feasible solution.

Notice that removing the degree constraint cannot decrease the optimal value, because an optimal solution to degree-constrained augmentation by 1 problem is also feasible to the corresponding unconstrained problem. Because $F$ is an optimal solution to the unconstrained problem, and $F'$ has the same value as $F$, $F'$ is an optimal solution.
\end{proof}

Next we analyze the running time.
Finding the surrogates $v_k$ takes time $\tO(F(m,n))$ for isolating cuts and $O(n)$ for picking the surrogates. Running \Cref{alg:augment-1} takes $\tO(F(m,n))$. So the total running time is $\tO(F(m,n))$.

\bibliographystyle{plain}

\bibliography{refs}

\appendix
\section{Example of Graph with Exponentially Many Extreme Sets}\label{sec:exponential}

The example is given in \Cref{fig:exponential}. All vertex sets of the form $\{s\}\cup \{v_i: i\in S\}$ for any $S\subseteq [n-2]$ are extreme sets. This is because their cut size is $n-|S|+1$ whereas any strict subset containing terminals must be of the form $\{s\} \cup \{v_i: i\in S'\}$ for some $S'\subset S$. The latter vertex set has cut value $n-|S'|+1$, which is strictly larger than $n-|S|+1$.

\begin{figure}[t]
    \centering
    \includegraphics[width=.4\textwidth]{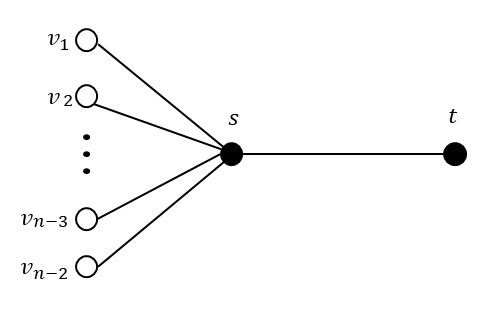}
    \caption{A graph with exponentially many extreme sets. The solid vertices $s,t$ are terminals and the empty vertices $v_1, v_2, \ldots, v_{n-2}$ are nonterminals.}
    \label{fig:exponential}
\end{figure}

\section{Augmenting Connectivity by Matching: Illustrative Example of a Cycle}

\begin{figure}
    \centering
    \begin{subfigure}[b] {\textwidth}
    \centering
    \includegraphics[width=0.6\textwidth]{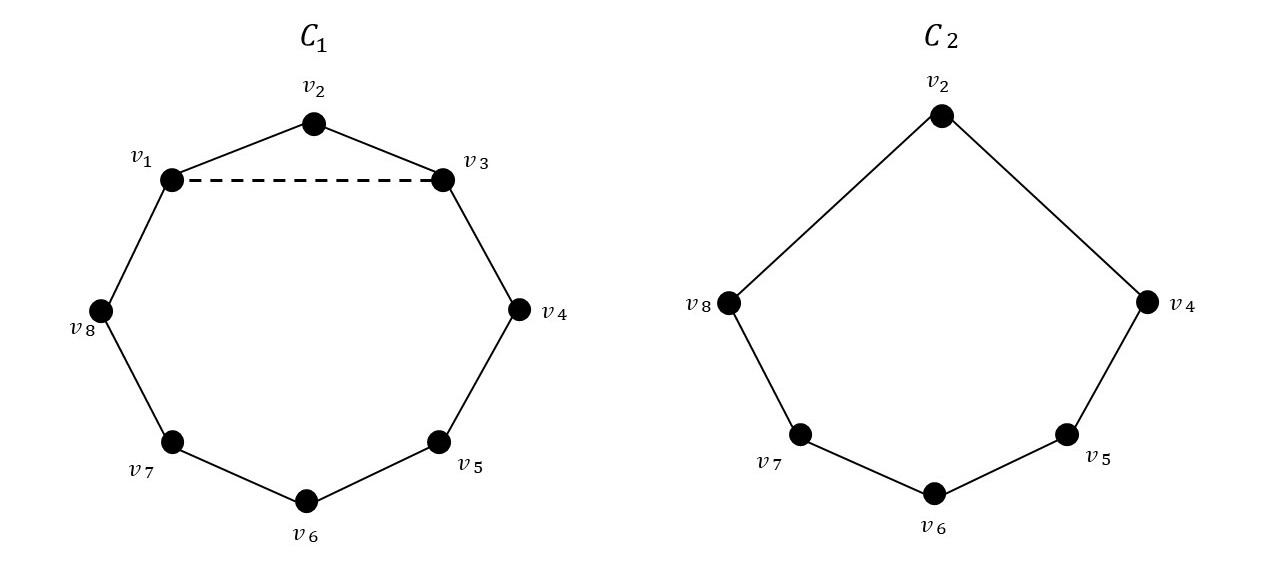}
        \caption{Here $C_1$ is augmented by one edge $v_1-v_3$. After this edge is added, augmenting $C_2$ to connectivity 3 gives an augmentation of $C_1$ to connectivity 3, so we have reduced our problem to an augmentation on a smaller graph.}
        \label{fig:augment-cycle-1}
    \end{subfigure}
    \begin{subfigure}[b]  {\textwidth}     \centering     \includegraphics[width=0.6\textwidth]{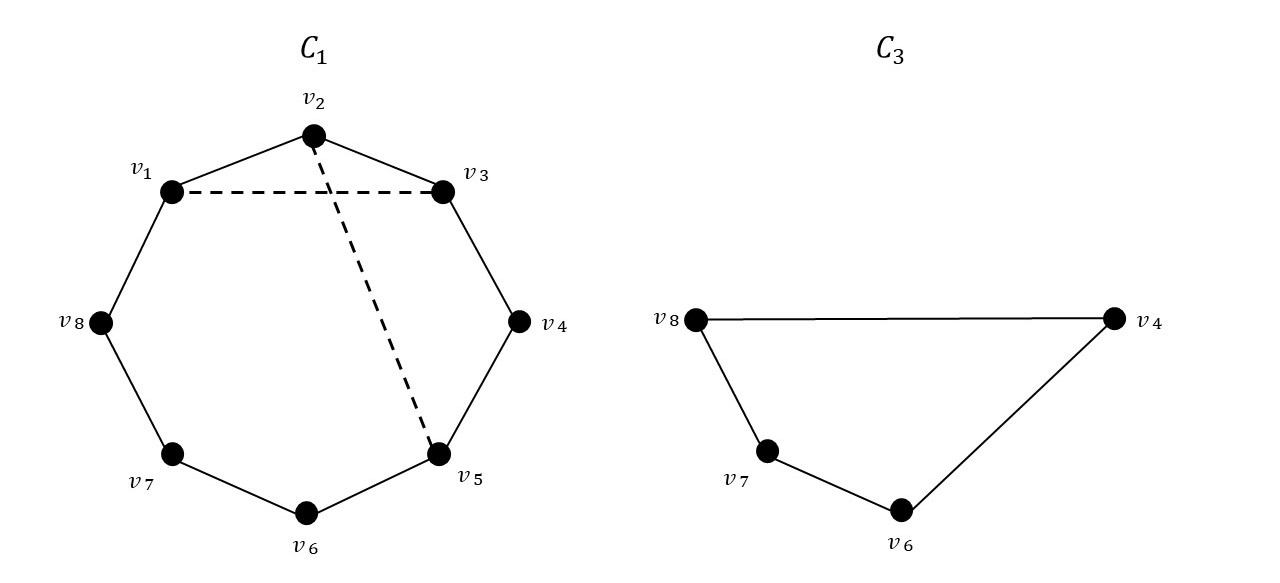}
        \caption{First add $v_1-v_3$ to $C_1$ to result in a reduction of augmenting $C_1 $ to connectivity 3 to augmenting $C_2$ as in \Cref{fig:augment-cycle-1}. Then adding the edge $v_2-v_5$ to $C_2$ reduces the augmentation of $C_1$ to connectivity 3 further to augmenting $C_3$ to connectivity 3.}
        \label{fig:augment-cycle-2}
    \end{subfigure}
    \begin{subfigure}[b]{\textwidth}\centering
    \includegraphics[width=0.6\textwidth]{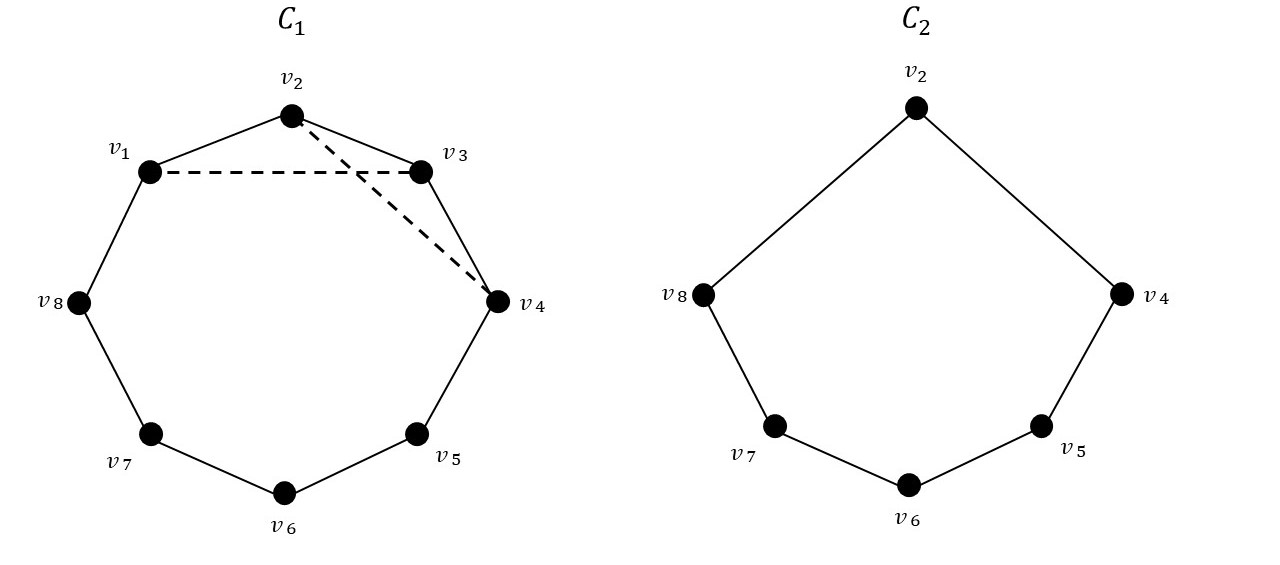}
        \caption{After adding the edge $v_1-v_3$ we reduce augmenting $C_1$ to augmenting $C_2$. But now, the edge $v_2-v_4$ does not result in a good reduction of augmentation of $C_2$ to a smaller problem because $v_2$ and $v_4$ are adjacent.}
        \label{fig:augment-cycle-3}
    \end{subfigure}
    \caption{Augmenting the connectivity of a cycle by 1}
    \label{fig:augment-cycle}
\end{figure}

We are interested in how to augment the connectivity of a cycle from $2$ to $3$. Our running example will be the cycle graph on 8 vertices, which we call $C_1$. \Cref{fig:augment-cycle-1} shows a case in which we add one edge to a cycle graph as a partial solution. Then we claim that after adding the edge $(v_1,v_3)$, in order to augment the connectivity of $C_1$ to 3, it suffices to augment the connectivity of $C_2$ to 3. This is because the presence of the $(v_1,v_3)$ edge allows us to shortcut the edges $(v_2,v_8)$ and $(v_2,v_4)$, since all cuts of value 2 in $C_1\cup \{(v_1,v_3)\}$ are accounted for as cuts of value 2 in $C_2$. 

The same reasoning applies to any edge $(v_i,v_j)$ between non-adjacent vertices in the cycle. However, notice that since $v_i$ and $v_{i+1}$ are adjacent in $C_1$, adding the edge $(v_i,v_{i+1})$ does not allow us to shortcut both vertices by deleting $v_i$ and $v_{i+1}$ and adding the edge $(v_{i-1},v_{i+2})$, since $\{v_i,v_{i+1}\}$ would be a cut of value 2 not accounted for in the shortcut version.

In \Cref{alg:augment-1} we actually add edges in batches of $\lceil\frac{n}{4}\rceil$ edges at a time, where $n$ is the number of vertices in the cycle, so it is important to understand which edges are compatible with each other. Consider the example in \Cref{fig:augment-cycle-2}. After adding the edge $(v_1,v_3)$ we can delete and shortcut those two vertices as discussed previously, and effectively work with $C_2$ instead. Now, the second edge given in the example is $(v_2,v_5)$, which in $C_2$ connects two non-adjacent vertices. Then, applying the same reasoning, we can delete and shortcut $v_2$ and $v_5$, and work with $C_3$ from \Cref{fig:augment-cycle-2} instead. This is because the two edges to be added were compatible with each other in the sense that there is a sequence of them such that in the sequence of deletions and shortcuts, each edge connects non-adjacent vertices in the sequence of simplified cycles.

On the other hand, \Cref{fig:augment-cycle-3} displays an example in which two edges are not compatible. That is, once we add one edge and apply deletion and the shortcut, the other edge connects adjacent vertices in the resulting graph. 

In \Cref{sec:augment-1} we show that any compatible sequence of additional edges can be completed to an optimal connectivity augmentation of any graph by 1. Then, we analyze the probability that a randomly selected partial matching is compatible in the sense that we have just described.

\section{Data Structures for Lazy Updates}
\label{sec:lazy}

In this section, we sketch the running time analysis using lazy data structures in Cen et al.\ \cite{CenLP22a} for maintaining supreme sets tree while the algorithm adds augmentation chains. That gives proofs for \Cref{thm: aug-chain-runtime} and \Cref{thm: runtime-deg-constr-aug-chain}.

\eat{\begin{lemma}[\cite{goldberg1991use}]\label{lem: data structure}
    Given a tree $T$ and a function $c:V(T)\to \mathbb{R}$, there is a data structure such that the following operations on $T$ can all be done in $O(\log n)$ time, where $n=|V(T)|$.
    \begin{enumerate}
        \item \emph{AddPath}($a,b,x$): Set $c(v)=c(v)+x$ for all $v$ on the unique path in $T$ from $a$ to $b$.
        \item \emph{MinPath}($a,b$): Return the minimum value of all vertices on the $a-b$ path in the tree, and
        \item \emph{MinSubtree}($a$): Return the minimum value of all vertices in the subtree rooted at $a$. 
    \end{enumerate}
\end{lemma}}

\begin{proof}[Proof of \Cref{thm: aug-chain-runtime}]
    We begin with an initial augmentation chain $\{(a_i,b_{i+1})\}_{1\leq i\leq r-1}$ with $a_i,b_i\in R_i$, where each $R_i$ is a root in the supreme sets forest with demand at least 2. This chain can be calculated in $O(n)$ time. We add it until it is no longer an augmentation chain, in which case we say it expires.
    This happens when one of the following occurs:
    \begin{enumerate}
        \item\label{out of vacant degree} Some vertex $a_i$ uses up all of its vacant degree.
        \item\label{demand satisfied} For some $i\le r$, $\text{dem}(R_i)\leq 1$.
        \item\label{no longer supreme} For some $i$ and $W\in \text{ch}(R_i)$, $d(W)\geq d(R_i)$.
    \end{enumerate}
    At this point, we replace $\{(a_i,b_{i+1})\}_{1\leq i\leq r-1}$ with a new augmentation chain $\{(a_i',b_{i+1}')\}_{1\leq i\leq r-1}$ and keep using the new chain until it expires again. Each of the three cases 
    can occur at most $n$ times because each vertex can use up degree once, and each node in the supreme sets forest can be deleted once. 

    Given that an expiration occurred, we first show how to select a new augmentation chain. If expiration occurred because of Case \ref{out of vacant degree}, we replace the terminal $a_i$ with no more vacant degree by another vacant vertex in the same supreme set. This takes constant time. If Case \ref{demand satisfied} occurs so that some $X_i$ has demand at most 1, then we remove the edges $(a_i,b_{i+1})$ and $(a_{i-1},b_i)$ from the augmentation chain and replace them by an edge $(a_{i-1},b_{i+1})$. This also takes constant time. If Case \ref{no longer supreme} occurs for some $R_i$, then we replace $(a_i,b_{i+1})$ and $(a_{i-1},b_i)$ by a partial augmentation chain connecting all children of $R_i$. The running time is proportional to the number of children. Since each node is added at most once, the total cost of Case \ref{no longer supreme} is $O(n)$.
    
    When some edge expires and gets updated, other chain edges remain and contribute to the cut values.
    So instead of maintaining the exact cut values of the supreme sets, we maintain them lazily by recording the linear growth (in the number of times an augmentation chain is added) of cut values. The expiration time does not change as long as the incident chain edges remain the same. When an edge of the augmentation chain is modified, we calculate the new expiration time of the chain as follows in time $O(\log(n))$. Each update of chain edges costs constant data structure operations, and there are $O(n)$ chain updates in total, so the total running time is $O(n\log n)$. 

    Now we show how to determine the expiration time of a new augmentation chain given by modifying one edge of an old augmentation chain. This is easy given the correct cut values for all supreme terminal sets, because we have that the expiration time for the chain $F$ is $t_F=\min(t_1,t_2,t_3)$, where
    \begin{align*}
        t_1&=\min_{v\in V}\lfloor{\frac{\beta(v)}{d_F(v)}}\rfloor,\\
        t_2&=\min_{i\in [r]}\lfloor \frac{\tau-d(X_i)}{d_F(X_i)}  \rfloor,\\
        t_3&=\min_{i\in[r]}\min_{W\in\text{ch}(X_i)}\lceil \frac{d(W)-d(X_i)}{d_F(X_i)-d_F(W)} \rceil.
    \end{align*}
    The values $t_1, t_2, t_3$ can be maintained by priority queues, so finding these minimums can be done in $O(\log(n))$ time. 
    
    Finally, we need to maintain the expiration time under an update of augmentation chain. 
    When adding a new chain edge, we first carry out previous lazy updates of relevant nodes to the current time, and then apply new lazy updates. The updates are uniformly adding a constant to the vacant degree and cut values of nodes on a tree path, so they can be handled in constant dynamic tree operations.
    (For each edge $(a_i,b_{i+1})$ in the chain $F$, the supreme sets that have their corresponding cut values increased are those on the forest path in $L$ from $a_i$ to $b_{i+1}$ excluding the least common ancestor of $a_i$ and $b_{i+1}$.)
\end{proof}

\begin{proof}[Proof of \Cref{thm: runtime-deg-constr-aug-chain}]
We show that the degree-constrained augmentation algorithm behaves in the same way as the unconstrained augmentation algorithm in terms of maintaining the supreme sets tree while greedily adding augmentation chains, despite several differences in algorithm details. It follows that we can use the same lazy data structures as in the unconstrained setting to get $\tO(n)$ running time.

We summarize the differences and remedies below.
\begin{enumerate}
    \item The algorithm works on the forest $L_{high}$ instead of $L$. So the data structure maintains the part of supreme sets forest corresponding to $L_{high}$.
    \item Because the chain edges incident on nonterminals and the supreme sets will change during the algorithm, we do not know whether an edge contributes to a supreme sets in degree-constrained setting. However, \Cref{lem:deg-vacant-in-current-supreme} shows that we can locate each vacant degree in a current supreme set projected to a node in $L_2$, so by laminarity of current supreme sets, we know the contribution of an edge to all supreme sets of projection in $L_{high}$. 
    \item The algorithm needs to remove a forest node when it is no longer the terminal projection of a supreme set. By \Cref{lem:deg-current-supreme-defined} and \Cref{lem:deg-delete-non-extreme}, this is detected by the algorithm. Moreover, the operation is the same as in unconstrained setting: by comparing whether $c(R)\ge c(W)$ for any node $R$ and its child $W$. So the same data structures can be applied to maintain the forest structure in degree-constrained setting.
\end{enumerate}
\end{proof}

\end{document}